\documentclass[letterpaper,11pt,anonymous]{article}
\usepackage[margin=1in]{geometry}
\usepackage{multirow}
\usepackage{amsthm}
\usepackage{amsmath}
\usepackage{hyperref}
\usepackage{thmtools}
\usepackage{pgf,tikz}
\usepackage{paralist, tabularx} 
\usepackage{makecell}
\usepackage{algorithmic}
\usepackage[Algorithm,ruled]{algorithm}
\usepackage{verbatim}

\usetikzlibrary{arrows}

\declaretheorem{theorem}
\declaretheoremstyle[%
  spaceabove=-6pt,%
  spacebelow=6pt,%
  headfont=\normalfont\itshape,%
  postheadspace=1em,%
  qed=\qedsymbol%
]{mystyle}

\def\ve#1{\mathchoice{\mbox{\boldmath$\displaystyle\bf#1$}}
	{\mbox{\boldmath$\textstyle\bf#1$}}
	{\mbox{\boldmath$\scriptstyle\bf#1$}}
	{\mbox{\boldmath$\scriptscriptstyle\bf#1$}}}

\newcommand{\Z}{\ensuremath{\mathbb{Z}}}

\makeatletter
\let\bfseries=\undefined
\DeclareRobustCommand\bfseries
{\not@math@alphabet\bfseries\mathbf
	\boldmath\fontseries\bfdefault\selectfont}
\makeatother

\def\Orthant_j{{\mathcal O}_{j}}

\newcommand\vea{{\ve a}}

\newcommand\veb{{\ve b}}
\newcommand\vecc{{\ve c}}
\newcommand\ved{{\ve d}}

\newcommand\vep{{\ve p}}

\newcommand\vet{{\ve t}}

\newcommand\vex{{\ve x}}
\newcommand\vey{{\ve y}}
\newcommand\vez{{\ve z}}

\newcommand{\OO}{{\mathcal{O}}}

\newtheorem{claim}{Claim}
\newtheorem{observation}{Observation}
\newtheorem{lemma}{Lemma}

\newtheorem{corollary}{Corollary}

\newtheorem*{T1}{Lemma~\ref{lemma:al2}}
\newtheorem*{T2}{Lemma~\ref{lemma:uniquesum-1}}
\newtheorem*{T3}{Lemma~\ref{lemma:maxsat-eth2}}
\newtheorem*{T4}{Lemma~\ref{lemma:linked-sum}}
\newtheorem*{T5}{Claim~\ref{claim:1}}
\newtheorem*{T6}{Lemma~\ref{lemma:shuffle-1}}

\newcommand{\head}[1]
 {\markright{\hbox to 0pt{\vtop to 0pt{\hbox{}\vskip 3mm \hrule
 width  \textwidth \vss} \hss}{\sc #1}}}
\usepackage{latexsym,graphicx,amsfonts,amsmath}
\begin{document}


\title{\bf Tight running times for minimum $\ell_q$-norm load balancing: beyond exponential dependencies on $1/\epsilon$ }

\author{
Lin Chen$^1$\footnote{Research of Lin Chen was partly supported by NSF Grant 1756014.}\ \ \  Liangde Tao$^2$\ \ Jos\'e Verschae$^3$ \\ 
\\{\small $^1$Department of Computer Science, Texas Tech University, US}
\\{\small chenlin198662@gmail.com}
\\{\small $^2$Department of Computer Science, Zhejiang University, China}
\\{\small vast.tld@gmail.com}
\\{\small $^3$Institute for Mathematical and Computational Engineering},\\ {\small Faculty of Mathematics and School of Engineering, Pontificia Universidad Católica de Chile, Chile}
\\{\small jverschae@uc.cl}
}

\date{}
\maketitle

\begin{abstract}
We consider a classical scheduling problem on $m$ identical machines. For an arbitrary constant $q>1$, the aim is to assign jobs to machines such that $\sum_{i=1}^m C_i^q$ is minimized, where $C_i$ is the total processing time of jobs assigned to machine $i$. It is well known that this problem is strongly NP-hard.

Under mild assumptions, the running time of an $(1+\epsilon)$-approximation algorithm for a strongly NP-hard problem cannot be polynomial on $1/\epsilon$, unless $\text{P}=\text{NP}$. For most problems in the literature, this translates into algorithms with running time at least as large as $2^{\Omega(1/\varepsilon)}+n^{O(1)}$. For the natural scheduling problem above, we establish the existence of an algorithm which violates this threshold. More precisely, we design a PTAS that runs in $2^{\tilde{O}(\sqrt{1/\epsilon})}+n^{O(1)}$ time. This result is in sharp contrast to the closely related minimum makespan variant, where an exponential lower bound is known under the exponential time hypothesis (ETH). We complement our result with an essentially matching lower bound on the running time, showing that our algorithm is best-possible under ETH. The lower bound proof exploits new number-theoretical constructions for variants of progression-free sets, which might be of independent interest. 

Furthermore, we provide a fine-grained characterization on the running time of a PTAS for this problem depending on the relation between $\epsilon$ and the number of machines $m$. More precisely, our lower bound only holds when $m=\Theta(\sqrt{1/\epsilon})$. Better algorithms, that go beyond the lower bound, exist for other values of $m$. In particular, there even exists an algorithm with running time polynomial in $1/\epsilon$ if we restrict ourselves to instances with $m=\Omega(1/\epsilon\log^21/\epsilon)$.

\bigskip
\smallskip\noindent{\bf Keywords:} {Polynomial Time Approximation Scheme, Scheduling, Exponential Time Hypothesis.}
\end{abstract}
\thispagestyle{empty}
\newpage
\setcounter{page}{1}

\section{Introduction}

We consider a classical scheduling problem on identical parallel machines. Suppose we are given $m$ identical machines and $n$ jobs, each having a processing time $p_j$. A feasible solution corresponds to an assignment of jobs to machines. For a given assignment, let $C_i$ be the total processing time of jobs assigned to machine $i$, that is, $C_i=\sum_{j\rightarrow i}p_j$. Our objective is to minimize $\sum_{i=1}^m C_i^q$, where $q>1$ is an arbitrary constant. For either exact algorithms or approximation schemes, minimizing $\sum_{i=1}^m C_i^q$ is equivalent to minimizing the $\ell_q$-norm of machine loads, i.e., $(\sum_{i=1}^m C_i^q)^{1/q}$. In the standard 3-field scheduling notation by Graham et al.~\cite{graham1979optimization}, this problem is denoted as $P||\sum_iC_i^q$.

Our problem is well-known to be strongly NP-hard by a simple reduction from 3-partition. On the other hand, a classic result by Alon et al.~\cite{alon1997approximation} shows that it admits a polynomial time approximation scheme (PTAS) with running time $f(1/\epsilon)+n^{O(1)}$, where $f(1/\epsilon)$ is doubly exponential in $1/\epsilon$. 
Very recently, improved running times have been obtained for $P||\sum_iC_i^q$ and other closely related load balancing problems. Particularly, for a variety of objective functions, which include both $\sum_iC_i^q$ and the makespan objective $C_{\max}=\max_i C_i$, Jansen et al.~\cite{jansen2020closing} show that the problem admits a PTAS with a running time of $2^{\tilde{O}(1/\epsilon)}+\tilde{O}(n)$. On the negative side, for the makespan objective, Chen et al.~\cite{chen2014optimality} show that such a running time is essentially best possible under the exponential time hypothesis (ETH). However, the lower bound does not hold for other objectives, including $P||\sum_iC_i^q$, leaving open the possibility for improved running times. In this paper, we study this question and explore the surprisingly rich complexity landscape of $P||\sum_iC_i^q$ in the context of approximation schemes.

\smallskip
\noindent\textbf{Contribution Overview.}
We study the complexity landscape of approximation schemes for $P||\sum_iC_i^q$.  
Consider some strongly NP-hard optimization problem whose optimal value $\text{OPT}(I)$ is integral and upper bounded by $\text{poly}(|I|_u)$ for any instance $I$, where $|I|_u$ is the input size written in unary. This implies that the problem does not admit a fully polynomial-time approximation scheme (FPTAS) unless P=NP~\cite{garey2002computers}. In the majority of cases, for such problems the literature presents PTASs with running time at least as large as $2^{\Omega(1/\epsilon)}+n^{O(1)}$, that is, the dependency on $1/\epsilon$ is exponential. We show that $P||\sum_{i}C_i^q$ does not fall into this case, and a running time subexponential on $1/\epsilon$ is achievable. More precisely, we give a PTAS with a running time of $2^{\tilde{O}(\sqrt{1/\epsilon})}+n^{O(1)}$. On the other hand, we show that this running time is essentially tight, by providing an almost matching lower bound under ETH. That is, we show that ETH rules out a PTAS of running time $2^{O(({1/\epsilon})^{1/2-\delta})}+n^{O(1)}$ for any $\delta>0$. We are not aware of any other PTAS for a strongly NP-hard problem with such a tight subexponential behavior on $1/\varepsilon$.

Besides the results above, we give a fine-grained study on the upper and lower bounds of the running time of a PTAS for $P||\sum_{i}C_i^q$. First of all, we notice that our lower bound only holds for a small range of values of $m$, depending on $\epsilon$. Moreover, for some other values, we can circumvent the lower bound and obtain improved running times. More precisely, the lower bound only holds when $m=\Theta(\sqrt{1/\epsilon})$.
 Quite surprisingly, when $m$ is larger, namely $m=\Omega(1/\epsilon\log^{2}(1/\epsilon))$, an algorithm that runs polynomially in $1/\epsilon$ exists, despite the problem being strongly NP-hard in general and our stronger lower bound. If $m = O(\sqrt{1/\epsilon})$ we can use a PTAS with running time $(1/\epsilon)
^{O(m)}$, which also breaks the lower bound for $m=o(\sqrt{1/\epsilon})$. See Figure~\ref{fig:alg_complex} for a depiction of our results. It remains an open problem to obtain tight running times when  $m=\Theta({(1/\epsilon)^\theta})$ for $\theta\in (1/2,1]$.

\begin{figure}
	\centering
	\includegraphics[height=52mm,width=99.08mm]{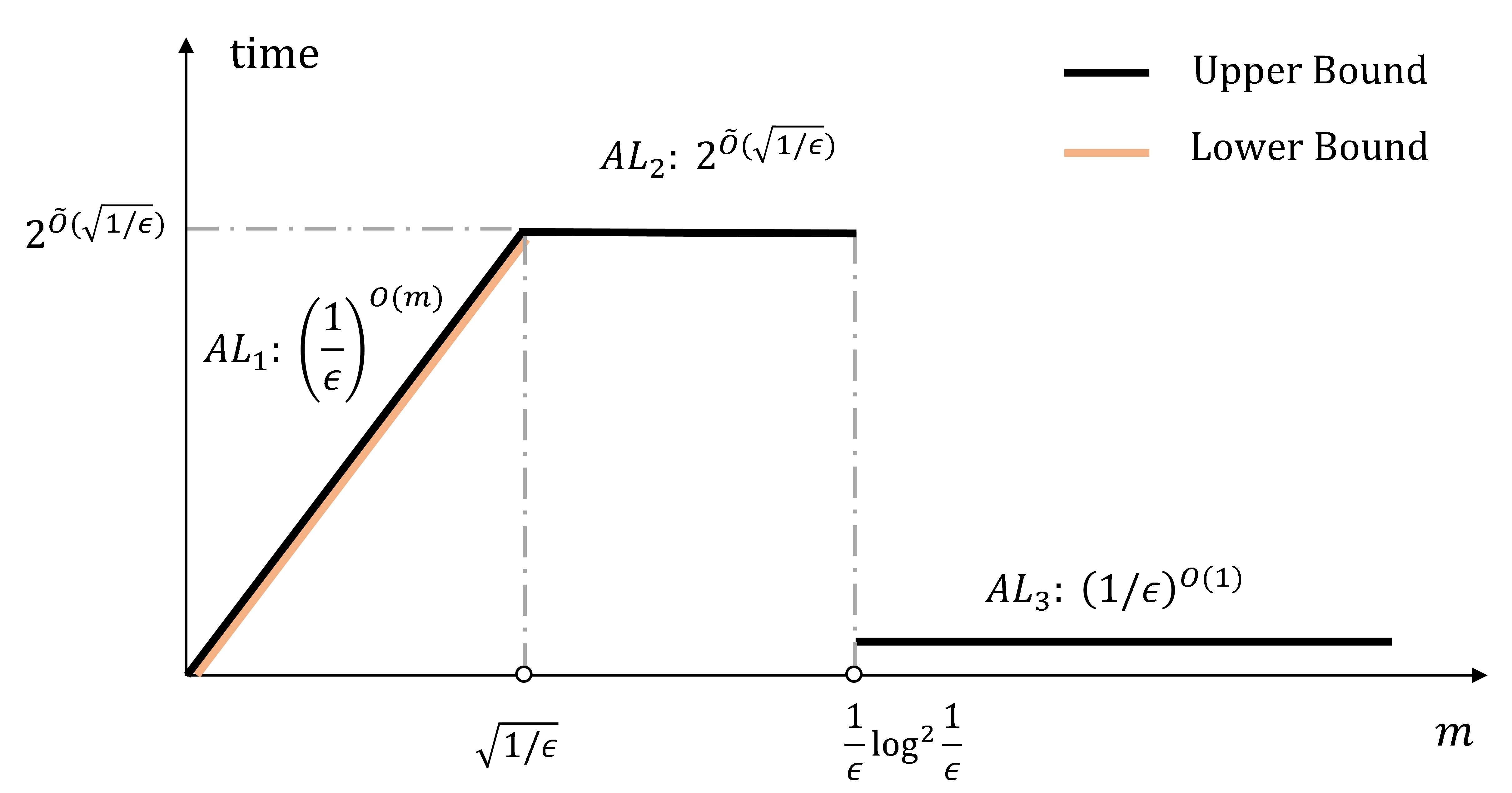}
	\caption{Complexity landscape of $P||\sum_iC_i^q$. The time axis specifies the dependency of the running time with respect to $1/\epsilon$. A term of $n^{O(1)}$ needs to be added in the running of each algorithm.}
	\label{fig:alg_complex}
\end{figure}

\smallskip
\noindent \textbf{Technical Contribution.}
Our main technical contribution lies in the lower bound proof. For this, we give a fine-grained reduction from a variant of Max3SAT to $P||\sum_i C_i^q$. To do so, we convert a set of clauses to a set of jobs. We enforce that two jobs which represent variables in the same clause are scheduled together in some carefully constructed gap (i.e., slot) of a given size. For such a construction, it is imperative to use pairs of numbers with unique sums, to guarantee that only these two jobs fit this gap. Hence, our construction is tightly related to {\it Sidon sets} and {\it Salem–Spencer sets} (also called progression-free sets), both of which have been studied extensively in number theory (see, e.g., \cite{erdos1941problem,o2004complete,moser1953non,gasarch2008finding}). A Sidon set $S=\{s_1,s_2,\ldots,s_n\}$ is a subset of natural numbers where all pairwise sums $s_i+s_j$, for $i\le j$, are distinct. That is,  $s_{i}+s_j=s_{i'}+s_{j'}$ implies $\{i,j\}=\{i',j'\}$. A weaker notion is that of a Salem-Spencer set, that is, a set $S=\{s_1,s_2,\cdots,s_n\}$ with no cardinality 3 progression, i.e., no triplet $(i,j,k)\in\Z_n:=\{1,2,\cdots,n\}$ of pairwise different numbers satisfies $s_i-s_j = s_k-s_i$. In other words, if $s_j+s_k=2s_i$ then $i=j=k$.
Our lower bound could be proved by adapting known techniques if a Sidon set $S\subseteq \Z_N$ (where $\Z_N:=\{1,2,\cdots ,N\}$) with cardinality $n$ exists for $N=n^{1+o(1)}$. Unfortunately, this is impossible, as Erdös and Turán~\cite{erdos1941problem} show that the cardinality of a Sidon set with $n$ elements requires $N=\Omega(n^2)$. 
We can circumvent this negative result by requiring only some pairs of numbers to have a unique sum, where these pairs correspond to the clauses in the given Max3SAT instance. Towards this, we first transform the given Max3SAT instance, with variables $z_{j}$ for $j\in \mathbb{Z}_n$, into a special structure such that all clauses can be divided into two disjoint subsets $C_1$ and $C_2$: $C_1$ consists of clauses $cl_2,cl_5,\cdots,cl_{n-1}$ such that $cl_\ell=(w_{\ell-1}\vee w_{\ell}\vee w_{\ell+1})$, where $w_{j}\in\{z_j,\neg z_{j}\}$ for all $j$; and $C_2$ consists of clauses $cl_1',cl_2',\cdots,cl_n'$ such that $cl_\ell'=(z_\ell\oplus \neg z_{\tau(\ell)})$, where $\tau$ is a permutation of $\Z_n$ and $\oplus$ is the XOR operation (see Section~\ref{subsec:maxsat} for details). For $C_1$, we construct a set of numbers $\{\sigma(1),\sigma(2),\cdots,\sigma(n)\}$ such that every adjacent sum $\sigma(i)+\sigma(i+1)$ is unique, and this will be achieved through extending a known construction of Salem-Spencer sets (Lemma~\ref{lemma:uniquesum-2}). For $C_2$, we extend the construction to additionally require that the $\sigma(i)$'s we construct admit a \textit{linked unique sum}. That is, there exists a 
subset of numbers $E_i=\{e_{i,1},e_{i,2},\ldots,e_{i,\omega}\}\subseteq \Z_{n^{1+o(1)}}$ for every $i\in\Z_n$ such that $E_i\cap E_{i'}=\emptyset$ for any $i'\neq i$, and the sum of each pair $\sigma(i)+e_{i,1},e_{i,1}+e_{i,2},\cdots,e_{i,\omega-1}+e_{i,\omega},e_{i,\omega}+\sigma(\tau(i))$ is unique in the sense that no other pairs in $S\cup E$ sum up to the same value, where $E=\bigcup_{i=1}^n E_i$. Note that a linked unique sum is a weaker notion than Sidon or Salem-Spencer, as for these there is no auxiliary set $E$. Nevertheless, the property of linked unique sum is strong enough for our reduction. The construction of the auxiliary set $E$ relies on further extending our technique for constructing unique adjacent sums, together with a group theoretic lemma that allows an \lq\lq orthogonal\rq\rq\, decomposition of the permutation $\tau$ (Lemma~\ref{lemma:shuffle-1}). Our results may be of separate interest for constructing fine-grained lower bounds on approximation or parameterized algorithms for other problems.

Another crucial observation, which may also be of independent interest, is a structural result needed for our PTAS with running time $2^{\tilde{O}(\sqrt{1/\epsilon})}+n^{O(1)}$ (see Section~\ref{subsec:alg2}). For many objective functions (like $C_{\max}$) we can round the processing times to powers of $1+\epsilon$ in order to bound the overall loss by a factor of $1+O(\epsilon)$. We observe that for minimizing $\sum_i C_i^q$ it is possible to consider a coarser grouping of jobs into sizes within a $(1+\sqrt{\epsilon})$ factor. Broadly speaking, by imposing extra structure to a near-optimal solution, we can use a Taylor expansion to bound the error, and notice that the linear term of the polynomial expansion cancels out. This leaves us only with the quadratic (and lower order) terms. This observation might translate to other problems with $\ell_q$-norm objective, and even other min-sum cost functions.

\smallskip
\noindent\textbf{Related Work.}
Load balancing problems are fundamental in computer science and have been studied extensively in the literature. In particular, the first PTAS for $P||C_{\max}$ dates back to the 80's~\cite{hochbaum1987using} and there is a long history of improvements on the running time for various identical machine scheduling problems, including  $P||C_{\max}$, $P||\sum_i C_i^q$, $P||\sum_j w_jC_j$, etc.; see, e.g., \cite{leung1989bin,alon1998approximation,hochbaum1997various, skutella1999ptas,jansen2010eptas,jansen2020closing}. Recently, more general objective functions based on arbitrary norms have been considered~\cite{ibrahimpur_minimum_norm_2021}. Parameterized algorithms for scheduling problems have also been studied extensively (see, e.g.~\cite{jansen2020structural,mnich2018parameterized,mnich2015scheduling,knop2018scheduling,chen2017parameterized}). 

The exponential time hypothesis (ETH) is a widely accepted complexity assumption introduced by Impagliazzo et al.~\cite{impagliazzo2001problems,impagliazzo2001complexity}, which can be used to obtain lower bounds on the running time of algorithms for various problems (see, e.g., \cite{lokshtanov2013lower} for a survey). In 2014, Chen et al. \cite{chen2014optimality} provide a concrete lower bound on the running time of a PTAS for $P||C_{\max}$ under ETH. Later, Jansen et al.~\cite{jansen2020closing} give a PTAS with running time $2^{\tilde{O}(1/\epsilon)}+n^{O(1)}$ for $P||C_{\max}$, which almost matches the lower bound. 

Despite PTASs having been established for a variety of optimization problems, much less is known regarding lower bounds on their running time.
In addition to $P||C_{\max}$, mentioned above, other well-known examples include multiple knapsack \cite{jansen2016bounding}, planar vertex cover, planar dominating set, and planar traveling salesperson~\cite{marx2007optimality}. Interestingly, all of these lower bounds have an almost linear dependency on $1/\epsilon$ in the exponent, which essentially matches the best-known PTAS. Generally, Chen et al.~\cite{chen2004linear} proved that if the problem, parameterized by $1/\epsilon$, is W[1]-hard under a linear FPT reduction, then there is no PTAS with $f(1/\epsilon)|I|^{o(1/\epsilon)}$ running time for an arbitrary computable function $f$, assuming all problems in SNP cannot be solved in sub-exponential time. We are not aware of a PTAS whose running time is subexponential in $1/\epsilon$, either for scheduling or other strongly NP-hard problems.

Unlike approximation algorithms, subexponential running times on a parameter have been observed in the field of parameterized algorithms and have received significant attention. In particular, a variety of optimization problems in planar graphs admit a fixed parameter tractable (FPT) algorithm that is subexponential in the parameter, including, e.g., independent set~\cite{demaine2005subexponential}, dominating set~\cite{demaine2005subexponential}, and multiway cut~\cite{klein2012solving,marx2012tight,pilipczuk2018network}. Note that, on the other hand, a subexponential PTAS was ruled out for the planar dominating set problem~\cite{marx2007optimality}.

\section{Approximation schemes}\label{sec:upper}
The goal of this section is to prove the following theorem.
\begin{theorem}\label{thm:algorithm}
	For any sufficiently small $\epsilon>0$, there exists an algorithm that outputs a $(1+\epsilon)$-approximate solution for the scheduling problem $P||\sum_iC_i^q$ within $2^{\tilde{O}(\sqrt{1/\epsilon})}+n^{O(1)}$ time. More specifically, there exists a:
	\begin{compactitem}
	\item $(1+\epsilon)$-approximation algorithm AL$_1$ that runs in time $(1/\epsilon)^{O(m)}+n^{O(1)}$ for $m=O(\sqrt{1/\epsilon})$;
	\item $(1+\epsilon)$-approximation algorithm AL$_2$ that runs in time $2^{\tilde{O}(\sqrt{1/\epsilon})}+n^{O(1)}$ for $m=(1/\epsilon)^{O(1)}$;
	\item $(1+\epsilon)$-approximation algorithm AL$_3$ that runs in time $(1/\epsilon)^{O(1)}+n^{O(1)}$ for $m=\Omega(1/\epsilon\log^2(1/\epsilon))$.
	\end{compactitem}
\end{theorem}

In particular, for a sufficiently small $\epsilon$, we may run AL$_1$ for $m\le \sqrt{1/\epsilon}$, run AL$_2$ for $\sqrt{1/\epsilon}<m\le 1/\epsilon^2$, and run AL$_3$ for $m\ge 1/\epsilon^2$. This guarantees a $2^{\tilde{O}(\sqrt{1/\epsilon})}+n^{O(1)}$ time algorithm for all values of $m$ and $1/\epsilon$.

We remark that standard techniques round the processing time of a job to some multiple of $1+\epsilon$, yielding an instance with $\tilde{O}(1/\epsilon)$ different types of jobs. However, such rounded instance cannot be solved to optimality in time $2^{(1/\epsilon)^{1-\delta}}+n^{O(1)}$ for any constant $\delta>0$~\cite{chen2014optimality}. 
Hence, we need a new approach for Theorem~\ref{thm:algorithm}.

We now give a brief overview of the proof of Theorem~\ref{thm:algorithm}. Algorithm~AL$_1$ is based on a standard dynamic programming, given in Appendix~\ref{appsubsec:al1}. Algorithm~AL$_2$ is based on an new observation (Lemma~\ref{thm:structure}) which shows that we can classify processing times on intervals of the form $[\epsilon(1+\sqrt{\epsilon})^{h-1},\epsilon(1+\sqrt{\epsilon})^{h})$ for an integer $h$. After preprocessing the instance (Lemma~\ref{lemma:prepocessing}), we can focus on only $\tilde{O}(1/\sqrt{\epsilon})$ such intervals. We show that there exists a near-optimal solution where jobs are scheduled in an ordered way following the mentioned classification. This algorithm is described in Section~\ref{subsec:alg2}. Algorithm~AL$_3$ (see Appendix~\ref{appsubsec:al3}) is based on modifying the famous algorithm for the bin packing problem by Karmarkar and Karp~\cite{karmarkar1982efficient}.

All the three algorithms will operate on a scheduling instance that is well-structured, as implied by the following lemma. The structure can be achieved through standard techniques, namely scaling and grouping of small jobs, see, e.g., \cite{alon1998approximation}. For an instance $I$, we denote by $\text{size}(I)$ the total processing time of jobs in $I$, and by $m(I)$ the number of machines.

\begin{lemma}[Alon et al.~\cite{alon1998approximation}]  \label{lemma:prepocessing}
For any sufficiently small $\epsilon>0$, given an arbitrary instance $I_0$ of $P||\sum_i C_i^q$, we can transform in linear time $I_0$ into a well-structured rounded instance $I$ with less or equal number of jobs and less or equal number of machines, that satisfies:
\begin{compactitem}
	\item $\text{size}(I)=m(I)$;
	\item the processing time of each job in $I$ belongs to $[\epsilon,1]$;
	\item there exists an optimal solution for $I$ such that the load of each machine belongs to $[1/2,2]$.
\end{compactitem}
Furthermore, any $(1+\epsilon)$-approximation solution for $I$ can be transformed into an $(1+O(\epsilon))$-approximation solution for $I_0$ in linear time.
\end{lemma}

In the following, we focus exclusively on the instance after the preprocessing. 
It is worth mentioning that for non-integral values of $q$, the objective function can be irrational even for rational processing times. For obtaining a 
PTAS this is however not a problem, as computing the objective function up to an additive error of $\epsilon/\text{poly}(n)$ suffices for our results. In what follows we omit this technicality, and assume that we can compute the objective function without error.

\subsection{Algorithm 2}\label{subsec:alg2}

In this subsection, we describe and analyze algorithm AL$_2$.

\begin{lemma}\label{lemma:al2}
	Consider an instance after the preprocessing of Lemma~\ref{lemma:prepocessing}. For any $\epsilon>0$, there exists an algorithm  \textrm{AL}$_{2}$ that outputs a $(1+O(\epsilon))$-approximation solution for $P||\sum_i C_i^q$ with $m^{\tilde{O}(1/\sqrt{\epsilon})}$ running time. 
\end{lemma}

We know there exists an optimal solution $\vex^*$ where the load of each machine belongs to $[1/2,2]$. Let $L_i^*$ be the load of machine $i$ in $\vex^*$ where $1/2\le L_i^*\le 2$. Without loss of generality we further assume that $ L_1^*\le L_2^*\le \cdots\le L_m^*$. For some integer $h\ge1$, let $\mathcal{G}_h$ be the set of jobs whose processing time lies in $[\epsilon(1+\sqrt{\epsilon})^{h-1},\epsilon(1+\sqrt{\epsilon})^h)$. Given that $\epsilon \le p_j\le 1$, every job belongs to some set $\mathcal{G}_h$ for $h\in\{1,\ldots,\tau\}$, where $\tau = \tilde{O}(1/\sqrt{\epsilon})$. For simplicity, we call a job in $\mathcal{G}_h$ a $\mathcal{G}_h$-job. 
The following structural result contains the key observation for the existence of a PTAS with subexponential time.
 
 \begin{lemma}\label{thm:structure}
 There exists a feasible solution $\hat{x}$ satisfying: i) its objective value is at most $(1+O(\epsilon)) OPT$, and ii) the machines can be ordered from $1$ to $m$ such that for any $1\le i\le m-1$ and $h$, the processing time of every $\mathcal{G}_h$-job on machine $i+1$ is at most the processing time of any $\mathcal{G}_h$-job on machine $i$.  
 \end{lemma}

\begin{proof}

Given an optimal solution $x^*$, we construct $\hat{x}$ as follows. For machine $m$, we replace all $\mathcal{G}_h$-jobs with the same number of the smallest $\mathcal{G}_h$-jobs. For machine $m-1$, we replace all $\mathcal{G}_h$-jobs with the same number of the remaining smallest $\mathcal{G}_h$-jobs, etc. Eventually, every $\mathcal{G}_h$-job on machine $i+1$ is no greater than any $\mathcal{G}_h$-job on machine $i$. Let $\hat{L}_i=L_i^*+\Delta_i$ be the new load of machine $i$. 

By the definition of $\mathcal{G}_h$, we know that the largest $\mathcal{G}_h$-job has a processing time at most $1+\sqrt{\epsilon}$ times the smallest one. This implies that $ L_i^*/(1+\sqrt{\epsilon})\le \hat{L}_i\le (1+\sqrt{\epsilon})L_i^*$, and hence $|\Delta_i|\le \sqrt{\epsilon}L_i^* \le 2\sqrt{\epsilon}$. In order to bound the objective function, first write $\sum_{i=1}^m (L_i^*+\Delta_i)^q = \sum_{i=1}^m {L_i^*}^q(1+\Delta_i/L_i^*)^q$. Using a Taylor expansion of order 1 on the function $(1+x)^q$ around $x=0$, we obtain that for some $0\le \xi_i\le \Delta_i/L_i^*\le 1$,
\begin{eqnarray*}
	\sum_{i=1}^m (L_i^*+\Delta_i)^q &&= \sum_{i=1}^m {L_i^*}^q\left(1+q\frac{\Delta_i}{L_i^*}+\frac{q(q-1)}{2}(1+\xi_i)^{q-2}\left(\frac{\Delta_i}{L_i^*}\right)^2\right)\\
	&&\le (1+O(\epsilon))\sum_{i=1}^m {L_i^*}^q+ q\sum_{i=1}^m \Delta_i {L_i^*}^{q-1}\\
	&&= (1+O(\epsilon))OPT+ q{L_1^*}^{q-1}\sum_{k=1}^m \Delta_k+q\sum_{i=2}^m\left[({L_i^*}^{q-1}-{L_{i-1}^*}^{q-1})\sum_{k=i}^m \Delta_k\right]\\
	&&\le (1+O(\epsilon))OPT.
\end{eqnarray*}

The last equality uses Abel's transformation (summation by parts). The last inequality follows since, for each $i$, it holds that 
$\sum_{k=i}^m\Delta_{k}\le 0$, as the last $m-i$ machines received the smallest $\mathcal{G}_h$-jobs for each $h$. 
\end{proof}

Exploiting the ordering of the jobs and machines given by Lemma~\ref{thm:structure}, we are able to develop a dynamic programming based algorithm to prove Lemma~\ref{lemma:al2}, see Appendix~\ref{appsubsec:al2}.

\section{Lower Bound}
In this section, we will prove the following theorem.
\begin{theorem}\label{thm:lower-bound}
Let $q>1$ be an arbitrary constant. Assuming ETH, there is no PTAS for $P||\sum_i C_i^q$ that runs in $2^{O((1/\varepsilon)^{1/2-\delta})}+n^{O(1)}$ time for any constant $\delta>0$.
\end{theorem}

For the proof we give a fine-grained reduction from a variant of Max3SAT, called 3SAT$'$ (which we elaborate in the following subsection), to $P||\sum_i C_i^q$.

\subsection{3SAT\texorpdfstring{$'$}{Lg} - Max3SAT with a Special Structure}\label{subsec:maxsat}
We study a variant of 3SAT, which we call 3SAT$'$, whose instances have the following structure: There are $n$ variables $z_1,\ldots,z_n$, where $n$ is a multiple of $3$. There are $4n/3$ clauses, such that the set of
clauses can be divided into two disjoint sets $C_1$ and $C_2$ such that:
\begin{compactitem}
\item In $C_1$, every clause is a disjunction (OR operator) of three literals. For each variable $z_i$, exactly one literal in $C_1$ belongs to $\{z_i,\neg z_i\}$. 
\item In $C_2$, every clause is of the form $z_i\oplus \neg z_k$, where $\oplus$ denotes the XOR operator. Also, for every variable $z_i$, each literal $z_i$ and $\neg z_i$ appears exactly once within $C_2$.
\end{compactitem}

For example, $C_1=\{(z_1\vee \neg z_2\vee z_3)\}$ and $C_2=\{(z_1\oplus \neg z_2),(z_2\oplus \neg z_3),(z_3\oplus \neg z_1)\}$ defines a 3SAT$'$ instance for $n=3$. 
Let $cl_2,cl_3,\cdots,cl_{n-1}$ be the clauses in $C_1$. By re-indexing we can assume that $cl_\ell$ is of the form $(w_{\ell-1}\vee w_{\ell}\vee w_{\ell+1})$, where $w_{j}\in\{z_j,\neg z_{j}\}$ for all $j$. Also notice that $|C_1|=n/3$ and $|C_2|=n$. Since every literal appears exactly once in $C_2$, we define a permutation $\tau:\mathbb{Z}_n\rightarrow \mathbb{Z}_n$ (i.e. a bijection) such that $\tau(i)=k$ for each $(z_i\oplus\neg z_k)\in C_2$. 

Similarly to 3SAT, it is also difficult to distinguish instances of 3SAT$'$ where almost all clauses are satisfiable and instances where at most certain fraction of the clauses can be satisfied, as implied by the following lemma. 
See Appendix~\ref{appsec:sat} for its proof.

\begin{lemma}\label{lemma:maxsat-eth2}
Assuming ETH, there exists a constant $\beta\in(0,1)$ such that for any sufficiently small $\epsilon',\delta>0$ there is no algorithm with running time $2^{O(n^{1-\delta})}$ that  distinguishes between instances of 3SAT$\,'$ with $4n/3 $ clauses where at least $(1-\epsilon')\cdot 4n/3$ clauses are satisfiable, from instances where at most $(\beta+\epsilon')\cdot 4n/3$ clauses are satisfiable.
\end{lemma}
\subsection{Overview of the reduction}\label{subsec:overview}
We now briefly describe the structure of the constructed scheduling instance. The detailed reduction will be presented in Appendix~\ref{appsec:reduction-construction}. We remark that the high-level structure of the scheduling instance resembles the classical reduction and that of~\cite{chen2014optimality}. New technical ingredients are in job processing times, as we will elaborate in Section~\ref{subsec:technical}.

For an instance $I_{sat}$ of 3SAT$'$ with $n$ variables, we construct the following 6 kinds of jobs:

$\bullet$ Variable jobs: For each positive (or negative, resp.) literal, say,
$z_i$ (or $\neg z_i$, resp.), two pairs of variable jobs $V_{i,+,1}^{\rho}$ and
$V_{i,+,2}^{\rho}$ (or $V_{i,-,1}^{\rho}$ and $V_{i,-,2}^{\rho}$, resp.) are
constructed where $\rho\in\{T,F\}$. In total, we construct 4 jobs for each (positive or negative) literal, i.e., 8 jobs for each variable.

$\bullet$ Clause jobs: 
 For each clause $cl_\ell$ of $C_1$, one clause job $\textrm{CL}^T_\ell$ and two copies of clause job $\textrm{CL}_\ell^F$ are
constructed. 
Recall that $|C_1|=n/3$, we construct $n$ clause jobs. 

$\bullet$ Truth-assignment jobs, link jobs and dummy jobs: These three kinds of jobs will be created suitably so that the conditions below (CO1 to CO4) are satisfied.

$\bullet$ Gap jobs: Let $Q$ be a target makespan. We construct $\tilde{O}(n)$ gap jobs and the same number of machines to create gaps. Roughly speaking, every feasible schedule whose objective value is not too large will have one gap job on each machine, leaving a gap that must be filled up such that the load of the machine is exactly $Q$. We will create 4 kinds of gaps (incurred by gap jobs) satisfying the following conditions:
\smallskip
\begin{compactitem}
\item[$\bullet$ CO1.] Variable-Truth gaps. To fill
up these gaps, for any $i$ either $V_{i,+,1}^F$, $V_{i,+,2}^F$, $V_{i,-,1}^T$,
$V_{i,-,2}^T$, or $V_{i,+,1}^T$,
$V_{i,+,2}^T$, $V_{i,-,1}^F$, $V_{i,-,2}^F$ are used. Truth-assignment jobs are created for this purpose.

\item[$\bullet$ CO2.] Variable-Clause-Dummy gaps. For each clause $cl_\ell\in C_1$, there are three variable-clause-dummy gaps. If the
positive (or negative, resp.) literal $z_i$ (or $\neg z_i$, resp.) is in $cl_\ell\in
C_1$, then a variable-clause-dummy gap is created so that it could only be
filled up by $\textrm{CL}_\ell^\rho$ and $V_{i,+,1}^{\rho'}$ (or $\textrm{CL}_\ell^\rho$ and $V_{i,-,1}^{\rho'}$, resp.), where $\rho,\rho'\in\{T,F\}$, together with a dummy job.  Further, the gap ensures that $\textrm{CL}^T_\ell$ has to be scheduled with 
either $V_{i,+,1}^T$ or $V_{i,-,1}^T$. 

\item[$\bullet$ CO3.] Variable-Link and Link-Link
gaps. For each clause $(z_i\oplus\neg z_k)\in C_2$ we create a collection of Variable-Link and Link-Link gaps. To fill up these gaps, either $V_{i,+,2}^T$ and $V_{k,-,2}^F$, or $V_{i,+,2}^F$ and $V_{k,-,2}^T$ are used. Link jobs are created for this purpose (see Section~\ref{subsec:technical} for more details on this construction).

\item[$\bullet$ CO4.] Variable-Dummy gaps. Recall that 8 variable jobs
are constructed for a variable and only 7 of them are used for the 4 kinds of gaps above (either
$V_{i,+,1}^\rho$ or $V_{i,-,1}^\rho$ is left, where $\rho\in\{T,F\}$), the remaining one together with a dummy job will be used to
fill these gaps. 
\end{compactitem}
\smallskip

With this construction, it is not difficult to verify that if every gap is filled exactly, $I_{sat}$ is satisfiable. To
see why, if $V_{i,+,1}^F$, $V_{i,+,2}^F$, $V_{i,-,1}^T$,
$V_{i,-,2}^T$ are used in the variable-truth gaps, then we let variable $z_i$ be true, otherwise
we let it be false. For any clause of $C_1$, say, $cl_\ell$, there is one $\textrm{CL}_\ell^T$ and it must
be scheduled with a true variable job, say, $V_{i,+,1}^T$ if $z_i$ is a literal in $cl_{\ell}$ (or $V_{i,-,1}^T$ if $\neg z_i$ is a literal in $cl_{\ell}$). If $V_{i,+,1}^T$ (or $V_{i,-,1}^T$, resp.) is scheduled
with $\textrm{CL}_\ell^T$, then the positive (or negative, resp.) literal $z_i$ (or $\neg z_i$, resp.) is in $cl_\ell$. Meanwhile the variable $z_i$ is true (or false, resp.) since
otherwise $V_{i,+,1}^T$ (or $v_{i,-,1}^T$,resp.) are used to fill variable-truth gaps. Thus clause $cl_\ell$ is satisfied. For any clause of $C_2$, say, $(z_i\oplus\neg z_k)$,
if $V_{i,+,2}^T$ and $V_{k,-,2}^F$ (or $V_{i,+,2}^F$ and $V_{k,-,2}^T$, resp.) are used to fill up the corresponding variable-link and link-link gaps, then
variables $z_i$ and $z_k$ are both true (false, resp.) since otherwise $V_{i,+,2}^T$ and $V_{k,-,2}^F$ ($V_{i,+,2}^F$ and $V_{k,-,2}^T$, resp.) would have been used to fill up the variable-truth gaps. Hence, $(z_i\oplus\neg z_k)$ is satisfied. Similarly, if $I_{sat}$ is satisfiable, then every gap can be filled up.

Chen et al.~\cite{chen2014optimality} provided a reduction that meets the above requirement with job processing times, and hence the target value $Q$, being $O(n^{1+\delta})$ for arbitrarily  small constant $\delta>0$. Unfortunately, using this reduction we can only deduce a weaker lower bound of $2^{O((1/\epsilon)^{1/3-\delta})}$ (see Appendix~\ref{appsec:old} for a detailed discussion). For our purpose, we need to design job processing times to achieve a stronger {\em ratio-preserving} property, as we elaborate below.


Recall that we are given an instance of 3SAT$'$ with $n$ variables and $4n/3$ clauses.  For a given solution, we say a machine is {\em good} if its load is exactly $Q$ (in which case there is exactly one gap job on it and the gap is filled up exactly), and is {\em bad} otherwise (in which case its load is at least $Q+1/2$ or at most $Q-1/2$). Our scheduling instance will additionally satisfy the following properties:
\begin{compactitem}
\item[i.] There are $m=\tilde{O}(n)$ machines and the target makespan is $Q=\tilde{O}(n)$.
\item[ii.] Each processing time is a multiple of $1/2$ and the total job processing time equals $mQ$. 
\item[iii.] Conditions CO1 to CO4 are satisfied. Additionally, the following ratio-preserving properties are satisfied. For any $\vartheta\in (0,1)$ it holds that: 
\begin{compactitem}
\item If the 3SAT$'$ instance admits a truth assignment where at most $\vartheta n$ clauses are not satisfied, 
then the constructed scheduling instance admits a feasible solution with at most $\vartheta_1n$ bad machines, for some $\vartheta_1=\Theta(\vartheta)$. In particular, the load of these bad machines is $Q+1$ or $Q-1$. 
\item If any truth assignment for the 3SAT$'$ instance has at least $\vartheta n$ clauses that are not satisfied, then in any feasible schedule of the constructed scheduling instance there are at least $\vartheta_2n$ bad machines, for some $\vartheta_2=\Theta(\vartheta)$.
\end{compactitem}

\end{compactitem}

Before giving more details of the construction, we briefly argue that an instance satisfying properties (i)-(iii) implies Theorem~\ref{thm:lower-bound}. 
\begin{proof}[Proof Idea (Theorem~\ref{thm:lower-bound})]
Take $q=2$ for simplicity. We assume by contradiction that there exist some sufficiently small $\delta>0$, such that for any $\epsilon>0$ there is an $(1+\epsilon)$-approximation algorithm with running time $2^{O((1/\epsilon)^{1/2-\delta})}+n^{O(1)}$. Let $\beta\in (0,1)$ be a constant, and let $\epsilon',\delta>0$ be sufficiently small numbers, as in Lemma~\ref{lemma:maxsat-eth2}. We show that, for an appropriately chosen $\epsilon$, the PTAS can be used to distinguish, in time $2^{n^{1-\delta}}$, 3SAT$'$ instances where at least $(1-\epsilon')\cdot 4n/3$ clauses are satisfiable, from 3SAT$'$ instances where at most $(\beta+\epsilon')\cdot 4n/3$ clauses are satisfiable, contradicting ETH by Lemma~\ref{lemma:maxsat-eth2}. 
Indeed, we first observe that every bad machine will cause the objective value to increase by at least some fixed constant. A straightforward but crucial observation follows from the fact that, for load balancing problems, the total difference from the average load is~0. That is, if $C_i=Q+\Delta_i$, then the cost is
\begin{eqnarray}
\sum_{i=1}^m C_i^2= \sum_{i=1}^m (Q+\Delta_i)^2 =\sum_{i=1}^m (Q^2+  2T\Delta_i + \Delta_i^2) =  mQ^2 + \sum_{i=1}^m\Delta_i^2, \label{eq:load}
\end{eqnarray}
where the last equality follows as $\sum_{i}\Delta_i=0$ (for general $q>1$, a similar statement follows from a Taylor expansion, as in the proof of Lemma~\ref{thm:structure}). Consequently, if at least $(1-\epsilon')\cdot 4n/3$ clauses of $I_{sat}$ are satisfiable, then at most $\Theta(\epsilon'n)$ machines will have a load of either $Q+1$ or $Q-1$, and hence the optimal objective value of the constructed scheduling instance is at most $mQ^2+\Theta(\epsilon' n)$ by Eq~\eqref{eq:load}. On the other hand, if at most $(\beta+\epsilon')\cdot 4n/3$ clauses of $I_{sat}$ are satisfiable for some constant $\beta<1$, then $\Delta_i\ge 1/2$ for at least $\Theta((1-\beta-\epsilon')n)$ machines. By Eq~\eqref{eq:load} the optimal objective value of the constructed scheduling instance is at least $mQ^2+\Theta((1-\beta-\epsilon')n)=mQ^2+\Theta(n)$ (see Lemma~\ref{lemma:calculate} for the detailed computation). Now we apply the efficient PTAS with $\epsilon=\Theta(\frac{n^{1-\delta}}{mQ^2})\approx \Theta(1/n^{2+\delta})$. Given the fact that $mQ^2\epsilon=\Theta(n^{1-\delta})$, if at least $(1-\epsilon')\cdot 4n/3$ clauses of $I_{sat}$ are satisfiable, then the PTAS should return a schedule with objective value at most $mQ^2+\Theta(\epsilon' n)+\Theta(n^{1-\delta})=mQ^2+\Theta(\epsilon' n)$. Otherwise, the PTAS returns a schedule with objective value at least $mQ^2+\Theta(n)$. Theorem~\ref{thm:lower-bound} follows as our PTAS has a running time of $2^{O((1/\epsilon)^{1/2-\delta})}+n^{O(1)} \le 2^{O(n^{1-\delta})}$. 
\end{proof}

\noindent\textbf{Remark.} One can verify that if $Q$ is larger, e.g., $Q=\Theta(n^2)$, then the above argument only rules out a PTAS of running time $2^{O((1/\epsilon)^{1/4-\delta})}$. Hence, simultaneously enforcing the ratio-preserving property while having $Q=\tilde{O}(n)$ is the main technical challenge, which we overcome with our new number-theoretic constructions, as we elaborate in the following.

The rest of the paper is organized as follows. In Section~\ref{subsec:technical} we give an overview of the main technical ingredients for the construction of the processing times in our reduction. We also motivate our number theoretical constructions, which are specified in Section~\ref{subsec:adjacent-sum}. 
In Appendix~\ref{appsec:reduction-construction} we present the complete reduction. In Appendix~\ref{appsec:remaining-proofs} we show its correctness and conclude Theorem~\ref{thm:lower-bound}.

\subsection{Defining Processing times: Main Techniques}\label{subsec:technical}
To illustrate the main technical ingredient, in the following part of this subsection we will focus on conditions CO2 and CO3 while ignoring the other conditions (which can be handled using the techniques for CO2 and CO3). Recall that our goal is to create suitable gap jobs that can only be filled up by specific jobs. 

We can view each job, say, $V_{i,+,1}^{T}$, as a combination of three components --  the type-component $V_{\cdot,+,1}$ (indexed by $i$), the index-component $i$, and the $T/F$-component $T$. Ignoring dummy jobs for simplicity, conditions CO2 and CO3 involve 5 different type-components, including $V_{\cdot,+,1}$, $V_{\cdot,+,2}$ $V_{\cdot,-,1}$ $V_{\cdot,-,2}$ and $\textrm{CL}_{\cdot}$. Denote by $s(\cdot)$ the processing time of a job. We can define the processing time of a job into a summation of three terms corresponding to components, e.g., 
$s(V_{i,+,1}^{T})=\mu(V_{\cdot,+,1})+\sigma(i)+\eta(T)$, where the functions $\mu,\sigma,\eta$ map the type-component, index-component and T/F-component of a job to some positive integers. Now the question becomes: how can we define functions $\mu,\sigma,\eta$ such that from their sum, e.g., $\mu(V_{\cdot,+,1})+\mu(\textrm{CL}_{\cdot})+\sigma(i)+\sigma(\ell)+\eta(T)+\eta(F)$, we can conclude that it can only be added up by $s(V_{i',+,1}^\rho)$ and $s(\textrm{CL}_{\ell'}^{\rho'})$, where $\{i',\ell'\}=\{i,\ell\}$ and $\{\rho,\rho'\}=\{T,F\}$. Notice that there are only a constant number of different type-components and T/F-components, it is thus easy to define $\mu$ and $\eta$. For example, let $\sigma_{max}$ be a sufficiently large value that exceeds the maximal value of $\sigma$ and $\eta$, and define $\mu(V_{\cdot,+,1}),\mu(V_{\cdot,+,2}),\mu(V_{\cdot,-,1}),\mu(V_{\cdot,-,2}),\mu(\textrm{CL}_{\cdot})$ to be $10^5\sigma_{max}, 10^4\sigma_{max}, 10^3\sigma_{max}, 10^2\sigma_{max}, 10\sigma_{max}$, then from the sum $\mu(V_{\cdot,+,1})+\mu(\textrm{CL}_{\cdot})$ it is very easy to identify the type-components of two jobs. 

The main difficulty lies in the function $\sigma$ as we require job processing times to be $\tilde{O}(n)$, whereas $\sigma$ must map $[n]$ to $[\tilde{O}(n)]$ such that  
\begin{eqnarray}
\sigma(i)+\sigma(\ell)=\sigma(i')+\sigma(\ell')\implies \{i,\ell\}=\{i',\ell'\}. \label{eq:tech-1}
\end{eqnarray} 

In other words, we require the sum $\sigma(i)+\sigma(\ell)$ to be unique among the sums of all possible pairs $(i,\ell)$. Recall the special structure of 3SAT$'$, where each clause $cl_\ell\in C_1$ is of the form $(w_{\ell-1}\vee w_{\ell}\vee w_{\ell+1})$ such that $w_{j}\in\{z_j,\neg z_{j}\}$ for all $j$. Hence, for condition CO2, it suffices to guarantee Eq~\eqref{eq:tech-1} for $i\in\{\ell-1,\ell,\ell+1\}$.
Recall that a Salem-Spencer set is a set of numbers where no three of which form an arithmetic progression, hence if we let $\sigma$ map $[n]$ to a Salem-Spencer set of size $n$, Eq~\eqref{eq:tech-1} always holds for $\ell=i$. For our purpose, we need to generalize the construction of a Salem-Spencer set such that in addition to $2\sigma(\ell)$, the sum of any two adjacent numbers $\sigma(\ell)+\sigma(\ell+1)$ is also unique, as we show in Lemma~\ref{lemma:uniquesum-2}.

Condition CO3 is more complicated, as each clause of $C_2$ is of the form $(z_i\oplus\neg z_k)$ with $k=\tau(i)$, where the permutation $\tau$ is arbitrary. If we consider the index-components of the two variable jobs $V_{i,+,2}^\rho$ and $V_{k,-,2}^{\rho'}$, we cannot guarantee that $\sigma(i)+\sigma(k)=\sigma(i')+\sigma(k')$ implies $\{i,k\}=\{i',k'\}$. Consider the following indirect approach. Suppose for each $(z_i\oplus\neg z_k)$ we can construct a pair of jobs LN$_{(i,k)}^{T}$ and LN$_{(i,k)}^{F}$ (called link jobs), and meanwhile create two gaps such that they must be filled up by $V_{i,+,2}^{\rho_1}$ together with LN$_{(i,k)}^{\rho_1'}$, and $V_{k,-,2}^{\rho_2}$ together with LN$_{(i,k)}^{\rho_2'}$ respectively, and furthermore, $\{\rho_1,\rho_1'\}=\{\rho_2,\rho_2'\}=\{T,F\}$, then we know that if both gaps are filled up, then either $V_{i,+,2}^T$ and $V_{k,-,2}^F$, or
$V_{i,+,2}^F$ and $V_{k,-,2}^T$ are used, which is sufficient for condition CO3. Using this idea, instead of designing $\sigma$ such that the sum $\sigma(i)+\sigma(k)$ is unique, we seek to design $\sigma$ such the pair $(i,k)$ is ``uniquely -linked'' in the sense that there exists some number $e_i=\tilde{O}(n)$ such that the sums $\sigma(i)+e_i$ and $\sigma(k)+e_i$ are both unique among the sums of all pairs. Unfortunately, requiring the uniqueness of $\sigma(i)+e_i$ and $\sigma(k)+e_i$ is still too strong. We will show in Lemma~\ref{lemma:linked-sum} that for every $k=\tau(i)$ there exists a sequence of $\omega=O(\frac{\log n}{\log\log n})$ numbers $e_{i,1}$, $e_{i,2}$, $\ldots$, $e_{i,\omega}$ such that the sums $\sigma(i)+e_{i,1}$, $e_{i,1}+e_{i,2}$, $\ldots$, $e_{i,\omega-1}+e_{i,\omega}$, $e_{i,\omega}+\sigma(k)$ are all unique. Consequently, instead of creating one pair of link jobs, we will create $\omega$ pairs of link jobs for each $(z_i\oplus\neg z_k)$, ensuring condition CO3.

\subsection{Set of Integers with Unique Adjacent Sum and Linked Sum}\label{subsec:adjacent-sum}

In this section, we present our main technical contribution regarding the number-theoretic constructions needed in our reduction.

\noindent\textbf{Some notation.} Recall that we let $\Z_n=\{1,2,\cdots,n\}$. All the logarithms are taken with base $e$ unless stated otherwise. We will use $\cdot$ in the subscript to denote an arbitrary index, e.g., $x_{\cdot}$ refers to $x_i$ for some $i$. We write vectors in boldface, e.g. $\vex, \vey$. Vectors start with its $0$-th coordinate. For any $\upsilon$-dimensional vector $\vecc$, $\vecc[h]$ denotes its $h$-th coordinate for $0\le h\le \upsilon-1$, and $\vecc\vex=\sum_{h=0}^{\upsilon-1}\vecc[h]x^h$.

\begin{lemma}\label{lemma:uniquesum-1}
Let $N\in\mathbb{Z}^+$. There exists a subset ${\cal{S}}\subseteq \Z_N$ such that $|{\cal{S}}|\ge N^{1-c_0\sqrt{\frac{1}{{\log N}}}}$ for some sufficiently large $c_0$ (in particular, $c_0\ge7$ suffices), and for any $y\in {\cal{S}}$ and $1\le h\le 5$, the linear equation $h\cdot y=y_1+y_2+\cdots+y_h$ with $y_i\in {\cal{S}}$ for all $i$ has a unique solution $y_1=y_2=\cdots=y_h=y$.
\end{lemma}

The proof of Lemma~\ref{lemma:uniquesum-1} mainly utilizes the idea for constructing Salem–Spencer sets~\cite{behrend1946sets} and can be found in Appendix~\ref{appsec:ajacent-sum}. In particular, we can show that $|{\cal{S}}|\ge N^{1-7\sqrt{\frac{1}{{\log N}}}}$.
For any integer $d\in\Z^+$, we denote by ${\cal{S}}_d$ the subset of $\Z_d$ that satisfies Lemma~\ref{lemma:uniquesum-1}. Now we are ready to prove Lemma~\ref{lemma:uniquesum-2}, which is one of our two main number-theoretical results.

\begin{lemma}\label{lemma:uniquesum-2}
Let $N\in\mathbb{Z}^+$, $d=\lceil e^{(\sqrt{\log\log N}+c)^2} \rceil=e^{\OO(\sqrt{\log\log N})}\cdot \log N$ for some sufficiently large $c$ ($c\ge 7$ suffices) and $x=5d+1$. There exists an injection $\sigma: \Z_N\rightarrow \Z_{N'}$ such that:
\begin{compactitem}
\item[1.] $N'= N^{1+O(\frac{1}{\sqrt{\log \log N}})}$;
\item[2.] For any $i\in\Z_N$, $\sigma(i)=\sum_{j=0}^{\gamma}\vea_i[j]x^j$ for some $\vea_i[j]\in{\cal{S}}_{d}$, $0\le j\le \gamma$,  where $\gamma=\lceil \frac{\log N}{\log\log N} \rceil+O(\frac{\log N}{(\log\log N)^{3/2}})$;
\item[3.] For any $1\le h\le 5$ and  $i\in\Z_N$, the equation $h\cdot\sigma(i)=\sigma(i_1)+\sigma(i_2)+\cdots+\sigma(i_h)$, $i_j\in\Z_N$ has a unique solution $i_1=i_2=\cdots=i_h=i$. Further, the equation $h\cdot \sigma(i)=\sigma(i_1)+\sigma(i_2)+\cdots+\sigma(i_k)$, $i_j\in\Z_N$ has no feasible solution when $1\le k< h$ or $h<k\le 5$;
\item[4.] For any $i\le N-1$, the linear equation $\sigma(i)+\sigma(i+1)=\sigma(i_1)+\sigma(i_2)$, $i_1\le i_2$ has a unique solution $i_1=i$, $i_2=i+1$. Furthermore, the linear equation $\sigma(i)+\sigma(i+1)=\sigma(i_1)+\sigma(i_2)+\cdots+\sigma(i_k)$ has no feasible solution when $k=1$ or $2<k\le 5$.
\end{compactitem}
\end{lemma}

\begin{proof}
By Lemma~\ref{lemma:uniquesum-1} we know $|{\cal{S}}_{d}|\ge d^{1-c_0\sqrt{\frac{1}{{\log d}}}}$ for some constant $c_0$ (in particular, we can choose $c_0=7$). Let $\hat{\cal{S}}\subseteq {\cal{S}}_{d}$ be an arbitrary subset such that $|\hat{\cal{S}}|=2^\omega$ for some integer $\omega$ such that $|\hat{\cal{S}}|\ge 1/2\cdot |{\cal{S}}_{d}|$, then it is easy to see that $|\hat{\cal{S}}|=d^{1-\Theta(\sqrt{\frac{1}{{\log d}}})}$.  Consider all the integers that can be written as $\sum_{i=0}^{\beta}\vea[i]x^i=\vea\cdot\vex$ for some integer $\beta$, $\vea=(\vea[0],\vea[1],\cdots,\vea[\beta])\in \hat{\cal{S}}^{\beta+1}$,  and $\vex=(1,x,\cdots,x^\beta)$, for $x=5d+1$. It is easy to see that we obtain $|\hat{\cal{S}}|^{\beta+1}$ different integers constructed this way.

Simple calculations show that $|\hat{\cal{S}}|^{\beta+1}\ge N$ if 
$$\beta\ge \frac{\log N}{\log d} \left(1+\Theta\left(\sqrt{\frac{1}{\log d}}\right)\right).$$ 
Hence, by picking $\beta=\lceil \frac{\log N}{\log\log N} \rceil+O(\frac{\log N}{(\log\log N)^{3/2}}),$ we can guarantee that $|\hat{\cal{S}}|^{\beta+1}\ge N$. For $d\ge e^{(\sqrt{\log\log N}+7)^2}$, we notice that $\sqrt{\log d}\ge \sqrt{\log\log N}+7$, hence 
\[|{\cal{S}}_{d}|\ge d^{1-7\sqrt{\frac{1}{{\log d}}}}=e^{\log d-7\sqrt{\log d}}>e^{\log\log N}=\log N>\beta+1.\] Hence, we can define an arbitrary injection $g$ that maps $j\in \{0,1,\cdots,\beta\}$ to a distinct number in~${\cal{S}}_{d}$.

Consider all the vector $\vea$'s. 
For any two vectors $\vea_j$ and $\vea_k$, we say they are {\em close} if $\vea_j$ and $\vea_k$ differ by exactly one coordinate, i.e., there exists some $0\le j^*\le \beta$ such that $\vea_i[j]=\vea_k[j]$ for all $j\neq j^*$ and $\vea_i[j^*]\neq\vea_k[j^*]$. We claim the following. 
\begin{claim}\label{claim:1}
Vectors in $ \hat{\cal{S}}^{\beta+1}$ can be ordered such that any two consecutive vectors are close.
\end{claim}
\begin{proof}
Recall that $|\hat{{\cal{S}}}|=2^\omega$, hence we can map each $a_j\in \hat{\cal{S}}$ to a distinct $\omega$-bit binary number (or more specifically, a binary string) within $\{0,1\}^{\omega}$. Let $\xi:\hat{\cal{S}}\rightarrow \{0,1\}^{\omega}$ be an arbitrary one-to-one mapping, then we can define an extended mapping $\xi':{\hat{\cal{S}}}^{\beta+1}\rightarrow \{0,1\}^{(\beta+1)\omega}$ such that $\vea$ is mapped to a $(\beta+1)\omega$-bit binary string $\overline{\xi(\vea[0])\xi(\vea[1])\cdots\xi(\vea[{\beta}])}_2$. If we can order all $(\beta+1)\omega$-bit binary string such that every adjacent numbers differ by exactly one bit, then the inverse of these binary strings gives a sequence of $\vea$'s such that adjacent vectors are close. 

Now we prove the following statement: for any $n\in \mathbb{Z}$, $n\ge 2$ and an arbitrary string $b\in \{0,1\}^n$, all binary strings of $\{0,1\}^n$ can be ordered in a sequence starting with $b$ such that any two adjacent strings differ by exactly one bit. We show this by induction. The statement is clearly true for $n=2$. Suppose it is true for all $n\le n'$, we prove it also holds for $n=n'+1$. Consider the first bit of $b$, which can be $0$ or $1$. Assume it is $1$ (the case of $0$ can be proved in a similar way), then $b=1b_1$ for some $b_1\in \{0,1\}^{n'}$. According to the induction hypothesis, all binary strings of $\{0,1\}^{n'}$ can be ordered in a sequence starting with $b_1$ such that any two adjacent binary strings only differ by one bit. Let such a sequence be $b_1,b_2,\cdots,b_{2^{n'}}$, then all binary strings of $\{0,1\}^{n'+1}$ can be ordered as $1b_1,1b_2,\cdots,1b_{2^{n'}},0b_{2^{n'}}, 0b_{2^{n'}-1},\cdots,0b_1$. Hence, the statement is true, and Claim~\ref{claim:1} follows. 
\end{proof}

Now consider an arbitrary ordering of vectors of $\hat{\cal{S}}^{\beta+1}$ that satisfies Claim~\ref{claim:1}. Let the sequence be $\vea_1,\vea_2,...$ where $\vea_i=(\vea_i[0],\vea_i[1],\cdots,\vea_i[\beta])$ denote the $i$-th vector in the sequence. Recall that each $\vea_i$ is a $(\beta+1)$-dimensional vector. Let ${prec}(i)$ be the unique coordinate where $\vea_i$ and $\vea_{i-1}$ differ. Similarly, let ${succ}(i)$ be the coordinate where $\vea_i$ and $\vea_{i+1}$ differ. By definition it holds that ${prec}(i+1)={succ}(i)$. For the first and last vectors in the sequence, we define additionally that $\vea_1[prec(1)]=\vea_{|\hat{\cal{S}}^{\beta+1}|}[succ(|\hat{\cal{S}}^{\beta+1}|)]=0$.
Let $\gamma=\beta+8$ and recall the injection $g:\{0,1,\cdots,\beta\}\rightarrow {\cal{S}}_{d}$. We define $\sigma$ such that 
\begin{align*}
&\sigma(i)=\vea_i\cdot\vex+\vea_i[{prec(i)}]x^{\beta+1}+\vea_i[{succ(i)}]x^{\beta+2}+g\left(prec(i)\right) x^{\beta+5}+g\left(succ(i)\right) x^{\beta+6},  \quad \text{if $i$ is odd}; \\
&\sigma(i)=\vea_i\cdot\vex+\vea_i[{prec(i)}]x^{\beta+3}+\vea_i[{succ(i)}]x^{\beta+4}+g\left(prec(i)\right) x^{\beta+7}+g\left(succ(i)\right) x^{\beta+8},  \quad \text{if $i$ is even}.
\end{align*}
where $\vex=(1,x,x^2,\cdots,x^\beta)$, while noting that $prec(i),succ(i)\le \beta<x$. Also, remark that all coeficients in the polynomial expression belong to $\mathcal{S}_d$. Given that $x=5d+1$ and $\beta=\frac{\log N}{\log d} \left(1+\Theta\left(\sqrt{\frac{1}{\log d}}\right)\right)$, one can verify that \[\sigma(i)\le (5d+1)^{\beta+9}=e^{(\beta+9)\log(5d+1)}\le e^{\left(\frac{\log N}{\log d}+\OO\left(\frac{\log N}{(\log d)^{3/2}}\right)\right)(\log d+\OO(1))}\le N'= N^{1+O(\frac{1}{\sqrt{\log \log N}})},\] for any $i\le N$, hence \textbf{Properties 1} and \textbf{2} of Lemma~\ref{lemma:uniquesum-2} hold.

Consider the equation $h\cdot \sigma(i)=\sigma(i_1)+\sigma(i_2)+\cdots+\sigma(i_k)$ for $k,h\le 5$. The right-hand side of this equation can be expressed as $\sigma(i_1)+\sigma(i_2)+\cdots+\sigma(i_k) = \sum_{j=1}^{\gamma}b_jx^j$, for some coefficients~$b_j$. Notice as each coefficient $\vea_i[j]$ belongs to $\hat{\cal{S}}\subseteq {\cal{S}}_d$, which is at most $\vea_i[j]\le d<x/5$, then each $b_h$ is a sum of at most $k\le 5$ such coefficients, and thus $0\le b_h <x$. We obtain a similar statement for the left-hand side. Hence the coefficients of terms of the same degree must coincide, and we have 
$h\cdot \vea_i[j]=\sum_{h=1}^k \vea_{i_h}[j]$ for all $0\le j\le \beta.$
According to Lemma~\ref{lemma:uniquesum-1}, we know the only solution for the above is $k=h$ and $i_1=i_2=\cdots=i_k=i$, hence \textbf{Property 3} is true.

It remains to prove \textbf{Property 4}. We suppose $i$ is odd in the following; the case of $i$ being even can be proved analogously.
Consider $\sigma(i)+\sigma(i+1)$ which is equal to
\begin{eqnarray*}
&&\sum_{j=0}^{\beta}b_jx^j+\vea_i[{prec(i)}]x^{\beta+1}+\vea_i[{succ(i)}]x^{\beta+2}+\vea_{i+1}[{prec(i+1)}]x^{\beta+3}+\vea_{i+1}[{succ(i+1)}]x^{\beta+4}\\
&+&g\left({prec(i)}\right)x^{\beta+5}+g\left({succ(i)}\right)x^{\beta+6}+g\left({prec(i+1)}\right)x^{\beta+7}+g\left({succ(i+1)}\right)x^{\beta+8}.
\end{eqnarray*}

As before, the coefficients of terms of the same degree must coincide.
We know that $\vea_i$ and $\vea_{i+1}$ only differs at coordinate $succ(i)=prec(i+1)$, hence $b_j=2\vea_i[j]$ for $j\neq succ(i)$. If $\sigma(i_1)+\sigma(i_2)=\sigma(i)+\sigma(i+1)$, we know $\vea_{i_1}[j]+\vea_{i_2}[j]=2\vea_i[j]$ for $j\neq succ(i)$, and by the fact that $\vea_k[j]\in{\cal{S}}_d$ we know it must hold that 
\begin{eqnarray}
\vea_{i_1}[j]=\vea_{i_2}[j]=\vea_i[j] \quad \text{for all }j\neq succ(i). \label{eq:1}
\end{eqnarray}

Now consider the $succ(i)$-th coordinate. We know that among $i_1$ and $i_2$ one is even and one is odd, for otherwise in $\sigma(i_1)+\sigma(i_2)$ either the coefficients of $x^{\beta+1}$ and $x^{\beta+2}$ are 0, or the coefficients of $x^{\beta+3}$ and $x^{\beta+4}$ are 0. In either case, this means that $\vea_i$ or $\vea_{i+1}$ has a 0 coefficient, which is a contradiction as $\mathcal{S}_d\subseteq \mathbb{Z}_d$. Let $\{i_o,i_e\}=\{i_1,i_2\}$ where $i_o$ is odd and $i_e$ is even.
Then from $\sigma(i_1)+\sigma(i_2)=\sigma(i)+\sigma(i+1)$, we have
 \begin{eqnarray*}
&& \vea_{i_o}[{prec(i_o)}]x^{\beta+1}+\vea_{i_o}[{succ(i_o)}]x^{\beta+2}+\vea_{i_e}[{prec(i_e)}]x^{\beta+3}+\vea_{i_e}[{succ(i_e)}]x^{\beta+4}+\\
&+&g\left({prec(i_o)}\right)x^{\beta+5}+g\left({succ(i_o)}\right)x^{\beta+6}+g\left({prec(i_e)}\right)x^{\beta+7}+g\left({succ(i_e)}\right)x^{\beta+8}\\
&=&\vea_i[{prec(i)}]x^{\beta+1}+\vea_i[{succ(i)}]x^{\beta+2}+\vea_{i+1}[{prec(i+1)}]x^{\beta+3}+\vea_{i+1}[{succ(i+1)}]x^{\beta+4}+\\
&+&g\left({prec(i)}\right)x^{\beta+5}+g\left({succ(i)}\right)x^{\beta+6}+g\left({prec(i+1)}\right)x^{\beta+7}+g\left({succ(i+1)}\right)x^{\beta+8}.
 \end{eqnarray*}
Now we can deduce that $\vea_{i_o}[{succ(i_o)}]=\vea_i[{succ(i)}]$ and $succ(i_o)=succ(i)$ (since $g$ is an injection). Using Equation~\eqref{eq:1} we conclude that $\vea_{i_o}=\vea_i$. Similarly, $\vea_{i+1}[prec(i+1)]=\vea_{i_e}[prec(i_e)]$,  $prec(i+1)=prec(i_e)$. As $succ(i)=prec(i+1)$, Equation~\eqref{eq:1} yields that $\vea_{i_e}=\vea_{i+1}$.

Each vector in the sequence is unique, so we know $i_o=i$ and $i_e=i+1$. Using that $i_1\le i_2$ we obtain that $i_1=i$ and $i_2=i+1$. Hence, \textbf{Property 4} is proved. Lemma~\ref{lemma:uniquesum-2} follows.
\end{proof}

With Lemma~\ref{lemma:uniquesum-2}, we are ready for the main result of this section.

\begin{lemma}\label{lemma:linked-sum}
Let $\tau$ be an arbitrary permutation of $\Z_n$. Then there exists a set ${\cal{S}}=\{s_1,s_2,\cdots,s_n\}$ of positive integers together with an auxiliary set $E=\bigcup_{i=1}^nE_i$ of positive integers such that:
\begin{compactitem}
\item All integers in ${\cal{S}}\cup E$ are bounded by $n^{1+O(\frac{1}{\sqrt{\log\log n}})}$;
\item $E_i=\{e_{i,1},e_{i,2},\cdots,e_{i,\omega}\}$  for all $i$, where $\omega=O(\frac{\log n}{\log\log n})$;
\item $E_i\cap E_{i'}=\emptyset$ for any $i'\neq i$;
\item For every $i\in\Z_n$ and $k=\tau(i)$, each sum $s_i+e_{i,1}$, $e_{i,1}+e_{i,2}$, $\cdots$, $e_{i,\omega-1}+e_{i,\omega}$, $e_{i,\omega}+s_k$ is unique, that is, there is no other pair in ${\cal{S}}\cup E$ that adds up to the same value.
\end{compactitem}
In particular, all these properties are satisfied by setting $s_i=\sigma(i)$ where $\sigma:\Z_N\rightarrow \Z_{N'}$ is the function specified in Lemma~\ref{lemma:uniquesum-2} by taking $N=n$. 
\end{lemma}

We start with a natural proof idea. Recall Lemma~\ref{lemma:uniquesum-2}, where $\sigma(i)=(\vea_i[0],\vea_i[1],\cdots,\vea_i[{\gamma}])\cdot (1,x,\cdots,x^{\gamma})=\vea_i\cdot \vex$ such that $\vea_i[j]<x/5$, whereas when we add two values, say, $\sigma(i)+\sigma(k)$, we can directly add each coordinate $\vea_i[j]+\vea_k[j]$. Given $i$ and $k$, how can we guarantee that the equation $\vea_i+\vea_k=(\vea_i[0]+\vea_k[0],\vea_i[1]+\vea_k[1],\cdots,\vea_i[{\gamma}]+\vea_k[{\gamma}])=\vea_{\ell}+\vea_r$ has a unique solution $\{\ell,r\}=\{i,k\}$? A simple observation is that, since $\vea_i[j]\in {\cal{S}}_d$ (by Property 3 of Lemma~\ref{lemma:uniquesum-2}), we know that $2\vea_i[j]$ can only be expressed as $\vea_i[j]+\vea_i[j]$, and $\vea_i[j]$ can only be expressed as $\vea_i[j]+0$. Consequently, we may lift up the dimension by writing $\sigma(i)=(\vea_i[0],\vea_i[1],\cdots,\vea_i[{\gamma}],0)\cdot (1,x,\cdots,x^{\gamma},x^{\gamma+1})$, and consider the following sequence of numbers:

\small
\begin{align*}
& & (&\vea_i[0],&\vea_i[1],\quad\cdots,\quad&\vea_i[{\gamma-1}],&\vea_i[{\gamma}], \quad&0) \\
&\rightarrow& (&0,&\vea_i[1],\quad\cdots,\quad&\vea_i[{\gamma-1}],&\vea_i[{\gamma}], \quad&\vea_k[{0}])\\
&\rightarrow& (&\vea_k[0],&\vea_i[1],\quad\cdots,\quad&\vea_i[{\gamma-1}],&\vea_i[{\gamma}], \quad&0)\\
&\rightarrow& (&\vea_k[0],&0,\quad\cdots,\quad&\vea_i[{\gamma-1}],&\vea_i[{\gamma}], \quad&\vea_k[{1}])\\
&\rightarrow& (&\vea_k[0],&\vea_k[1],\quad\cdots,\quad&\vea_i[{\gamma-1}], &\vea_i[{\gamma}], \quad&0)\\
&\rightarrow& &\cdots\\
&\rightarrow& (&\vea_k[0],&\vea_k[1],\quad\cdots,\quad&\vea_k[{\gamma-1}], &0, \quad&\vea_k[{\gamma}])\\
&\rightarrow& (&\vea_k[0],&\vea_k[1],\quad\cdots,\quad&\vea_k[{\gamma-1}], &\vea_k[{\gamma}], \quad&0)
\end{align*}
\normalsize

It is easy to verify that the sum of any two adjacent vectors in the above sequence is unique, which gives possible values for $e_{i,j}$'s. Unfortunately, numbers constructed in this way do not necessarily satisfy that $E_i\cap E_{i'}=\emptyset$. In particular, there might exist some pair $i',k'$ with $k'=\tau(i')$ where $\vea_i[j]=\vea_{i'}[j]$ for $j\le \gamma-2$, and $\vea_k[j]=\vea_{k'}[j]$ for $j\ge \gamma -1$. In this case, we have $$(\vea_{i'}[0],\vea_{i'}[1],\cdots,\vea_{k'}[{\gamma-1}], \vea_{k'}[{\gamma}], 0)=(\vea_i[0],\vea_i[1],\cdots,\vea_k[{\gamma-1}], \vea_k[{\gamma}], 0),$$ 
violating $E_i\cap E_{i'}=\emptyset$. 

How can we construct unique $e_{i,j}$'s?
Towards this, we consider all the one-to-one mappings $g:{\cal{S}}^{\gamma+1}_d\rightarrow {\cal{S}}^{\gamma+1}_d$. Under composition of functions, all such mappings form a group $Aut({\cal{S}}^{\gamma+1}_d)$. We are interested in the special mapping that maps each $\vea_i$ to $\vea_k$ (which corresponds to the permutation $\tau$), which belongs to $Aut({\cal{S}}^{\gamma+1}_d)$. We show that any mapping in $Aut({\cal{S}}^{\gamma+1}_d)$, and hence this special mapping, can be decomposed into a sequence of simple mappings. More precisely, we consider any finite set ${\cal{M}}$ and the group $Aut({\cal{M}}^{\upsilon})$ of one-to-one mappings from ${\cal{M}}^{\upsilon}$ to itself. We call a mapping in $Aut({\cal{M}}^{\upsilon})$ an $h$-shuffler if this mapping only changes the $h$-th coordinate of the input vector, i.e., an $h$-shuffler $f$ satisfies that for any $\vey\in {\cal{M}}^{\upsilon}$,
$$f(\vey)=f(\vey[0],\vey[1],\cdots,\vey[\upsilon-1])=(\vey[0],\cdots,\vey[h-1],z,\vey[h+1],\cdots,\vey[\upsilon-1])),$$
for some $z\in {\cal{M}}$.
We show the following group theoretic lemma which states that any mapping of $Aut({\cal{M}}^{\upsilon})$ can be decomposed into $2\upsilon$ $h$-shufflers; see Appendix~\ref{appsubsec:shuffle} for the proof. Our $e_{i,j}$'s can be obtained from these $h$-shufflers. 

\begin{lemma}\label{lemma:shuffle-1}
Let ${\cal{M}}$ be a finite set, $\upsilon\in\Z^+$ and $Aut({\cal{M}}^{\upsilon})$ be the group of all one-to-one mappings from ${\cal{M}}^{\upsilon}$ to itself. The group operation is function composition and denoted as $\circ$. For any $\pi\in Aut({\cal{M}}^{\upsilon})$, there exist $h$-shufflers $f_h,\hat{f}_h\in Aut({\cal{M}}^{\upsilon})$ for every $1\le h\le \upsilon$ such that $F(\vey)=\hat{F}(\pi(\vey))$ for any $\vey\in {\cal{M}}^{\upsilon}$, where $F=f_{\upsilon-1}\circ f_{\upsilon-2}\circ \cdots\circ f_{0}$ and $\hat{F}=\hat{f}_{\upsilon-1} \circ \hat{f}_{\upsilon-2}\circ \cdots\circ \hat{f}_{0}$. Furthermore, $f_h$'s and $\hat{f}_h$'s can be constructed in time that is polynomial in $|{\cal{M}}^\upsilon|$.
\end{lemma}
Lemma~\ref{lemma:shuffle-1} implies a decomposition of $\pi$ into $h$-shufflers, i.e., $\pi=\hat{f}_0^{-1}\circ \hat{f}_1^{-1}\circ\cdots\circ \hat{f}_{\upsilon-1}^{-1}\circ f_{\upsilon-1}\circ\cdots\circ f_0$.

With Lemma~\ref{lemma:shuffle-1}, we are ready to prove Lemma~\ref{lemma:linked-sum}. We first set the value of all parameters. 
Towards this, we will apply Lemma~\ref{lemma:uniquesum-2} twice. 

At first, we apply Lemma~\ref{lemma:uniquesum-2} by taking
$N=n$. Then we obtain $\sigma: \Z_n\rightarrow \Z_{n'}$ where $n'=n^{1+\OO(1/\sqrt{\log\log n})}$, $\sigma(i)=\sum_{j=0}^\gamma \vea_i[j]x^j$ for $\vea_i[j]\in {\cal{S}}_d$ where $\gamma=\lceil\frac{\log n}{\log\log n}\rceil+\OO(\frac{\log n}{(\log\log n)^{3/2}})$, $d=\lceil e^{(\sqrt{\log\log n}+7)^2}\rceil=e^{\OO(\sqrt{\log\log n})}\cdot \log n $ and $x=5d+1$. Except $N$, all the other parameters, including $n',\sigma,x,d,\gamma$ and $\vea_i$'s are fixed throughout the following part of this section.

Next, we apply Lemma~\ref{lemma:uniquesum-2} again by setting $N=4\gamma+4$ where $\gamma$ takes the value we determined above. 
By doing so we obtain another injection $\sigma'$. We have the following simple observation.

\begin{observation}
If 
$\ell\le 4\gamma +4$, then $\sigma'(\ell)=o(\log^2 n)<x^2$.
\end{observation}
\begin{proof}
Note that $\gamma=O(\frac{\log n}{\log\log n})$. By Lemma~\ref{lemma:uniquesum-2}, $\sigma'(\ell)\le (4\gamma+4)^{1+O(\frac{1}{\sqrt{\log\log (4\gamma+4)}})}=o(\gamma^2)=o(\log^2n)$. 
\end{proof}

Next, We will apply Lemma~\ref{lemma:shuffle-1}. In the following part of this paper, any $\upsilon$-dimensional vector~$\vecc$ represents the number given by the polynomial expression $\sum_{i=1}^{\upsilon-1}\vecc[i]x^i$; vectors and polynomial expressions are used interchangeably. 
Given our permutation $\tau$, we define $\hat{\tau}$ as a one-to-one mapping that maps each vector $\vea_i$ to $\vea_k$, or equivalently, maps $\sigma(i)$ to $\sigma(k)$ if $k=\tau(i)$. Notice that $\vea_i$'s form a subset of ${\cal{S}}_d^{\gamma+1}$, so currently $\hat{\tau}$ is only defined on this subset. We can extend $\hat{\tau}$ to ${\cal{S}}_d^{\gamma+1}$ such that for $\vecc\in {\cal{S}}_d^{\gamma+1}$ and $\vecc\neq \vea_i$, then $\hat{\tau}(\vecc)=\vecc$. Hence, $\hat{\tau} \in Aut({\cal{S}}_d^{\gamma+1})$. According to Lemma~\ref{lemma:shuffle-1}, we can obtain $h$-shufflers $f_h$ and $\hat{f}_h$ for all $0\le h\le \gamma$ such that $f_{\gamma}\circ f_{\gamma-1}\circ\cdots\circ f_{0}=\hat{f}_{\gamma}\circ\hat{f}_{\gamma-1}\circ\cdots\circ \hat{f}_0\circ \hat{\tau}$.

For ease of notation, define $F_h=f_h\circ f_{h-1}\circ\cdots\circ f_0$ and $\hat{F}_h=\hat{f}_h\circ \hat{f}_{h-1}\circ\cdots\circ \hat{f}_0$. As $h$-shufflers only changes the $h$-th coordinate, we have the following observation.

\begin{observation}\label{obs:F-hat} The following statements are true:
\begin{itemize}
    \item For any $0\le h\le \gamma$ $$(F_{\gamma}(\vea_i))[h]=(F_{\gamma-1}(\vea_i))[h]=\cdots=(F_h(\vea_i))[h],$$
    $$(\hat{F}_{\gamma}(\vea_i))[h]=(\hat{F}_{\gamma-1}(\vea_i))[h]=\cdots=(\hat{F}_h(\vea_i))[h];$$
    \item For any $0\le h\le \gamma$, $$(F_{h-1}(\vea_i))[h]=(F_{h-2}(\vea_i))[h]=\cdots=(F_{0}(\vea_i))[h]=\vea_i[h],$$  $$(\hat{F}_{h-1}(\vea_i))[h]=(\hat{F}_{h-2}(\vea_i))[h]=\cdots=(\hat{F}_{0}(\vea_i))[h]=\vea_i[h].$$
\end{itemize}

\end{observation}

Now we are ready to construct a unique linking sequence for every $i$. 
Intuitively, note that since $\sigma'(h)<x^2$, then $\sigma'(h)$ occupies two \lq\lq bits\rq\rq in the polynomial (that is, two coordinates). With this in mind, we let $\hat{\sigma}'(i)=(0,\sigma'(i))$. Moreover, we define $\hat{\sigma}'(0)=(0,0)$.  Now consider the following two sequences, starting from $\vea_i$ and $\vea_k$, respectively, and end up at the same vector:
\small
\begin{align*}
&&(&\vea_i[0],&\vea_i[1],\quad\cdots,\quad&\vea_i[{\gamma-1}], &\vea_i[{\gamma}],\quad&0,&\hat{\sigma}'(0)) &:=&\veb^0_{i}\\&\rightarrow& (&0,&\vea_i[1],\quad\cdots,\quad &\vea_i[{\gamma-1}],&\vea_i[{\gamma}], \quad&(F_0(\vea_i))[0], &\hat{\sigma}'(1))&:=&\veb^1_{i}\\
&\rightarrow& (&(F_0(\vea_i))[0],&\vea_i[1],\quad\cdots,\quad&\vea_i[{\gamma-1}],&\vea_i[{\gamma}], \quad&0, &\hat{\sigma}'(2))&:=&\veb^2_{i}\\
&\rightarrow& (&(F_1(\vea_i))[0],&0,\quad\cdots,\quad&\vea_i[{\gamma-1}],&\vea_i[{\gamma}], \quad&(F_1(\vea_i))[1], &\hat{\sigma}'(3))&:=&\veb^3_i\\
&\rightarrow& (&(F_1(\vea_i))[0],&(F_1(\vea_i))[1],\quad\cdots,\quad&\vea_i[{\gamma-1}], &\vea_i[{\gamma}], \quad&0, &\hat{\sigma}'(4))&:=&\veb^4_i\\
&\rightarrow& &\cdots\\
&\rightarrow& (&(F_{\gamma}(\vea_i))[0],&(F_{\gamma}(\vea_i))[1],\quad\cdots,\quad&(F_{\gamma}(\vea_i))[\gamma-1], &0, \quad&(F_{\gamma}(\vea_i))[\gamma], &\hat{\sigma}'(2\gamma+1)) &:=&\veb^{2\gamma+1}_i\\
&\rightarrow& (&(F_{\gamma}(\vea_i))[0],&(F_{\gamma}(\vea_i))[1],\quad\cdots,\quad&(F_{\gamma}(\vea_i))[\gamma-1], &(F_{\gamma}(\vea_i))[\gamma], \quad&0, &\hat{\sigma}'(2\gamma+2))&:=&\veb^{2\gamma+2}_i
\end{align*}\normalsize
and
\small
\begin{align*}
&&(&\vea_k[0],&\vea_k[1],\quad\cdots,\quad&\vea_k[{\gamma-1}], &\vea_k[{\gamma}],\quad&0,&\hat{\sigma}'(4\gamma+4)) &:=&\hat{\veb}^0_k\\&\rightarrow& (&0,&\vea_k[1],\quad\cdots,\quad &\vea_k[{\gamma-1}],&\vea_k[{\gamma}], \quad&(\hat{F}_0(\vea_k))[0], &\hat{\sigma}'(4\gamma+3)) &:=&\hat{\veb}^1_k\\
&\rightarrow& (&(\hat{F}_0(\vea_k))[0],&\vea_k[1],\quad\cdots,\quad&\vea_k[{\gamma-1}],&\vea_k[{\gamma}], \quad&0, &\hat{\sigma}'(4\gamma+2))&:=&\hat{\veb}^2_k\\
&\rightarrow& (&(\hat{F}_1(\vea_k))[0],&0,\quad\cdots,\quad&\vea_k[{\gamma-1}],&\vea_k[{\gamma}], \quad&\hat{F}_1(\vea_k))[1]), &\hat{\sigma}'(4\gamma+1))&:=&\hat{\veb}^3_k\\
&\rightarrow& (&(\hat{F}_1(\vea_k))[0],&\hat{F}_1(\vea_k))[1],\quad\cdots,\quad&\vea_k[{\gamma-1}], &\vea_k[{\gamma}], \quad&0, &\hat{\sigma}'(4\gamma))&:=&\hat{\veb}^4_k\\
&\rightarrow& &\cdots\\
&\rightarrow& (&(\hat{F}_{\gamma}(\vea_k))[0],&(\hat{F}_{\gamma}(\vea_k))[1],\quad\cdots,\quad&(\hat{F}_{\gamma}(\vea_k))[\gamma-1], &0, \quad&(\hat{F}_{\gamma}(\vea_k))[\gamma]), &\hat{\sigma}'(2\gamma+3))&:=&\hat{\veb}^{2\gamma+1}_k\\
&\rightarrow& (&(\hat{F}_{\gamma}(\vea_k))[0],&(\hat{F}_{\gamma}(\vea_k))[1],\quad\cdots,\quad&(\hat{F}_{\gamma}(\vea_k))[\gamma-1]), &(\hat{F}_{\gamma}(\vea_k))[\gamma], \quad&0, &\hat{\sigma}'(2\gamma+2))&:=&\hat{\veb}^{2\gamma+2}_k
\end{align*}\normalsize

Consider each vector in the above sequence, say, $\veb^2_{i}$. According to Observation~\ref{obs:F-hat}, we know
\begin{align*}
&&(&(F_0(\vea_i))[0],&\vea_i[1],\quad&\cdots,&\vea_i[{\gamma-1}],\quad&\vea_i[{\gamma}], &0, \quad&\hat{\sigma}'(2))\\
&=&(&(F_0(\vea_i))[0],&(F_0(\vea_i))[1],&\cdots,&(F_0(\vea_i))[\gamma-1],\quad&(F_0(\vea_i))[\gamma], &0, \quad&\hat{\sigma}'(2))
\end{align*}
More generally, it is easy to verify that every $\veb^{2j}_i$ is the concatenation of the vector $\left(F_{j-1}(\vea_i),0\right)$ and $\hat{\sigma}'(2j)$, and each $\veb^{2j+1}_i$ is a combination $SW_j\left(\left(F_{j-1}(\vea_i),0\right)\right)$ and $\hat{\sigma}'(2j+1)$, where $SW_j$ is a one-to-one mapping that swaps two coordinates of a vector. A similar statement holds for $\hat{\veb}^{2j}_k$'s and $\hat{\veb}^{2j+1}_k$'s. Since $SW_j$'s, $F_j$'s and $\hat{F}_j$'s are all one-to-one mappings, and $\sigma'$ is an injection, each of the vectors in the sequence above is unique. More precisely, we have the following.
\begin{lemma}\label{lemma:veb-uniquesum-1}
For any $0\le h\le 2\gamma+1$ and $1\le i\le n$, 
\begin{itemize}
    \item If $\veb^h_{i}=\veb^{h'}_{i'}$, then $h=h'$ and $i=i'$;
    \item If $\hat{\veb}^h_i=\hat{\veb}^{h'}_{i'}$, then $h=h'$ and $i=i'$.
\end{itemize} 
\end{lemma}

Furthermore, by the fact that $F_{\gamma}=\hat{F}_{\gamma}\circ \hat{\tau}$ and $\vea_k=\hat{\tau}(\vea_i)$, we have the following observation:
\begin{observation}
$$\veb^{2\gamma+2}_i=\hat{\veb}^{2\gamma+2}_k.$$ 
\end{observation}

Next, we consider any two adjacent vectors in the above sequence. We observe that they differ at exactly three positions -- the last coordinate (i.e., $\hat{\sigma}'(j)$'s), and other two coordinates such that one of the two vectors has $0$ coordinate. Other coordinates, e.g., $(F_0(\vea_i))[0]$ in $\veb^2_{i}$ and $(F_1(\vea_i))[0]$ in $\veb^3_i$ are identical according to Observation~\ref{obs:F-hat}. This leads to the following Lemma.

\begin{lemma}\label{lemma:veb-uniquesum-2}
For any $0\le h\le 2\gamma+1$ and $1\le i\le n$, \begin{itemize}
    \item If $\veb^h_{i}+\veb^{h+1}_i=\veb^{h'}_{i'}+\veb^{h''}_{i''}$ where $h'\le h''$, then $h'=h$, $h''=h+1$, $i'=i''=i$;
    \item If $\hat{\veb}^h_i+\hat{\veb}^{h+1}_i=\hat{\veb}^{h'}_{i'}+\hat{\veb}^{h''}_{i''}$ where $h'\le h''$, then $h'=h$, $h''=h+1$, $i'=i''=i$.
\end{itemize} 
\end{lemma}
\begin{proof}
We prove the first statement, that is, $\veb^h_{i}+\veb^{h+1}_i=\veb^{h'}_{i'}+\veb^{h''}_{i''}$ implies $h'=h$, $h''=h+1$, $i'=i''=i$. The second statement can be proved in the same way. 

We first consider the last coordinate of the summation $\veb^h_{i}+\veb^{h+1}_i=\veb^{h'}_{i'}+\veb^{h''}_{i''}$, which is $\sigma'(h)+\sigma'(h+1)=\sigma'(h')+\sigma'(h'')$. If $h=0$, according to \textbf{Property 3} of Lemma~\ref{lemma:uniquesum-2}, $1\times \sigma'(1)$ can only be expressed as $0+\sigma'(1)$, hence we have $h'=0$ and $h''=1$. If $h\ge 1$, according to \textbf{Property 4} of Lemma~\ref{lemma:uniquesum-2}, we have $h'=h$ and $h''=h+1$.

Consider other coordinates of the equation $\veb^h_{i}+\veb^{h+1}_i=\veb^{h}_{i'}+\veb^{h+1}_{i''}$. On the left-side it is either $a$ or $2a$ where $a\in {\cal{S}}_d$. Consider the equation $a=b+c$ where $b,c\in \{0\}\cup {\cal{S}}_d$. By Lemma~\ref{lemma:uniquesum-1}, there do not exist two numbers in $ {\cal{S}}$ that add up to $a$, hence we know $b,c\in\{0,a\}$. Similarly if $2a=b+c$ for $a\in {\cal{S}}_d$ and $b,c\in\{0\}\cup {\cal{S}}_d$, then $b$ and $c$ are both nonzero, for otherwise the linear equation $2y=y_1$ admits a solution, which is a contradiction to that ${\cal{S}}_d$ satisfies Lemma~\ref{lemma:uniquesum-1}. Hence, $2a=b+c$ for $a,b,c\in {\cal{S}}_d$, and again by Lemma~\ref{lemma:uniquesum-1} we have $b=c=a$. Hence, we conclude that each coordinate of $\veb^{h}_{i'}$ and $\veb^{h+1}_{i''}$ must be the same as $\veb^h_{i}$ and $\veb^{h+1}_i$, respectively. By Lemma~\ref{lemma:veb-uniquesum-1}, it follows that $i'=i''=i$.
\end{proof}

Using the same argument, we know if $2\veb_i^{2\gamma+2}=\veb_{i'}^{h'}+\hat{\veb}_{i''}^{h''}$, then $\veb_i^{2\gamma+2}=\veb_{i'}^{h'}=\hat{\veb}_{i''}^{h''}$. Thus the following is also true.
\begin{lemma}\label{lemma:veb-uniquesum-3}
$2\veb_i^{2\gamma+2}=\veb_{i'}^{h'}+\hat{\veb}_{i''}^{h''}$, then $h'=h''=2\gamma+2$, $(i',i'')=(i,\tau(i))$.
\end{lemma}

We are now ready to prove Lemma~\ref{lemma:linked-sum}.
\begin{proof}[Proof of Lemma~\ref{lemma:linked-sum}]The sequence of $E_i$ that links $\sigma(i)$ and $\sigma(k)$ for $k=\tau(i)$ is exactly $\veb^0_{i}=\vea_i,\veb^1_{i},\veb^2_{i},\cdots,\veb^{2\gamma+2}_i=\hat{\veb}^{2\gamma+2}_k$, $\hat{\veb}^{2\gamma+1}_k$, $\cdots$, $\hat{\veb}^1_k,\hat{\veb}^0_k=\vea_k$ (recall that by a vector $\veb$ we mean the integer $\veb\cdot\vex$ with $x=5d+1=O(\log n)$). The uniqueness of each linking sequence is ensured by Lemma~\ref{lemma:veb-uniquesum-1} and Lemma~\ref{lemma:veb-uniquesum-2}. The largest number is bounded by $x^{\gamma+5}=n^{1+O(\frac{1}{\sqrt{\log\log n}})}$. Furthermore, $\omega=O(\gamma)=O(\frac{\log n}{\log\log n})$. Thus, all properties of Lemma~\ref{lemma:linked-sum} are satisfied. 
\end{proof}

\clearpage

\appendix

\section{Omitted Proofs in Section~\ref{sec:upper} - Algorithms AL\texorpdfstring{$_1$}{Lg}, AL\texorpdfstring{$_2$}{Lg}, AL\texorpdfstring{$_3$}{Lg}}\label{appsec:alg}

\subsection{Algorithm 1}\label{appsubsec:al1}

\begin{lemma}
Consider an instance after the preprocessing of Lemma~\ref{lemma:prepocessing}.	For any $\epsilon>0$, there exists an algorithm \textrm{AL}$_{\ref{alg:1}}$ that returns an $(1+O(\epsilon))$-approximate solution for $P||\sum_i C_i^q$ and runs in time $(m/\epsilon)^{O(m)}$.
\end{lemma}

The algorithm can be formulated as a dynamic program. For each $h\in \Z_n$, we create a set of states $\mathcal{F}_h$. A state $(h,L_1,\dots,L_m)$ belongs to $\mathcal{F}_h$ if it is possible to assign jobs $1,\dots,h$ on machines such that the total load on machine $i$ equals $L_i$ for all $i\in \Z_m$. Starting from $(0,\dots,0)\in \mathcal{F}_0$, every state in $\mathcal{F}_{i-1}$ can give rise to some states in $\mathcal{F}_i$ by trying all the possible assignment of job $h$. And the solution is given by the state in $\mathcal{F}_n$ with the minimum objective i.e., $\sum_{i=1}^m L_i^q$. Due to Lemma \ref{lemma:prepocessing}, we know that one of the optimal solutions has a load vector  $({L_1^*},\dots,{L_m^*})\in \mathcal{F}_n$ satisfying $L_i^*\leq 2$ for all $i$. Clearly, the running time of Algorithm \ref{alg:1} is $O(m\sum_i |\mathcal{F}_i|)$. However, the total number of states stored during the dynamic programming can be pseudo-polynomial. Fortunately, using the framework by Woeginger \cite{woeginger1999does}, we are able to trim the state space to make it polynomial. More precisely, for any $\delta>0$ construct $\hat{\mathcal{F}}_1,\dots,\hat{\mathcal{F}}_n$ satisfying the following property:
\begin{itemize}
	\item the size of each $\hat{\mathcal{F}}_i$ is bounded by $(1/\delta)^{O(m)}$;
	\item for each $(i,L_1,\dots,L_m)\in \mathcal{F}_i$, there exists $(i,\hat{L}_1,\dots,\hat{L}_m)\in \hat{\mathcal{F}}_i$ such that $L_j\leq \hat{L}_j \leq L_j(1+\delta)^i$.
\end{itemize}

In Algorithm \ref{alg:1}, we set the parameter $\delta=\epsilon/n$. Due to Lemma~\ref{lemma:prepocessing}, $n\in O(m/\epsilon)$, the total running time is bounded by $(m/\epsilon)^{O(m)}$. Let $(n,\hat{L}_1,\dots,\hat{L}_m)$ be the state with the minimum objective in $\hat{\mathcal{F}}_n$, and $(n,L_1^*,\dots,L_m^*)$ be the state with the minimum objective in $\mathcal{F}_n$. Hence, $(n,\hat{L}_1,\dots,\hat{L}_m)$ is an $(1+O(\epsilon))$-approximate solution, as, 

\[ \sum_i \hat{L}_i^q \leq (1+\delta)^{nq}
 \sum_i {L_i^*}^q  \leq (1+2q\epsilon) OPT,\]
 where $q$ is a fixed constant.

\begin{algorithm}[tb]
	\caption{Pseudo Code Description of AL$_1$}
	\label{alg:1}
	\textbf{Input}: $I,\epsilon$\\
	\textbf{Output}: $\min \{\sum_{i=1}^m \hat{L}_i^q: (\hat{L}_1,\dots,\hat{L}_m) \in \hat{\mathcal{F}_n} \}$
	\begin{algorithmic}[1]
		\STATE{$\delta=\epsilon/n$}
		\STATE{$\Gamma = \{[0], [\epsilon,\epsilon(1+\delta)),[\epsilon(1+\delta),\epsilon(1+\delta)^2),\dots \}$}
		\STATE{$\hat{\mathcal{F}}_0=\{(0,\dots,0)\}$ }
		\FOR{$j=1$ to $n$}
		\STATE {$\mathcal{F}'_j=\emptyset$}
		\FORALL {$(L_1,\dots,L_m) \in \mathcal{F}_{i-1}$}
		\FORALL{$i\in[1,m]$ and $L_i+p_j\leq 2$}
		\STATE{$\mathcal{F}'_j= \mathcal{F}'_j \cup (L_1,\dots,L_{i-1},L_i+p_j,L_{i+1},\dots,L_m)$}
		\ENDFOR
		\ENDFOR
		\STATE{$\hat{\mathcal{F}}_j=\emptyset$}
		\FORALL{$S\in \Gamma^m$}
		\FOR{$i=1$ to $m$}
		\STATE{$\hat{L}_i = \max \{L_i: (\dots,L_i,\dots)\in \mathcal{F}'\cap S \}$}
		\ENDFOR
		\STATE{$\hat{\mathcal{F}}_j = \hat{\mathcal{F}}_j \cup (\hat{L}_1,\dots,\hat{L}_m)$}
		\ENDFOR
		
		\ENDFOR
	\end{algorithmic}
\end{algorithm}

\subsection{Algorithm 2}\label{appsubsec:al2}
We first recall Lemma~\ref{lemma:al2} and then present its proof.
\begin{T1}
  	Consider an instance after the preprocessing of Lemma~\ref{lemma:prepocessing}. For any $\epsilon>0$, there exists an algorithm  \textrm{AL}$_{2}$ that outputs an $(1+O(\epsilon))$-approximation solution for $P||\sum_i C_i^q$ within $m^{\tilde{O}(1/\sqrt{\epsilon})}$ time.
\end{T1} 
\begin{proof}
Based on Lemma \ref{thm:structure}, we design \textrm{AL}$_{2}$ as a dynamic program as follows.

We say that a vector  $(i,v,u_1,u_2,\ldots,u_\tau)$ is a valid state if it is possible to assign the $u_h$ largest jobs in $\mathcal{G}_h$, for each $h$, on the first $i$ machines with the objective equal to $v$. In our algorithm, for each $i\in\{0,\ldots,m\}$, we construct a set $\mathcal{F}_i$ of valid states. Start from $(0,\dots,0)\in \mathcal{F}_0$. To construct the valid states in $\mathcal{F}_{i}$, we consider a state in $\mathcal{F}_{i-1}$ and try all possible assignments of jobs to machine $i$ that respect the ordering of jobs given by Lemma~\ref{thm:structure} and the load bound for each machine implied by Lemma~\ref{lemma:prepocessing}. Given the sets $\mathcal{F}_{0},\ldots,\mathcal{F}_{m}$, the answer can be found by searching the state with the minimum objective in $\mathcal{F}_m$. In order to limit the number of states, we can eliminate \emph{dominated} states. Namely, if two states $(i,v,u_1,\dots,u_\tau)$ and $(i,v',u_1,\dots,u_\tau)$ in $\mathcal{F}_i$ satisfy $v'<v$, then we say that  $(i,v',u_1,\dots,u_\tau)$ is dominated and delete it from $\mathcal{F}_i$. 
 
The overall running time of \textrm{AL}$_{2}$ can be bounded as follows. Recall that by Lemma~\ref{lemma:prepocessing}, $p_j\ge \epsilon$ and that the load of each machine is at most 2. Hence, each state in $\mathcal{F}_{i-1}$ can give rise to at most $(2/\epsilon+1)^{\tau}$ (dominated or undominated) states in $\mathcal{F}_i$. As also the number of jobs in each set $\mathcal{G}_h$ is bounded by $2m/\epsilon$, the number of undominated states in $\mathcal{F}_{i}$ is at most $(2m/\epsilon+1)^{\tau}=m^{\tilde{O}(\sqrt{1/\epsilon})}$. Hence, each set $\mathcal{F}_i$ can be constructed in time $(2/\epsilon+1)^{\tau}m^{\tilde{O}(\sqrt{1/\epsilon})}=m^{\tilde{O}(\sqrt{1/\epsilon})}$, which implies the same bound for the overall running of our dynamic programming algorithm. The lemma follows.
\end{proof}

\subsection{Algorithm 3}\label{appsubsec:al3}
The goal of this subsection is to prove the following lemma.
\begin{lemma}\label{lemma:bin-packing-scheduling}
Consider an instance after the preprocessing of Lemma~\ref{lemma:prepocessing}. For any $\epsilon>0$, there exists an algorithm that outputs a feasible schedule for well structured instance of $P||\sum_i C_i^q$ whose objective value is at most $OPT+O(\log^2 m)$ within $(1/\epsilon)^{O(1)}+n^{O(1)}$ time.
\end{lemma}

Recall Lemma~\ref{lemma:prepocessing}, and that an approximation scheme for well structured instances also implies an approximation scheme for general instances. Given Lemma~\ref{lemma:bin-packing-scheduling} and the fact that $OPT\ge m$ for well structured instances, we know that the additive error $O(\log^2 m) \le \epsilon OPT$ if $m=\Omega(1/\epsilon\log^2 (1/\epsilon))$. Hence, Theorem~\ref{thm:algorithm} is proved.

In the following, we present Algorithm $\ref{alg:3}$, which is the algorithm claimed in Lemma~\ref{lemma:bin-packing-scheduling}. Algorithm $\ref{alg:3}$ modifies upon the famous algorithm for bin packing by Karmarkar and Karp \cite{karmarkar1982efficient}. Given an instance $I$, $m(I)$ denotes the number of machines, and $\text{size}(I)$ denotes the total processing time of jobs. To exclude the trivial cases, if $I$ consists of less than $m(I)$ jobs, then it is obvious that the optimal solution assigns each job to a separate machine. Hence, in the following parts, we can assume without of loss generality there are more than $m(I)$ jobs.

We give a very high-level description. The scheduling problem can be interpreted as a bin packing problem where the bin number is a constraint, and the objective is to minimize the cost of bins instead of minimizing the number of bins (where the cost of a bin is its load to the power of $q$). Under such an interpretation, we are able to iteratively apply the {\it harmonic grouping scheme} to round job processing times and establish a {\it configuration LP} for the rounded instance. Based on the extreme point solution of the configuration LP, we assign jobs to roughly $m(I)/2$ machines and continue with the remaining jobs and machines. 


As discussed in Lemma \ref{lemma:prepocessing}, here we assume the processing time of all jobs in $I$ are within $[\epsilon,1]$. Suppose there are $\chi$ distinct job processing times with $b_1$ jobs of processing time $p_1$, $b_2$ jobs of processing time $b_2$, ... , $b_\chi$ jobs of processing time $p_\chi$. Consider the the subset of jobs that can be scheduled on a single machine, which can be characterized by a $\chi$-tuple $\vet_j=(t_{1j},t_{2j},\cdots,t_{\chi j})$ where $t_{ij}$ indicates the number of jobs of processing time $p_i$ on this machine. We call any $\vet_j$ with $t_{ij}\le b_i$ for every $i$ as a configuration. Let $N$ denote the number of configurations, let $\vet_1,\vet_2,\cdots,\vet_N$ be a complete enumeration of them. 

We establish a configuration integer program for instance $I$ as follows. We introduce a variable $x_j$ for each configuration $\vet_j$ which indicates the number of machines which is scheduled according to $\vet_j$. Consequently, all machines of configuration $\vet_j$ accommodate $t_{ij}x_j$ jobs of processing time $p_i$. We define the load, or total processing time of configuration $\vet_j$ as $L(\vet_j)=\sum_{i=1}^\chi t_{ij}p_i$. We define the cost of configuration $\vet_j$ as $v_j=(\sum_{i=1}^\chi t_{ij}p_i)^q$.

\begin{subequations}
\begin{eqnarray}
\text{Conf-IP($I$)}: &\min& \sum_{j=1}^N {v}_jx_j \nonumber\\
&s.t.& \sum_{j=1}^N x_j= m(I)  \label{eq:11}\\
&& \sum_{j=1}^N t_{ij}x_j\ge b_i, \quad i=1,2,\cdots, \chi \label{eq:2}\\
&& x_j\in\mathbb{N}, \quad j=1,2,\cdots, N \nonumber
\end{eqnarray}
\end{subequations}

Let ${OPT}_{IP}(I)$ be the optimal objective value of Conf-IP($I$). 
Relaxing the integral constraint $x_j\in\mathbb{N}$ to $x_j\ge 0$ in Conf-LP(${I}$), we obtain a configuration linear programming Conf-LP(${I}$). Let $OPT_{LP}({I})$ be its optimal objective value, it is obvious that $OPT_{LP}({I})\le OPT_{IP}({I})$. We have the following lemma.

\begin{lemma}\label{lemma:conf}
	There exists an algorithm of running time polynomial in $\chi, \log(m/\epsilon)$ and $1/\varsigma$, and returns a feasible extreme point solution of objective value at most $OPT_{LP}({I})+\varsigma$ for Conf-LP(${I}$). 
\end{lemma}

\begin{proof}
We consider the dual of Conf-LP(${I}$):
\begin{subequations}
\begin{eqnarray}
\text{Dual-LP(${I}$)}: &\max& \sum_{i=1}^\chi b_iy_i-m(I)z \nonumber\\
&s.t.& \sum_{i=1}^\chi t_{ij}y_i-z\le {v}_j, \quad j=1,2,\cdots, N \label{eq:dual1}\\
&& y_i\ge 0, \quad j=1,2,\cdots, N \nonumber
\end{eqnarray}
\end{subequations}

We use a similar algorithm as that for the classical bin packing problem. We employ the ellipsoid method to solve Dual-LP($\bar{I}$). We give a brief description. 
The ellipsoid method iteratively computes a sequence of ellipsoids $E_0,E_1,\cdots,$. In each iteration, it implements a separation oracle to check whether the center of the current ellipsoid $E_k$, say, $(\vey^k,z^k)=(y_1^k,\cdots,y_\chi^k,z^k)$, is feasible. If it is, then it outputs a cut $\veb\cdot\vey-m(I)z\ge \veb\cdot\vey^k-m(I)z^k$ where $\veb=(b_1,\cdots,b_{\chi})$; Otherwise, it finds out a violating constraint, say, $\vet_j\cdot \vey-z\le\bar{v}_j$, and outputs $\vet_j\cdot\vey-z\le \vet_j\cdot \vey^k-z^k$. Incorporating the cut output by the separation oracle, the ellipsoid method computes a new ellipsoid $E_{k+1}$ and guarantees that the volume of the new ellipsoid is smaller than $E_k$ by a factor of $e^{-\frac{1}{5(\chi+1)}}$. After a polynomial number of iterations (specifically, which is polynomial in $\log 1/\varsigma$), the ellipsoid method finds a near-optimal feasible solution with an additive error of $\varsigma$.

In their seminal work, Karmarkar and Karp \cite{karmarkar1982efficient} further prove that to compute an approximate solution to Conf-LP(${I}$) up to an additive precision of $\varsigma$, it suffices to construct an approximate separation oracle such that in each iteration, instead of checking whether the center $(\vey^k,z^k)$ is feasible and returns a violating constraint if it is infeasible, the approximate separation oracle checks a point $(\tilde{\vey}^k,\tilde{z}^k)$ and does the following:
\begin{itemize}
\item If $(\tilde{\vey}^k,\tilde{z}^k)$ violates a constraint, say, $\vet_j\cdot\tilde{\vey}^k-\tilde{z}^k> {v}_j$, then outputs cut $\vet_j\cdot\vey-{z}\le\vet_j\cdot\vey^k-\tilde{z}^k$;
\item If $(\tilde{\vey}^k,\tilde{z}^k)$ does not violate any constraint, then outputs
$\veb\cdot{\vey}-m(I)z\ge \veb\cdot\tilde{\vey}^k-m(I)\tilde{z}^k$.
\end{itemize}

Karmarkar and Karp showed that ellipsoid method equipped with the approximate separation oracle can return a near-optimal solution within an additive error of $O(\varsigma)$ as long as the point $(\tilde{\vey}^k,\tilde{z}^k)$ in each iteration satisfies that 
\begin{itemize}
\item For any constraint, if $\vet_j \cdot\tilde{\vey}^k-\tilde{z}^k>{v}_j$, then $\vet_j \cdot{\vey}^k-{z}^k>{v}_j$;
\item For the objective function, $\veb\cdot\vey^k-m(I)z^k\le \veb\cdot\tilde{\vey}^k-m(I)\tilde{z}^k+\varsigma$.
\end{itemize}

Here the first property ensures that if $(\tilde{\vey}^k,\tilde{z}^k)$ is infeasible, then $({\vey}^k,{z}^k)$ is also infeasible by violating the same constraint, and therefore the approximate separation oracle proceeds exactly the same as an (accurate) separation oracle. The second property ensures that the approximate separation oracle will never cut off a feasible point whose objective value is significantly better than $(\tilde{\vey}^k,\tilde{z}^k)$ by $\varsigma$, and hence ensures the near-optimality. 

Now we describe our approximate separation oracle as follows. We first round $v_j$ {\it up} to be the nearest value of the form $(1+\epsilon)^k$ and let it be $\bar{v}_j$. Notice that there are a polynomial number of distinct rounded values. Given an arbitrary point $(\vey^k,z^k)$, we consider all inequalities of the form 
$$\vet_j\vey^k\le \bar{v}_j+z^k.$$ 

Our goal is to find out a violating constraint or determine there is none. Since there are only a polynomial number of different values for $\bar{v}_j$'s, we can sequentially check for every value $(1+\epsilon)^h$, whether $(\vey^k,z^k)$ violates the constraint $\vet_j\vey^k\le (1+\epsilon)^h+z^k$ for all configurations $\vet_j$ such that $(1+\epsilon)^{h-1}<(\vet_j\vep)^q\le (1+\epsilon)^h$ where $\vep=(p_1,\cdots,p_\chi)$. We argue that we can drop the lower bound by sequentially checking for every value $(1+\epsilon)^h$, whether $(\vey^k,z^k)$ violates the constraint $\vet_j\vey^k\le (1+\epsilon)^h+z^k$ for all configurations $\vet_j$ such that $(\vet_j\vep)^q\le (1+\epsilon)^h$. 
This is because that if $(\vey^k,z^k)$ violates $\vet_j\vey^k\le (1+\epsilon)^h+z^k$ but $(\vet_j\vep)^q\le (1+\epsilon)^{h-1}$, say, $(\vet_j\vep)^q$ rounded up to $(1+\epsilon)^{h'}$ for $h'<h$, then $(\vey^k,z^k)$ also violates $\vet_j\vey^k\le (1+\epsilon)^{h'}+z^k$, which will be found out already. Hence, finding a violating constraint is equivalent as finding a vector $\vet\le \veb$ such that $(\vet\cdot\vep)^q\le (1+\epsilon)^h$ and $\vet\cdot\vey^k$ is maximized, and comparing this maximal value with $z^k+(1+\epsilon)^k$. This is a knapsack problem which admits a fully polynomial time approximation scheme (FPTAS). More precisely, using the same method as Karmarkar and Karp~\cite{karmarkar1982efficient}, we can round $\vey^k$ to some value $\tilde{\vey}^k$ close enough such that 
\begin{itemize}
\item For any $\chi$-dimensional vector $\ved$ whose coordinates are non-negative and $\|\ved\|_{1}\le n$, $0\le \ved\cdot {\vey}^k-\ved\cdot \tilde{\vey}^k\le \varsigma$;
\item In polynomial time (specifically, polynomial in $1/\varsigma$), we are able to find $\vet^*$ such that by taking $\vet=\vet^*$, $\vet\cdot\tilde{\vey}^k$ is maximized subject to $(\vet\cdot\vep)^q\le (1+\epsilon)^h$.
\end{itemize}

Overall, our above argument ensures that in polynomial time we either determine some configuration $\vet_j$ such that $\vet_j\tilde{\vey}^k> \bar{v}_j+z^k$ for some $j$, and hence 
$$\vet_j{\vey}^k+z^k\ge \vet_j\tilde{\vey}^k \bar{v}_j+z^k>\bar{v}_j+z^k\ge v_j+z^k;$$
or we conclude there is no such configuration and guarantee that $\veb\cdot\vey^k-m(I)z^k\le \veb\cdot\tilde{\vey}^k-m(I)\tilde{z}^k+\varsigma$.

Hence, there exists an approximate separation oracle for Dual-LP($I$), indicating that Dual-LP($I$) can be solved using the ellipsoid method in polynomial time (up to arbitrary precision). Notice that the derived solution may not necessarily be an extreme point, however, by using exactly the same argument as that of Karmarkar and Karp \cite{karmarkar1982efficient}, we can make it into an extreme point solution.
\end{proof}

 
Consider the near-optimal extreme point solution $\vex(I)$ given by Lemma~\ref{lemma:conf}. We denote by $\lfloor \vex(I)\rfloor=(\lfloor x_1(I)\rfloor,\cdots, \lfloor x_N(I)\rfloor)$ the rounded solution. Similar as the algorithm for bin packing, we assign jobs to machine according to $\lfloor \vex(I)\rfloor$ and then proceed with the residue instance $I_{res}(\vex)$. Denote by $I_{res}(\vex)$ the {\it residue instance} where we take away jobs scheduled according to $\lfloor \vex(I)\rfloor$ from the original instance $I$, i.e., $I_{res}$ consists jobs in $\vex(I)-\lfloor \vex(I)\rfloor$. It is easy to see $I_{res}(\vex)$ consists of $m(I_{res}(\vex))=m(I)- \sum_j \lfloor x_j(I) \rfloor$ machines. 

\begin{lemma}
	$OPT_{LP}(I_{res}(\vex)) + OPT_{LP}(I\setminus I_{res}(\vex)) \leq OPT_{LP}(I)$
\end{lemma}

\begin{proof}
	We know each configuration of instance $I$ is also a configuration of instance $I_{res}$ and $I\setminus I_{res}(\vex)$. Hence, for any solution $\vex(I)$ of instance $I$, $\lfloor \vex(I) \rfloor$ is an feasible solution of $I\setminus I_{res}(\vex)$ and $\vex(I)-\lfloor \vex(I) \rfloor$ is an feasible solution of $I_{res}(\vex)$.
\end{proof}

Similar to the bin packing algorithm by Karmarkar and Karp \cite{karmarkar1982efficient}, we will employ the {\it harmonic grouping scheme} to round the instance and apply Lemma~\ref{lemma:conf} on the rounded instance. It has to be noticed that the processing time of every job in the instance is no more than 1 which is guaranteed by Lemma~\ref{lemma:prepocessing}. The harmonic grouping works as follows: We deal with the job one by one in non-decreasing order of its processing time and pack the job into the current group. At any time, only one group is open. When the total processing time of jobs in the current group is at least 2, we close it and start a new group. By doing this, all jobs in $I$ are packed into $r$ groups i.e., $G_1,\dots,G_r$. We discard all jobs in $G_1$ along with $|G_{i-1}|-|G_i|$ jobs with smallest processing time in $G_i$ for each $i\in[2,r]$ and let it be $I_d$. For those remaining jobs, we lift the processing time to the largest one among their group and let it be $I'$. The {\it harmonic grouping scheme} has the following property \cite{williamson2011design}:
\begin{itemize}
	\item the number of distinct job processing times in $I'$ is at most  $\text{size}(I)/2$;
	\item $\text{size}(I_d) \in O(\log \text{size}(I))$.
\end{itemize}

\begin{lemma}
	$OPT_{LP}(I')\leq OPT_{LP}(I)$
\end{lemma}

\begin{proof}
	Given any solution $\vex(I)$ of instance $I$, for each configuration we can replace each job $j$ in group $G_i$ with a job $j'$ in group $G_{i+1}$. In such a modified solution, all jobs of $I'$ are scheduled. Since $p_{j'}\leq p_j$ holds, the objective does not increase.
\end{proof}

Now we formally present Algorithm $\ref{alg:3}$. Given an instance $I$ of the scheduling problem,  we first apply the harmonic grouping scheme to obtain the rounded instance $I'$ along with another instance $I_d$ composed by the discarded jobs. Then we apply Lemma~\ref{lemma:conf} to derive a feasible solution $\vex(I')$ for $I'$ and assign jobs to machines according to $\lfloor\vex(I')\rfloor$ and close these machines. The remaining jobs $\vex(I')-\lfloor\vex(I')\rfloor$ and the remaining empty machines $m(I')-\|\lfloor\vex(I'_j)\rfloor\|_1$ (here $\|\cdot\|_1$ is the 1-norm, which counts the number of integral configurations) forms a new instance $I'_{res}$. In the next iteration, $I'_{res}$ servers as the input and we repeat this process until there are only constant machines left. For the instance with constant machines, we call Algorithm $\ref{alg:1}$. Then we do the balancing operation to make sure that for every two machines their load difference is at most 1. At last, we group all the discarded jobs into $O(\log^2m)$ groups where the total processing time of each group is at most 2. This can be easily done, since the processing time of each job is at most 1. We pick arbitrarily $O(\log^2m)$ machines and schedule each group of jobs on one machine.

\begin{algorithm}[tb]
    \setcounter{algorithm}{2}
	\caption{Pseudo Code Description of AL$_3$}
	\label{alg:3}
	
	\begin{algorithmic}[1]
		\STATE{$I_D=\emptyset$}
		\WHILE{$m(I)> 10$}
			\STATE{apply harmonic grouping scheme to create instance $I'$ and $I_d$}
			\STATE{$I_D=I_D\cup I_d$}
			\STATE{let $\hat{\vex}(I')$ be the $(\epsilon/\log m)$-balanced solution for Conf-IP($I'$)}
			\STATE{schedule $I'\setminus I'_{res}(\hat{\vex})$ according to $\lfloor \hat{\vex}(I') \rfloor$}
			\STATE{$I=I'_{res}(\hat{\vex})$}
		\ENDWHILE
		\STATE{schedule $I$ using Algorithm \ref{alg:1}}
		\WHILE{$\exists i,j$ s.t. $L_i-L_j>1$}
		\STATE{move the job with the largest processing time in machine $i$ to machine $j$}
		\ENDWHILE
		\STATE{schedule $I_D$ on arbitrarily $O(\log^2m)$ machines with max increasing load 2}
	\end{algorithmic}
\end{algorithm}

Finally, we estimate the overall loss incurred. In each iteration, only $m(I)/2+1$ variables take non-zero value in solution $I'_{res}(\hat{\vex})$. Hence, $m(I'_{res}(\hat{\vex}))\leq m(I)/2+1$ and after at most $O(\log m)$ rounds Algorithm \ref{alg:3} terminates. Each iteration introduces an additional cost of $\varsigma$, together with discarded jobs of total processing time $O(\log^2m)$, which need to be handled at last. Set $\varsigma=O(\log m)$ and observing that $OPT\ge OPT_{LP}(I)$ and $OPT\ge m$, we know the overall additional cost is bounded by $O(\log^2m)$. Now consider all the discarded jobs. Through Lemma~\ref{lemma:prepocessing} and the balancing operation, we know the load of machines in the solution is at most 3. Meanwhile, due to the convexity of the objective, the balancing operation does not increase the objective value. Given that $q$ is a constant, hence, the overall objective value increases by at most $O(\log^2 m)$, and Lemma~\ref{lemma:bin-packing-scheduling} is proved.

\section{Omitted Proofs in Section~\ref{subsec:maxsat} - Proof of Lemma~\ref{lemma:maxsat-eth2}}\label{appsec:sat}
The goal of this section is to prove the following lemma.
\begin{T3}
Assuming ETH, there exists a constant $\beta\in(0,1)$ such that for any sufficiently small $\epsilon',\delta>0$, it is not possible to distinguish between instances of 3SAT$\,'$ with $(1-\epsilon')\cdot 4n/3 $ clauses where at least $4n/3$ clauses are satisfiable, from instances where at most $(\beta+\epsilon')\cdot 4n/3$ clauses are satisfiable, in time $2^{O(n^{1-\delta})}$.
\end{T3}

Towards the proof, we start with the following result (see, e.g., Corollary 1 of~\cite{bonnet2015subexponential}).
\begin{lemma}\cite{moshkovitz2010two,bonnet2015subexponential}\label{lemma-cite:1}
Under ETH, for sufficiently small $\epsilon'>0$, and $\delta>0$, it is impossible to distinguish between instances of 3SAT with $\Lambda$ clauses where at least $(1-\epsilon')\Lambda$ are satisfiable from instances where at most $(7/8+\epsilon')\Lambda$ are satisfiable, in time $O(2^{\Lambda^{1-\delta}})$.
\end{lemma}

Applying the classical technique of constructing enforcer via expander for 3SAT (see, e.g., Theorem 5 of \cite{trevisan2004inapproximability}), we have the following,
\begin{lemma}\cite{trevisan2004inapproximability}\label{lemma-cite:2}
There exists a constant $d_0$ such that given a 3SAT formula $\phi$ with $\Lambda$ clauses, another 3SAT formula $\phi'$ with $\Lambda'=\Lambda+3d_0\Lambda=O(\Lambda)$ clauses can be constructed in polynomial time such that:
\begin{itemize}
\item Every variable occurs in at most $2d+1$ clauses in $\phi'$;
\item There is an assignment for $\phi$ where at most $k$ clauses are not satisfied if and only if there is an assignment for $\phi'$ such that at most $k$ clauses are not satisfied.
\end{itemize}
\end{lemma}

Denote by 3SAT-$d$ the 3SAT problem where every variable occurs at most $d$ times. Combining Lemma~\ref{lemma-cite:1} and Lemma~\ref{lemma-cite:2}, we have the following lemma.
\begin{lemma}\label{lemma:maxsat-eth}
Under ETH, there exists some constants $d\in\mathbb{N}$ and $\alpha\in(0,1)$ such that for sufficiently small $\epsilon'>0$, and $\delta>0$, it is impossible to distinguish between instances of 3SAT-$d$ with $\Lambda$ clauses where at least $(1-\epsilon')\Lambda$ are satisfiable from instances where at most $(\alpha+\epsilon')\Lambda$ are satisfiable, in time $O(2^{\Lambda^{1-\delta}})$.
\end{lemma}

It is worth mentioning that the reduction in~\cite{trevisan2004inapproximability} involves constructing 2-clauses, that is, 3SAT-$d$ in Lemma~\ref{lemma:maxsat-eth} refers to a 3SAT instance where clauses may contain 2 or 3 variables. For ease of presentation, we want to enforce every clause to contain exactly 3 variables\footnote{We remark, however, that our reduction also works if $C_1$ contains 2-causes and 3-clauses. It suffices to create two CL$_\ell$, one true copy and one false copy instead of three, and meanwhile adjust the number of dummy jobs.}. This can be done by introducing dummy variables together with 3-clauses that enforce a dummy variable to be true or false (called enforcers). In particular, Berman et al.~\cite{berman2003approximation} provide a general enforcer that allows them to deduce the APX-hardness of MAX3SAT (where every clause contains 3 variables and every variable appears 4 times) through the APX-hardness of MAX2SAT. We can apply their technique directly to get a strengthened version of Lemma~\ref{lemma:maxsat-eth} where in 3SAT-$d$ every clause contains exactly 3 variables.


The following proof is a slight variation of that from Tovey~\cite{tovey}. 

\begin{lemma}\label{le:tovey}
Given a 3SAT-$d$ formula $\phi$ with $\Lambda$ clauses, a 3SAT$\,'$ formula $\phi'$ with $|C_1|=\Lambda$ and $|C_2|\le 3\Lambda$ clauses can be constructed in polynomial time such that:
\begin{itemize}
\item If there is an assignment for $\phi$ where at most $k$ clauses are not satisfied, then there is an assignment for $\phi'$ where at most $k$ clauses are not satisfied. 
\item If there is an assignment for $\phi'$ where there are at most $k$ clauses not satisfied, then there is an assignment for $\phi$ where at most $kd$ clauses are not satisfied.
\end{itemize}
\end{lemma}
\begin{proof}
Let $z$ be any variable in $\phi$ and suppose it appears $\ell\le d$
times in clauses. If $\ell=1$ then we add a dummy clause $(z\oplus \neg
z)$. Otherwise $\ell\ge 2$ and we introduce $\ell$ new variables $z_1$,
$z_2$, $\cdots$, $z_\ell$ and $\ell$ new clauses $(z_1\oplus \neg z_2)$,
$(z_2\oplus \neg z_3)$, $\cdots$, $(z_{\ell}\oplus \neg z_{1})$ which enforce $z_1,z_2,\cdots,z_\ell$ to take the same truth value. Meanwhile
we replace the $\ell$ occurrences of $z$ in the original clauses by $z_1$, $z_2$, $\cdots$,
$z_\ell$ in turn and remove $z$. By doing so we transform $\phi$
into a new formula $\phi'$ by introducing at most $3\Lambda$ new variables and $3\Lambda$
new clauses.

Notice that each new clause we add in $\phi'$ is of the form
$(z_i\oplus \neg z_{i'})$. We let $C_2$ be the set of them and let $C_1$
be the set of other clauses. It is easy to verify that $\phi'$ is
an instance of 3SAT$'$. Notice that every clause in $C_1$ has a corresponding clause in $\phi$ by replacing $z_i$'s with $z$.

Suppose there is an assignment for $\phi$ where at most $k$ clauses are not satisfied. Then for any variable $z$ in $\phi$ that occurs $\ell$ times, we let $z_1$, $z_2$, $\cdots$, $z_\ell$ all take the same value as $z$. It is easy to see that at most $k$ clauses in $C_1$ of $\phi'$ are not satisfied.

Suppose there is an assignment for $\phi'$ where at most $k$ clauses in $C_1$ are not satisfied. For any variable $z$ in $\phi$ that correspond to $z_1$, $z_2$, $\cdots$, $z_\ell$ in $\phi'$, we let $z_i$ take the same value of $z_1$ for all $i$. Now we check the number of additional unsatisfied clauses in $C_1$ we introduce by doing so. If all $z_i$'s take the same value, then no additional unsatisfied clauses are introduced. Otherwise, it is possible that some of the clauses in $C_1$, which is satisfied by $z_2$, $z_3$, $\cdots$ or $z_\ell$, becomes unsatisfied. But there are at most $d-1$ such kind of clauses. Hence, at most $d-1$ unsatisfied clauses are introduced, if there is at least one unsatisfied clause among $(z_1\oplus \neg z_2)$, $(z_2\oplus \neg z_3)$, $\cdots$, $(z_{\ell}\oplus \neg z_{1})$. This implies that we have introduced at most $k(d-1)$ unsatisfied clauses by setting $z_i=z_1$ for all variables, i.e., there are at most $k(d-1)+k=kd$ unsatisfied clauses in $C_1$ now. Hence, there is an assignment for $\phi$ where at most $kd$ clauses are not satisfied.
\end{proof}

Given Lemma~\ref{le:tovey}, we know that if there is an assignment for $\phi$ with are at most $\alpha \Lambda$ unsatisfied clauses, then there is an assignment for $\phi'$ with every clause in $C_2$ satisfied and at most $\alpha \Lambda$ clauses in $C_1$ unsatisfied; if every assignment has at least $\beta \Lambda$ unsatisfied clauses in $\phi$, then there are at least $\beta \Lambda/d$ unsatisfied clauses in $\phi'$. Hence, according to Lemma~\ref{lemma:maxsat-eth}, Lemma~\ref{lemma:maxsat-eth2} is proved.

\section{Omitted Contents in Section~\ref{subsec:overview} - Why Old Reduction does not Work}\label{appsec:old}
Chen et al.~\cite{chen2014optimality} provided a reduction that meets the conditions CO1 to CO4 with job processing times, and hence the target value $T$, being $O(n^{1+\delta})$ for any arbitrary small constant $\delta>0$. This reduction provides a strong lower bound for $P||C_{max}$, but does not work well for our problem $P||\sum_iC_i^q$. To see this, we take $q=2$ as an example and compare the two objective values for the constructed scheduling problem when $I_{sat}$ is satisfiable and when it is not. If $I_{sat}$ is satisfiable, then there exists a schedule such that every machine has a load of exactly~$T$, implying that the optimal objective value is $mT^2$. Otherwise, at least one machine has load $T+1$ or more and machine has load $T-1$ or less, and then the optimal objective is at least $(m-2)T^2+(T+1)^2+(T-1)^2=mT^2+2$. For the sake of contradiction, let us assume that there exists a PTAS with running time $2^{O((1/\epsilon)^{\kappa})}$ for some $\kappa\in (0,1)$. If we take $\epsilon$ to be sufficiently small such that $mT^2\epsilon\le 1$, then the PTAS can be used to determine whether the constructed scheduling instance admits a feasible schedule of objective value at most $mT^2+1<mT^2+2$, and hence whether $I_{sat}$ is satisfiable. The running time of the PTAS becomes $2^{O((1/\epsilon)^{\kappa})}=2^{O((mT^2)^{\kappa})}$. Plug in $m=O(n)$ and $T=O(n^{1+\delta})$ in the reduction, we have $mT^2=O(n^{3+2\delta})$. Hence, if $\kappa=1/3-\delta$, we have $(mT^2)^{\kappa}\le O(n^{1-\delta})$, and an efficient PTAS of running time $2^{O((1/\epsilon)^{1/3-\delta})}$ can thus determine the satisfiability of $I_{sat}$ in $2^{O(n^{1-\delta})}$ time, contradicting ETH. To summarize, the above argument implies a lower bound of $2^{O((1/\epsilon)^{1/3-\delta})}$ on the running time of PTAS for arbitrary constant $\delta>0$, which is not strong enough to match our algorithms in Theorem~\ref{thm:algorithm}.

To overcome the obstacle, a natural idea is to decrease the value of $m$ or $T$ in the reduction. However, if, say, $m=n^{0.9}$ and $T=n^{O(1)}$, then we know the standard dynamic programming for scheduling returns the optimal solution in $T^{O(m)}=2^{O(n^{0.9})}$ time; similarly, if $T=n^{0.9}$ and $m=n^{O(1)}$, then we know there are at most $n^{0.9}$ different kinds of jobs, and the scheduling problem can also be solved in time $2^{O(n^{0.9})}$ through dynamic programming. Hence, we cannot expect to reduce $I_{sat}$ to such scheduling instances, assuming ETH. 

As a consequence, in this paper, we will not try to decrease $m$ or $T$. Instead, we increase the gap between the two optimal objective values for the constructed scheduling problem when $I_{sat}$ is satisfiable and when $I_{sat}$ is not satisfiable by exploiting the hardness gap in Lemma~\ref{lemma:maxsat-eth2}. 

\section{Omitted Proofs in Section~\ref{subsec:adjacent-sum} - Proof of Lemma~\ref{lemma:uniquesum-1}}\label{appsec:ajacent-sum}
\begin{T2}
Let $N\in\mathbb{Z}^+$. There exists a subset ${\cal{S}}\subseteq \Z_N$ such that $|{\cal{S}}|\ge N^{1-c_0\sqrt{\frac{1}{{\log N}}}}$ for some sufficiently large $c_0$ (in particular, $c_0=7$ suffices), and for any $y\in {\cal{S}}$ and $1\le h\le 5$, the linear equation $h\cdot y=y_1+y_2+\cdots+y_h$ with $y_i\in {\cal{S}}$ for all $i$ has a unique solution $y_1=y_2=\cdots=y_h=y$.
\end{T2}
\begin{proof}
For any $d\ge 2$, $M\ge 2$ and $k\le (M+1)(d-1)^2$, we let $x=5d-1$ and $\vex=(1,x,x^2,\cdots,x^M)$, $\vecc=(\vecc[0],\vecc[1],\cdots,\vecc[M])$. We define the set $S_k(M,d)$ as:
\begin{eqnarray*}
{\cal{S}}_k(M,d)&=&\{y: y=\vecc\vex+x^{M+1}=\vecc[0]+\vecc[1]x^1+\cdots+\vecc[M]x^M+x^{M+1},\\ &&\vecc[i]\in\mathbb{N}, 0\le \vecc[i]<d, \sum_{i=0}^M(\vecc[i])^2=k\}.
\end{eqnarray*}
That is, ${\cal{S}}_k(M,d)$ is the set of all integers which can be expressed in the form of ${\vecc}\cdot{\vex}+x^{M+1}$ such that $\vecc[i]\in [0,d)$ and $\|\vecc\|_2^2=k$, where $\|\cdot\|_2$ is the $\ell_2$-norm of a vector. 

We claim that for any $y\in {\cal{S}}_k(M,d)$ and $1\le h,h'\le 5$, if $hy=y_1+y_2+\cdots+y_{h'}$, $y_i\in S_k(M,d)$, then we have $h'=h$, $y_1=y_2=\cdots=y_h=y$. Let $y=\vecc\vex+x^{M+1}$ and $y_j=\vecc_j\vex+x^{M+1}$, where $\vecc_j=(\vecc_j[0],\vecc_j[2],\cdots,\vecc_j[M])$. 
Using the fact that $x=5d-1$ and $\vecc_j[i]< d$, we know
for $h'\le 5$ we have $\sum_{j=1}^{h'}\vecc_j[i]<x$. Hence by checking the coefficient of $x^{M+1}$ on both sides of the equation $hy=\sum_{j=1}^{h'}y_j$, we have $h=h'$. 
Moreover, we can conclude that the coefficient of $x^i$ in the sum $y_1+y_2+\cdots+y_{h'}$ equals $\sum_{j=1}^{h'}\vecc_j[i]$, hence by comparing the coefficient of $x^i$ on both sides, we have  
$$\sum_{j=1}^h\vecc_j=h\vecc.$$ 

By the definition of ${\cal{S}}_k(M,d)$, the followings are true:
$$\sum_{i=0}^{M}(\vecc_j[i])^2=k, \quad \forall j,$$
and 
$$\sum_{i=0}^{M}(\frac{\sum_{j=1}^h \vecc_j[i]}{h})^2=\sum_{i=0}^{M}(\vecc[i])^2=k$$
Hence,
\begin{eqnarray}\label{eq:ari}
 k=\sum_{i=0}^{M}(\frac{\sum_{j=1}^h \vecc_j[i]}{h})^2=\frac{\sum_{j=1}^h\sum_{i=0}^{M}(\vecc_j[i])^2}{h}=\sum_{i=0}^{M}\frac{\sum_{j=1}^h(\vecc_j[i])^2}{h}. 
\end{eqnarray}

According to the inequality between the quadratic mean and arithmetic mean, we know
$$ \left(\frac{\sum_{j=1}^h \vecc_j[i]}{h}\right)^2\le \frac{\sum_{j=1}^h (\vecc_j[i])^2}{h}, \quad \forall i$$
and the equality only holds when $c_j[i]$'s are identical for all $j$. Hence by Eq~\eqref{eq:ari} we know $\vecc=\vecc_j$, and consequently
$y_j=y$ for $1\le j\le h$. Hence, the claim is true. 

It remains to select an appropriate ${\cal{S}}_k(M,d)$ such that ${\cal{S}}_k\subseteq \Z_N$ and has a large cardinality. Towards this, we first observe that the largest number of ${\cal{S}}_k(M,d)$ is bounded by $(5d-1)^{M+2}$. We shall select $d$ and $M$ such that $(5d-1)^{M+2}\le N$. Notice that there are $d^{M+1}-1$ different positive integer numbers which can be expressed as $\vecc\cdot\vex+x^{M+1}$ where $\vecc[i]\in [0,d)$. Furthermore $k=\|\vecc\|_2^2\le (M+1)(d-1)^2$, hence there exists some $1\le k^*\le (M+1)(d-1)^2$ such that
$$|{\cal{S}}_{k^*}({M,d})|\ge \frac{d^{M+1}-1}{(M+1)(d-1)^2}>\frac{d^{M-1}}{M+1}.$$

It remains to select $d$ and $M$ subject to $(5d-1)^{M+2}\le N$ such that $\frac{d^{M-1}}{M+1}$ is large. Below all logarithms are taken with the base $e$. We pick $M=\lfloor \sqrt{\frac{\log N}{\log 5}}\rfloor-2$ and $d=\lfloor\frac{e^{\sqrt{\log 5\cdot \log N}}}{5}\rfloor$. It is easy to see that $(M+2)\log (5d-1)\le \sqrt{\frac{\log N}{\log 5}}\cdot \log e^{\sqrt{\log 5\cdot \log N}}=\log N$, hence $(5d-1)^{M+2}\le N$. Furthermore,
for sufficiently large $N$ (e.g., $N>e^{10}$), we know 
$d\ge \frac{e^{\sqrt{\log 5\cdot \log N}}}{10}$, hence
\begin{eqnarray*}
    \frac{d^{M-1}}{M+1}=e^{(M-1)\log d -\log (M+1)}&\ge& e^{(\sqrt{\log N\cdot\log 5}-\log 10)(\sqrt{\frac{\log N}{\log 5}}-4)-\frac{1}{2}(\log\log N-\log\log 5)} \\&\ge&e^{\log N\cdot (1+\frac{-4\sqrt{\log N\log 5}-\frac{\log 10}{\sqrt{\log 5}}\sqrt{\log N}-\Omega(\log\log N)}{\log N})}\\&=&N^{1-\Omega(\frac{1}{\sqrt{\log N}})}
\end{eqnarray*}

Hence, Lemma~\ref{lemma:uniquesum-1} is proved. In particular, it is easy to verify that $4\sqrt{\log 5}+\frac{\log 10}{\sqrt{\log 5}}\le 7$, hence the $\frac{d^{M-1}}{M+1}\ge N^{1-\frac{7}{\sqrt{\log N}}}$.
\end{proof}

\section{Omitted proofs in 
Section~\ref{subsec:adjacent-sum} - Proof of Lemma~\ref{lemma:linked-sum}}
\label{appsubsec:shuffle}

The goal of this subsection is to prove the following.
\begin{T6}
For any $\pi\in Aut({\cal{M}}^{\upsilon})$, there exist $h$-shufflers $f_h,\hat{f}_h\in Aut({\cal{M}}^{\upsilon})$ for every $1\le h\le \upsilon$ such that $F(\vey)=\hat{F}(\pi(\vey))$ for any $\vey\in {\cal{M}}^{\upsilon}$, where $F=f_{\upsilon-1}\circ f_{\upsilon-2}\circ \cdots\circ f_{0}$ and $\hat{F}=\hat{f}_{\upsilon-1} \circ \hat{f}_{\upsilon-2}\circ \cdots\circ \hat{f}_{0}$. Furthermore, $f_h$'s and $\hat{f}_h$'s can be constructed in time that is polynomial in $|{\cal{M}}^\upsilon|$.
\end{T6}

For any finite set $X$, we denote by $Aut(X)$ the set of all one-to-one mapping from $X$ to itself. For any $f_1,f_2\in Aut(X)$, we denote by $f_1\circ f_2 \in Aut(X)$ the composition of $f_1$ and $f_2$, i.e., $f_1\circ f_2(x)=f_1(f_2(x))$ for any $x\in X$. Note that $Aut(X)$ is a symmetric group under composition. We denote by $f^{-1}\in Aut(X)$ the inverse of $f\in Aut(X)$.

Let ${\cal{M}}$ be an arbitrary finite set of cardinality $t$. 
Let $\upsilon\in \Z^+$. Denote by ${\cal{M}}^{\upsilon}$  the set of all $\upsilon$-dimensional vectors whose entries belong to ${\cal{M}}$.

For any vector $\vey\in {\cal{M}}^{\upsilon}$, we denote by $\vey[h]\in {\cal{M}}$ the $h$-th coordinate of $\vey$, and $\vey[-h]\in  {\cal{M}}^{\upsilon-1}$ the vector obtained by removing the $h$-th coordinate from $\vey$. 

For any $f\in Aut({\cal{M}}^{\upsilon})$ and $0\le h\le \upsilon-1$, we call $f$ an $h$-shuffler if $\left(f(\vey)\right)[-h]=\vey[-h]$ for all $y\in \cal{M}^\upsilon$, that is,
$$f(\vey)=f(\vey[0],\vey[1],\cdots,\vey[\upsilon-1])=(\vey[0],\cdots,\vey[h-1],z,\vey[h+1],\cdots,\vey[\upsilon-1])),$$
for some $z\in {\cal{M}}$.

We prove the following lemma.

\begin{lemma}\label{lemma:shuffle}
For any $\pi\in Aut({\cal{M}}^{\upsilon})$ and $0\le h\le \upsilon-1$, there exist $h$-shufflers $f,\hat{f}\in Aut({\cal{M}}^{\upsilon})$ such that for any $\vey\in {\cal{M}}^{\upsilon}$, $\left(f(\vey)\right)[h]=\left(\hat{f}(\pi(\vey))\right)[h]$. Furthermore, $f$ and $\hat{f}$ can be constructed in time that is polynomial in $|{\cal{M}}^\upsilon|$.
\end{lemma}

\begin{figure}[tb]
	\centering
	\includegraphics[height=89.94mm,width=135mm]{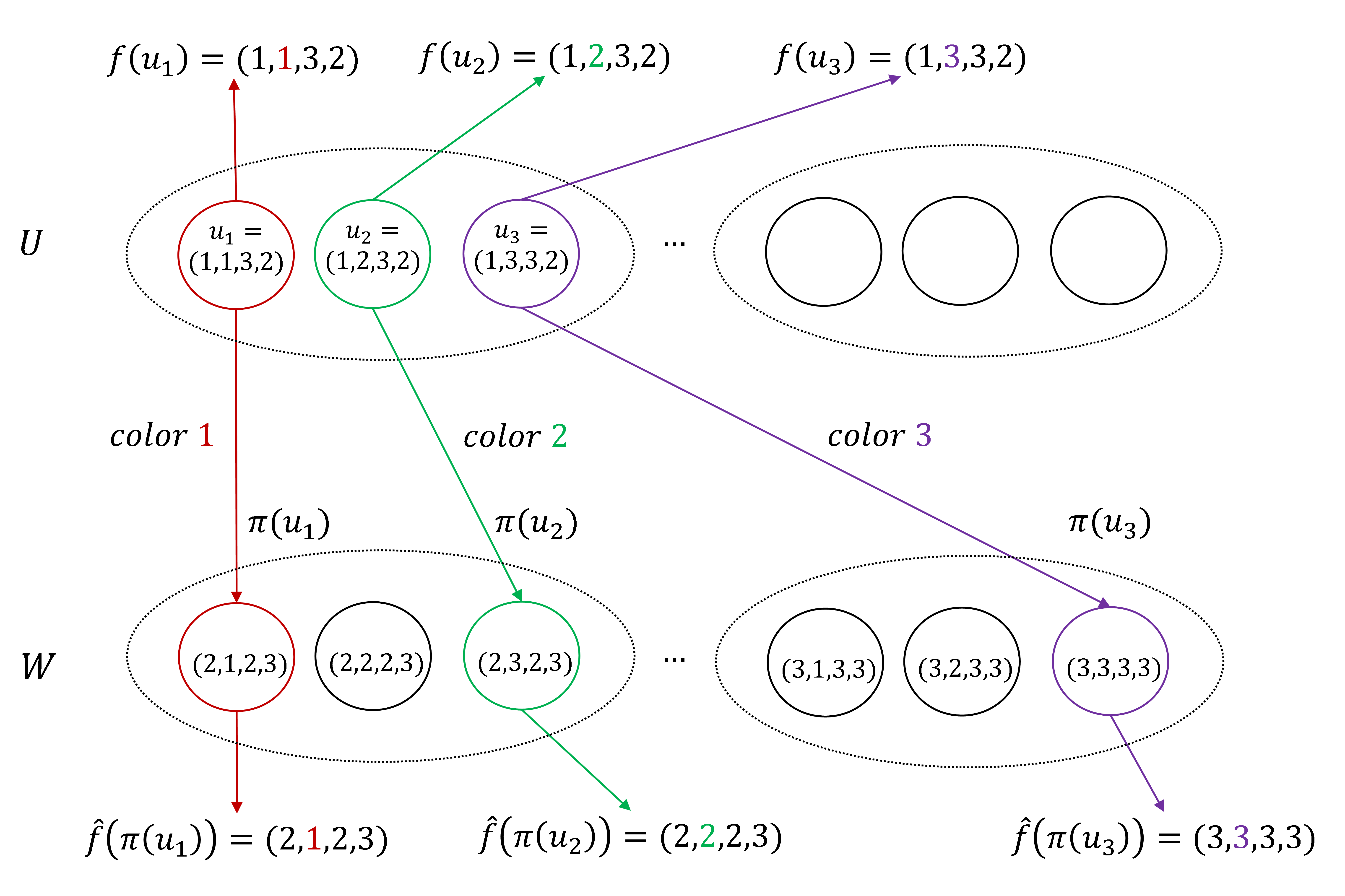}
	\caption{Illustration of Lemma \ref{lemma:shuffle} for $M=\{1,2,3\}$, $\upsilon=4$, $h=1$ (recall that vectors start with the $0$-th coordinate). Each vertex in $U$ (i.e., the solid circle) is a 4-dimensional vector. Each mega-vertex in $\bar{U}$ (i.e., the dotted circle) contains exactly 3 vertices. Each solid line between vertices represents the mapping $\pi$ which is colored by one of 3 colors.}
	\label{fig:shuffle_lemma}
\end{figure}

Briefly speaking, $f$ and $\hat{f}$ shuffles the $h$-coordinate of $\vey$ and its image under $\pi$ such that they become identical.

Now we are ready to prove Lemma~\ref{lemma:shuffle-1}.

\begin{proof}[Proof of Lemma~\ref{lemma:shuffle-1}]
For ease of presentation, we let ${\cal{M}}=\{1,2,\cdots,t\}$. 
Note that $|{\cal{M}}^{\upsilon}|=t^{\upsilon}$. Sort elements (vectors) of ${\cal{M}}^{\upsilon}$ in an arbitrary order and denote them by $\{\vecc_1,\vecc_2,\cdots,\vecc_{t^{\upsilon}}\}$. We create a bipartite graph $G=(U\cup W, E)$ to represent $\pi$ as follows: Both $U$ and $W$ contain $t^{\upsilon}$ vertices. Let $U=\{u_1,u_2,\cdots,u_{t^{\upsilon}}\}$ and $W=\{w_1,w_2,\cdots,w_{t^{\upsilon}}\}$. There is an edge between $u_i$ and $w_j$ if and only if $\pi(\vecc_i)=\vecc_j$. Since $\pi$ is one-to-one mapping, $G$ is 1-regular.

 As each $u_i$ and $w_i$ correspond to $\vecc_i$, we will slightly abuse notation and write $u_i[h]$ or $u_i[-h]$ to refer to $\vecc_i[h]$ and $\vecc_i[-h]$. 

\smallskip
\noindent\textbf{Contraction.} We contract the graph $G$ as follows. We partition $U$ (or $W$) into $t^{\upsilon-1}$ subsets such that $u_i$ and $u_j$ (or $w_i$ and $w_j$) are in the same subset if and only if $u_i[-h]=u_j[-h]$ (or $w_i[-h]=w_j[-h]$). Denote by $\bar{U}_i$ (or $\bar{W}_i$), $1\le i\le t^{\upsilon-1}$, all the subsets in the partition of $U$ (or $W$). It is clear that each  $\bar{U}_i$ (or $\bar{W}_i$) contains exactly $t$ vertices from $U$ (or $W$). We now contract all the $t$ vertices in  $\bar{U}_i$ (or $\bar{W}_i$) into one mega-vertex, and denote this mega-vertex as $\bar{u}_i$ (or $\bar{w}_i$). By doing so we generate parallel edges, that is, there are $\ell$ parallel edges between each pair of mega-vertices $\bar{u}_i$ and $\bar{w}_j$ if there are $\ell$ edges between vertices in  $\bar{U}_i$ and $\bar{W}_j$ in the original graph $G$. We denote by $\psi$ an arbitrary one-to-one mapping between a parallel edge (between mega-vertices $\bar{u}_i$ and $\bar{w}_j$) and an edge in $G$ (between some vertex in subset $\bar{U}_i$ and some vertex in subset $\bar{W}_j$).  
Denote by $\bar{G}=(\bar{U}\cup \bar{W},\bar{E})$ the contracted graph. Given that $G$ is $1$-regular and every mega-vertex contains exactly $t$ vertices, we have the following observation:
\begin{observation}
The contracted graph $\bar{G}=(\bar{U}\cup \bar{W},\bar{E})$ is a $t$-regular bipartite graph.
\end{observation}

\smallskip
\noindent\textbf{Coloring.} It is known that every bipartite regular graph admits a perfect matching (see, e.g.~\cite{konig1916graphen}). Consequently, every $t$-regular bipartite graph can be decomposed into $t$ perfect matchings. We decompose $\bar{G}$ into $t$ perfect matchings and color edges in each perfect matching with a distinct color. 
Overall we have used $t$ colors. Since $|{\cal{M}}|=t$, we can map the $i$-th color to integer $i\in {\cal{M}}$. 

Recall that $\psi$ is a one-to-one mapping between $\bar{E}$ and $E$, hence via $\psi$ we also obtain a coloring for $E$ (by coloring each edge in $E$ with the same color as its corresponding edge in $\bar{E}$). Recall that $G$ is 1-regular. Thus we can extend the edge coloring to a vertex coloring, such that each vertex in $G$ is colored with the same color as the unique edge incident to it. 

\smallskip
\noindent\textbf{Define functions $f$ and $\hat{f}$.} 
Consider every vertex set $\bar{U}_k$. We know $\bar{U}_k$ contains $t$ vertices, and let $\bar{U}_k=\{u_{k_1},u_{k_2},\cdots,u_{k_t}\}$.  By definition $u_{k_i}[-h]$'s are identical and  $u_{k_i}[h]$'s are exactly the $t$ elements in ${\cal{M}}$. Recall that we decompose $\bar{G}$ into $t$ perfect matchings and each perfect matching is colored with a unique color, we know the $t$ parallel edges incident to the mega-vertex $\bar{u}_k$ are colored with $t$ distinct colors. Consequently, each vertex $u_{k_i}$ is also colored with a distinct color. Recall the one-to-one correspondence between a vertex $u_j$ and $\vecc_j\in{\cal{M}}$. 
Now we define a function $f$ such that $f(\vecc_{k_i})[-h]=\vecc_{k_i}[-h]$, and $f(\vecc_{k_i})[h]$ equals the color of $u_{k_i}$, where we interpret each color as a number in $\{1,\cdots,t\}$. Consequently, $(f(\vecc_{k_1}),f(\vecc_{k_2}),\cdots,f(\vecc_{k_t}))$ is a permutation of $(\vecc_{k_1},\vecc_{k_2},\cdots,\vecc_{k_t})$. Hence, $f\in Aut({\cal{M}}^{\upsilon})$.

Similarly, we consider each $\bar{W}_k=\{w_{k_1},w_{k_2},\cdots,w_{k_t}\}$  and define a function $\hat{f}$ such that $\hat{f}(\vecc_{k_i})[-h]=\vecc_{k_i}[-h]$, and $\hat{f}(\vecc_{k_i})[h]$ equals the color of $w_{k_i}$. Consequently, $(\hat{f}(\vecc_{k_1}),\hat{f}(\vecc_{k_2}),\cdots,\hat{f}(\vecc_{k_t}))$ is also a permutation of $(\vecc_{k_1},\vecc_{k_2},\cdots,\vecc_{k_t})$, and $\hat{f}\in Aut({\cal{M}}^{\upsilon})$.

Furthermore, the color of each $u_i$ or $w_j$ is defined as the color of the edge incident to it, hence if there is an edge between $u_i$ and $w_j$ in $G$, then we know $(f(u_i))[h]=(\hat{f}(w_j))[h]$. Hence, Lemma~\ref{lemma:shuffle} is proved.
\end{proof}

See Figure~\ref{fig:shuffle_lemma} for an illustration of the mapping $f,\hat{f}$ we construct in Lemma~\ref{lemma:shuffle}. Iteratively applying Lemma~\ref{lemma:shuffle}, we are able to prove the following.

\begin{proof}
We prove the following statement by induction: For $0\le k\le \upsilon-1$, there exist $h$-shufflers $f_h,\hat{f}_h\in Aut({\cal{M}}^{\upsilon})$ for every $0\le h\le k$ such that $F_k(\vey)[h']=\hat{F}_k(\pi(\vey))[h']$ for any $\vey\in {\cal{M}}^{\upsilon}$ and $h'\le k$, where $F_k=f_{k}\circ f_{k-1}\circ \cdots\circ f_{0}$ and $\hat{F}_k=\hat{f}_{k} \circ \hat{f}_{k-1}\circ \cdots\circ \hat{f}_{0}$.

The statement is true for $k=1$ by Lemma~\ref{lemma:shuffle}. Suppose the statement is true for $k$, we prove it is true for $k+1$.

Consider $\hat{f}_{k}\circ \hat{f}_{k-1}\circ\cdots\circ \hat{f}_0\circ \pi\circ f_0^{-1}\circ f_1^{-1}\circ\cdots \circ f_k^{-1}$. According to Lemma~\ref{lemma:shuffle}, there exist $(k+1)$-shufflers $f_{k+1},\hat{f}_{k+1}$ such that 
\begin{eqnarray}\label{eq:induction}
\left( f_{k+1}(\vey)\right)[k+1]=\left(\left(\hat{f}_{k+1}\circ\hat{f}_{k}\circ \cdots\circ \hat{f}_0\circ \pi\circ f_0^{-1}\circ f_2^{-1}\circ\cdots \circ f_k^{-1}\right)(\vey)\right)[k+1], \quad\forall \vey\in {\cal{M}}^{\upsilon}.
\end{eqnarray}

Since $f_k\circ f_{k-1}\circ\cdots\circ f_1\in Aut({\cal{M}}^{\upsilon})$, for every $\vey\in {\cal{M}}^{\upsilon}$ there exists some $\vez\in {\cal{M}}^{\upsilon}$ such that $\vey=\left(f_k\circ f_{k-1}\circ\cdots\circ f_1\right)(\vez)$, plug this into Equation~\eqref{eq:induction}, for all $\vez\in {\cal{M}}^{\upsilon}$ we get 
\begin{eqnarray*}
&&\left(f_{k+1}\circ f_k\circ\cdots\circ f_0\right)(\vez)[k+1]\\
&=&\left(\left(\hat{f}_{k+1}\circ \cdots\circ \hat{f}_0\circ \pi\circ f_0^{-1}\circ \cdots \circ f_k^{-1}\circ f_k\circ f_{k-1}\circ\cdots\circ f_0\right)(\vez)\right)[k+1]\\
&=&\left(\left(\hat{f}_{k+1}\circ \cdots\circ \hat{f}_0\circ \pi\right)(\vez)\right)[k+1]
\end{eqnarray*}

Moreover, for any $h\le k$, recall that $f_{k+1}$ and $\hat{f}_{k+1}$ does not change the $h$-th coordinate, hence 
\begin{eqnarray*}
\left(f_{k+1}\circ f_k\circ\cdots\circ f_0\right)(\vez)[h]
&=&\left( f_k\circ\cdots\circ f_0\right)(\vez)[h]\\
&=&\left(\left(\hat{f}_{k}\circ \cdots\circ \hat{f}_1\circ \pi\right)(\vez)\right)[h]\\
&=&\left(\left(\hat{f}_{k+1}\circ\hat{f}_{k}\circ \cdots\circ \hat{f}_0\circ \pi\right)(\vez)\right)[h]
\end{eqnarray*}

Hence, the statement holds for all $k\le \upsilon-1$, and Lemma~\ref{lemma:shuffle-1} is proved.
\end{proof}

\section{Construction of the Scheduling Instance}\label{appsec:reduction-construction}

Now we provide the details of the reduction. We first recall all the functions and parameters we have set in proving Lemma~\ref{lemma:linked-sum}. 
\begin{itemize}
    \item Recall that $\tau$ is the one-to-one mapping that maps $i$ to $k$ for every $(z_i\oplus \neg z_k)\in C_2$. 
    \item Apply Lemma~\ref{lemma:uniquesum-2} by taking $N=n$, we get $\sigma: \Z_n\rightarrow \Z_{n'}$ where $n'=n^{1+\OO(\frac{1}{\sqrt{\log\log n}})}$, $\sigma(i)=\sum_{j=0}^\gamma \vea_i[j]x^j$ for $\vea_i[j]\in {\cal{S}}_d$ where $\gamma=\lceil\frac{\log n}{\log\log n}\rceil+O(\frac{\log n}{(\log\log n)^{3/2}})$, $d=e^{\OO(\sqrt{\log\log n})}\log n$ and $x=5d+1$. We lift the dimension such that $\vea_i=(\vea_i[0],\cdots,\vea_i[\gamma+3])$ where $\vea_i[\gamma+1]=\vea_i[\gamma+2]=\vea_i[\gamma+3]=0$. 
    \item Apply Lemma~\ref{lemma:uniquesum-2} again by setting $N=4\gamma+4$, we get another injection $\sigma'$ such that $\sigma'(y)=o(\log^2 n)<x^2$ for $y\le 4\gamma+4$.
    \item We have constructed in the proof of Lemma~\ref{lemma:linked-sum}: 
    $\veb^0_{i}=\vea_i,\veb^1_{i},\veb^2_{i},\cdots,\veb^{2\gamma+2}_i=\hat{\veb}^{2\gamma+2}_k$, $\hat{\veb}^{2\gamma+1}_k$, $\cdots$, $\hat{\veb}^1_k,\hat{\veb}^0_k=\vea_k$ where $k=\tau(i)$. 
    \item Again, each vector $\vecc$ represents the polynomial $\sum_i\vecc[i]x^i$. Polynomials and vectors are used interchangeably.
    \item Let $\sigma_{max}=x^{\gamma+6}=n^{1+O(\frac{1}{\sqrt{\log\log n}})}$, and thus $\sigma_{max}>x\cdot \veb^h_{i}\vex$ and $\sigma_{max}> x\cdot \hat{\veb}^h_{i}\vex$ for any $i,h$, and also $\sigma_{max}>\sigma'(y)x^{\gamma+3}$ for any $y\le 4\gamma+4$.
\end{itemize}

\smallskip
\noindent\textbf{Construction of the scheduling instance.} 
We shall construct two major classes of jobs, gap jobs and main jobs. 
Main jobs are divided into $5$ types: dummy jobs, clause jobs, truth-assignment jobs, link jobs and variable jobs. 
The three types -- truth-assignment, link and variable jobs -- are further divided into sub-types, e.g., variable jobs are further divided into 4 sub-types (see Table~\ref{table:job-time}). A gap job is defined as a fixed huge value $10^{14}\sigma_{max}$ subtracting several main jobs. 

The processing time of each job can be expressed as a summation over three components: Type, Index and True/False. The type component of a main job is always of the form $10^j\sigma_{max}$ where $2\le j\le 13$. Table~\ref{table:job-time} summarizes the value $j$ for each kind of main job, e.g., the type-component of a variable job whose sub-type belongs to V$_{\cdot,+,1}$ is $10^5\sigma_{max}$. The index-component of clause jobs, truth-assignment jobs and variable jobs is of the form $10\sigma(i)$ for some index $i$. Dummy jobs do not have index-component; Link jobs have much more complicated index-components, which will be specified in the following part of this subsection. Each main job has a true version and a false version. A gap job does not have a true/false version but only one unified version.     

\begin{table}[!ht]
\renewcommand\arraystretch{1.3}
\setlength\tabcolsep{1.5pt}
\centering{}
\resizebox{14cm}{!}{
\begin{tabular}{|c|c|c|c|c|c|c|c|c|c|c|c|c|}
\hline
$\backslash$&Dummy& Clause & \multicolumn{4}{c|}{Truth-assignment} & \multicolumn{2}{c|}{Link} & \multicolumn{4}{c|}{Variable}  \\ \hline
$\backslash$&DM& CL$_{\cdot}$ &  TR$_{\cdot,a}$   &  TR$_{\cdot,b}$   &  TR$_{\cdot,c}$   & TR$_{\cdot,d}$    &  LN$_{\cdot,+}$ & LN$_{\cdot,-}$  & V$_{\cdot,+,1}$     &    V$_{\cdot,+,2}$       &   V$_{\cdot,-,1}$  &   V$_{\cdot,-,2}$      \\ \hline
$\zeta(\cdot)$&13 & 12 &   11  &  10   &   9  &   8  &       7    &      6     &  5   &   4  &   3  &   2  \\ \hline
\end{tabular}
}
\caption{Type-component of main jobs}
\label{table:job-time}
\end{table}

Define a function $\zeta$ that maps the (sub)-type of a main job to the exponent of $10$ as indicated by Table~\ref{table:job-time}, e.g., $\zeta(\textrm{TR}_{\cdot,a})=11$. 
Now we provide the exact processing time of every job. In the following $\rho\in\{T,F\}$, $\iota\in\{+,-\}$.

$\bullet$ Variable jobs: 4 jobs $V_{i,+,1}^{\rho}$ and $V_{i,+,2}^{\rho}$ are
constructed for the positive literal $z_i$, and 4 jobs $V_{i,-,1}^{\rho}$ and $V_{i,-,2}^{\rho}$ are
for the negative literal $\neg z_i$.
\begin{eqnarray*}
&&s(V_{i,\iota,\kappa}^T)=10^{\zeta(\textrm{V}_{\cdot,\iota,k})}\sigma_{max}+10\sigma(i)+1,\\
&&s(V_{i,\iota,\kappa}^F)=10^{\zeta(\textrm{V}_{\cdot,\iota,k})}\sigma_{max}+10\sigma(i)+2,\quad
\kappa=1,2, \iota=+,-
\end{eqnarray*}

$\bullet$ Truth-assignment jobs: 8 jobs $\textrm{TR}_{i,a}^{\rho}$, $\textrm{TR}_{i,b}^{\rho}$,
$\textrm{TR}_{i,c}^{\rho}$ and $\textrm{TR}_{i,d}^{\rho}$ are constructed for every $i$.
\begin{eqnarray*}
&&s(\textrm{TR}_{i,\kappa}^{T})=10^{\zeta(\textrm{TR}_{\cdot,\kappa})}\sigma_{max}+10\sigma(i)+1.5, \\
&&s(\textrm{TR}_{i,\kappa}^{F})=10^{\zeta(\textrm{TR}_{\cdot,\kappa})}\sigma_{max}+10\sigma(i)+1. \quad \kappa=a,b,c,d
\end{eqnarray*}

$\bullet$ Clause jobs: there are 3 clause jobs for every clause $cl_{\ell}\in C_1$
where $\ell\in \{2,5,\cdots,n-1\}$, with one $\textrm{CL}_{\ell}^T$ and two copies of $\textrm{CL}_{\ell}^F$:
$$s(\textrm{CL}_{\ell}^T)=10^{\zeta(\textrm{CL}_{\cdot})}\sigma_{max}+10\sigma(\ell)+2, \quad s(\textrm{CL}_{\ell}^F)=10^{\zeta(\textrm{CL}_{\cdot})}\sigma_{max}+10\sigma(\ell)+1.$$

$\bullet$ Dummy jobs: there are $n+n/3$ true dummy jobs $\textrm{DM}^T$ of processing time $10^{\zeta(\textrm{DM})}\sigma_{max}+1$, and $n-n/3$
false dummy jobs $\textrm{DM}^F$ of processing time $10^{\zeta(\textrm{DM})}\sigma_{max}+2$.

$\bullet$ Link jobs: 
We create $4\gamma+4$ links jobs for each clause in $C_2$.
Recall the vectors $\veb^h_{i}$ and $\hat{\veb}^h_i$ for $1\le h\le 2\gamma+2$. 
For every clause $(z_i\oplus z_k)\in C_2$ and every $1\le h\le 2\gamma+2$, we create two pairs of link jobs, LN$_{i,h,+}^T$ and LN$_{i,h,+}^F$, and LN$_{k,h,-}^T$ and LN$_{k,h,-}^F$ such that
$$s(\textrm{LN}_{i,h,+}^T)=10^{\zeta(\textrm{LN}_{\cdot,+})}\sigma_{max}+10\veb^{h}(i)\vex+1, \quad s(\textrm{LN}_{i,h,+}^F)=10^{\zeta(\textrm{LN}_{\cdot,+})}\sigma_{max}+10\veb^h_{i}\vex+2,$$
$$s(\textrm{LN}_{k,h,-}^T)=10^{\zeta(\textrm{LN}_{\cdot,-})}\sigma_{max}+10\hat{\veb}^{h}(k)\vex+1, \quad s(\textrm{LN}_{k,h,-}^F)=10^{\zeta(\textrm{LN}_{\cdot,-})}\sigma_{max}+10\hat{\veb}^{h}(i)\vex+2.$$

Let
$\textrm{TR}_A$, $\textrm{TR}_B$, $\textrm{TR}_C$, $\textrm{TR}_D$ be the set of jobs $\textrm{TR}_{i,a}^{\rho}$, $\textrm{TR}_{i,b}^{\rho}$,
$\textrm{TR}_{i,c}^{\rho}$ and $\textrm{TR}_{i,d}^{\rho}$ respectively. Sometimes we may
drop the superscript for simplicity, e.g., we use $\textrm{TR}_{i,a}$ to represent
$\textrm{TR}_{i,a}^T$ or $\textrm{TR}_{i,a}^F$. We construct gap jobs. There are 5 kinds of gap jobs.

$\bullet$ There are two gap jobs (variable-link jobs) $\theta_{\textrm{LN},i,+}$ and
$\theta_{\textrm{V-L},i,-}$ for each variable $z_i$:
\begin{eqnarray*}
s(\theta_{\textrm{V-L},i,+})&=&(10^{14}-10^{7}-10^4)\sigma_{max}-10\left(\vea_i[0]+2\sum_{j=1}^{\gamma}\vea_i[j]x^j+(F_0(\vea_i))[0]\cdot x^{\gamma+1}+{\sigma}'(1)x^{\gamma+2}\right)-3\\
&=& \left(10^{14}-10^{\zeta(\textrm{LN}_{\cdot,+})}-10^{\zeta(V_{\cdot,+,2})}\right)\sigma_{max}-10(\vea_i+\veb^1_{i})\vex-3\\
s(\theta_{\textrm{V-L},i,-})&=&(10^{14}-10^6-10^2)\sigma_{max}-10\left(\vea_i[0]+2\sum_{j=1}^{\gamma}\vea_i[j]x^j+(\hat{F}(\vea_i))[0]\cdot x^{\gamma+1}+{\sigma}'(1)x^{\gamma+2}\right)-3\\
&=& \left(10^{14}-10^{\zeta(\textrm{LN}_{\cdot,-})}-10^{\zeta(V_{\cdot,-,2})}\right)\sigma_{max}-10\left(\vea_i+\hat{\veb}^1_i\right)\vex-3
\end{eqnarray*}

$\bullet$ There are $4\gamma+3$ gap jobs (link-link jobs), $\theta_{\textrm{L-L},i,h,+}$ and
$\theta_{\textrm{L-L},i,h,-}$ and $\theta_{\textrm{L-L},i,+,-}$ for every $1\le i\le n$. For
$h=1,2,\cdots,2\gamma+1$, we define
\begin{eqnarray*}
s(\theta_{\textrm{L-L},i,h,+})
&=& \left(10^{14}-2\times 10^{\zeta(\textrm{LN}_{\cdot,+})}\right)\sigma_{max}-10\left(\veb^h_{i}\vex+\veb^{h+1}_i\right)\vex-3\\
s(\theta_{\textrm{L-L},i,h,-})
&=& \left(10^{14}-2\times 10^{\zeta(\textrm{LN}_{\cdot,-})}\right)\sigma_{max}-10\left(\hat{\veb}^h_i\vex+\hat{\veb}^{h+1}_i\right)\vex-3
\end{eqnarray*}
Additionally, we define
\begin{eqnarray*}
s(\theta_{\textrm{L-L},i,+,-})
= \left(10^{14}-10^{\zeta(\textrm{LN}_{\cdot,+})}-10^{\zeta(\textrm{LN}_{\cdot,-})}\right)\sigma_{max}-2\times 10\veb^{2\gamma+2}_i\vex-3
\end{eqnarray*}
Here recall that $\veb^{2\gamma+2}_i=\hat{\veb}^{2\gamma+2}_{\tau(i)}$.

$\bullet$ There are three gap jobs (variable-clause-dummy jobs) for each
$cl_{\ell}\in C_1$ ($\ell\in \{2,5,\cdots,n-1\}$):
for $i=\ell-1,\ell,\ell+1$,
if $z_i\in cl_\ell$, we construct $\theta_{\textrm{V-C-D},\ell,i,+}$, otherwise $\neg
z_i\in cl_\ell$, and we construct $\theta_{\textrm{V-C-D},\ell,i,-}$:
\begin{eqnarray*}
&&s(\theta_{\textrm{V-C-D}, \ell,i,+})=\left(10^{14}-10^{\zeta(\textrm{DM})}-10^{\zeta(\textrm{CL}_{\cdot})}-10^{\zeta(V_{\cdot,+,1})}\right)\sigma_{max}-10(\sigma(\ell)+\sigma(i))-4,\\
&&s(\theta_{\textrm{V-C-D},\ell,i,-})=\left(10^{14}-10^{\zeta(\textrm{DM})}-10^{\zeta(\textrm{CL}_{\cdot})}-10^{\zeta(V_{\cdot,-,1})}\right)\sigma_{max}-10(\sigma(\ell)+\sigma(i))-4.
\end{eqnarray*}

$\bullet$ There is one gap job (variable-dummy job) for each variable.
Notice that each variable appears exactly once in clauses of $C_1$, if
$z_i$ appears in $C_1$, we construct
$\theta_{\textrm{V-D},i,-}$. Otherwise, we construct $\theta_{\textrm{V-D},i,+}$ instead.
\begin{eqnarray*}
&&s(\theta_{\textrm{V-D},i,+})=\left(10^{14}-10^{\zeta(\textrm{DM})}-10^{\zeta(V_{\cdot,+,1})}\right)\sigma_{max}-10\sigma(i)-3,\\
&&s(\theta_{\textrm{V-D},i,-})=\left(10^{14}-10^{\zeta(\textrm{DM})}-10^{\zeta(V_{\cdot,-,1})}\right)\sigma_{max}-10\sigma(i)-3.
\end{eqnarray*}

Thus, for each clause $cl_\ell$ and $i=\ell-1,\ell,\ell+1$, either
$\theta_{\textrm{V-D},i,+}$ and $\theta_{\textrm{V-C-D},\ell,i,-}$ exist, or
$\theta_{\textrm{V-D},i,-}$ and $\theta_{\textrm{V-C-D},\ell,i,+}$ exist.

$\bullet$ There are four gap jobs (variable-truth jobs) for each
variable $z_i$, namely $\theta_{\textrm{V-T},i,a,c}$, $\theta_{\textrm{V-T},i,b,d}$,
$\theta_{\textrm{V-T},i,a,d}$ and $\theta_{\textrm{V-T},i,b,c}$:
\begin{eqnarray*}
&&s(\theta_{\textrm{V-T},i,a,c})=\left(10^{14}-10^{\zeta(V_{\cdot,+,1})}-10^{\zeta(\textrm{TR}_{\cdot,a})}-10^{\zeta(\textrm{TR}_{\cdot,c})}\right)\sigma_{max}-30\sigma(i)-4,\\
&&s(\theta_{\textrm{V-T},i,b,d})=\left(10^{14}-10^{\zeta(V_{\cdot,+,2})}-10^{\zeta(\textrm{TR}_{\cdot,b})}-10^{\zeta(\textrm{TR}_{\cdot,d})}\right)\sigma_{max}-30\sigma(i)-4,\\
&&s(\theta_{\textrm{V-T},i,a,d})=\left(10^{14}-10^{\zeta(V_{\cdot,-,1})}-10^{\zeta(\textrm{TR}_{\cdot,a})}-10^{\zeta(\textrm{TR}_{\cdot,d})}\right)\sigma_{max}-30\sigma(i)-4,\\
&&s(\theta_{\textrm{V-T},i,b,c})=\left(10^{14}-10^{\zeta(V_{\cdot,-,2})}-10^{\zeta(\textrm{TR}_{\cdot,b})}-10^{\zeta(\textrm{TR}_{\cdot,c})}\right)\sigma_{max}-30\sigma(i)-4,
\end{eqnarray*}

Overall, we have constructed $2\gamma n+8n$ gap jobs. We also construct $2\gamma n+8n$ machines. The following Table~\ref{table:job-time-whole} summarizes the processing times of all jobs. 

\begin{table}[!ht]
\renewcommand\arraystretch{1.3}
\setlength\tabcolsep{1.5pt}
\centering
\resizebox{14cm}{!}{

\begin{tabular}{|c|l|l|l|c|c|}
\hline
       Job-type           & Sub-type & Type-component & Index-component &   \makecell{T/F\\(T)}        &     \makecell{T/F\\(F)}      \\ \hline
\multirow{4}{*}{Variable} & $V_{i,+,1}$ &  $10^{\zeta(\textrm{V}_{\cdot,+,1})}\sigma_{max}$ & $10\sigma(i)$ &      1     &     2      \\ \cline{2-6} 
                  & $V_{i,+,2}$ & $10^{\zeta(\textrm{V}_{\cdot,+,2})}\sigma_{max}$ & $10\sigma(i)$ &     1     &     2      \\ \cline{2-6} 
                  & $V_{i,-,1}$ & $10^{\zeta(\textrm{V}_{\cdot,-,1})}\sigma_{max}$  & $10\sigma(i)$ &      1     &     2      \\ \cline{2-6} 
                  & $V_{i,-,2}$ & $10^{\zeta(\textrm{V}_{\cdot,-,2})}\sigma_{max}$ & $10\sigma(i)$ &     1      &        2   \\ \hline
\multirow{4}{*}{Truth-assignment} & TR$_{i,a}$ & $10^{\zeta(\textrm{TR}_{\cdot,a})}\sigma_{max}$ & $10\sigma(i)$ &    1.5       &     1      \\ \cline{2-6} 
                  & TR$_{i,b}$ & $10^{\zeta(\textrm{TR}_{\cdot,b})}\sigma_{max}$  & $10\sigma(i)$ &     1.5      &     1      \\ \cline{2-6} 
                  & TR$_{i,c}$ & $10^{\zeta(\textrm{TR}_{\cdot,c})}\sigma_{max}$  & $10\sigma(i)$ &      1.5     &     1      \\ \cline{2-6} 
                  &  TR$_{i,d}$ & $10^{\zeta(\textrm{TR}_{\cdot,d})}\sigma_{max}$ & $10\sigma(i)$ &     1.5      &     1      \\ \hline
                 Clause  & CL$_\ell$ & $10^{\zeta(\textrm{CL}_{\cdot})}\sigma_{max}$ & $10\sigma(\ell)$ &       2    &   1        \\ \hline
                 Dummy & DM &  $10^{\zeta(\textrm{DM})}\sigma_{max}$ & 0 &     1      &    2       \\ \hline
\multirow{2}{*}{\makecell{Link\\$h\in\{0,1,\cdots,2\gamma+2\}$}} & {LN$_{i,h,+}$} & $10^{\zeta(\textrm{LN}_{\cdot,+})}\sigma_{max}$ & $10\veb^h_{i}\vex$ &    1       &      2     \\ \cline{2-6} 
                  & {LN$_{i,h,-}$} & $10^{\zeta(\textrm{LN}_{\cdot,-})}\sigma_{max}$ & $10\veb^h_{i}\vex$ &    1       &  2         \\ \cline{2-6} \hline

\multirow{2}{*}{Variable-Link} & $\theta_{\textrm{V-L},i,+}$ & $(10^{14}-10^{\zeta(\textrm{LN}_{\cdot,+})}-10^{\zeta(V_{\cdot,+,2})})\sigma_{max}$ & $-10(\vea_i+\veb^1_{i})\vex$ & \multicolumn{2}{l|}{-3} \\ \cline{2-6} 
                  & $\theta_{\textrm{V-L},i,-}$ & $(10^{14}-10^{\zeta(\textrm{LN}_{\cdot,-})}-10^{\zeta(V_{\cdot,-,2})})\sigma_{max}$ & $-10(\vea_i+\hat{\veb}^1_i)\vex$ & \multicolumn{2}{l|}{-3} \\ \hline
\multirow{3}{*}{\makecell{Link-Link\\$h\in\{1,\cdots,2\gamma+1\}$}} & {$\theta_{\textrm{L-L},i,h,+}$} & $(10^{14}-2\times 10^{\zeta(\textrm{LN}_{\cdot,+})})\sigma_{max}$ & $-10(\veb^h_{i}+\veb^{h+1}_i)\vex$ & \multicolumn{2}{l|}{-3} \\ \cline{2-6} 
                  & {$\theta_{\textrm{L-L},i,h,-}$} & $(10^{14}-2\times 10^{\zeta(\textrm{LN}_{\cdot,-})})\sigma_{max}$ & $-10(\veb^h_{i}+\hat{\veb}^{h+1}_i)\vex$ & \multicolumn{2}{l|}{-3} \\ \cline{2-6} 
                  & $\theta_{\textrm{L-L},i,+,-}$ & $(10^{14}-10^{\zeta(\textrm{LN}_{\cdot,+})}-10^{\zeta(\textrm{LN}_{\cdot,-})})\sigma_{max}$ & $-20\veb^{2\gamma+2}_i\vex$ & \multicolumn{2}{l|}{-3} \\ \cline{2-6} \hline
\multirow{2}{*}{\makecell{Variable-Clause\\-Dummy, $|i-\ell|\le 1$}} & $\theta_{\textrm{V-C-D},\ell,i,+}$ & $(10^{14}-10^{\zeta(\textrm{DM})}-10^{\zeta(\textrm{CL}_{\cdot})}-10^{\zeta(V_{\cdot,+,1})})\sigma_{max}$ & $-10(\sigma(\ell)+\sigma(i))$ & \multicolumn{2}{l|}{-4} \\ \cline{2-6} 
                  & $\theta_{\textrm{V-C-D},\ell,i,-}$ & $(10^{14}-10^{\zeta(\textrm{DM})}-10^{\zeta(\textrm{CL}_{\cdot})}-10^{\zeta(V_{\cdot,-,1})})\sigma_{max}$ & $-10(\sigma(\ell)+\sigma(i))$ & \multicolumn{2}{l|}{-4} \\ \hline
\multirow{2}{*}{Variable-Dummy} & $\theta_{\textrm{V-D},i,+}$ & $(10^{14}-10^{\zeta(\textrm{DM})}-10^{\zeta(V_{\cdot,+,1})})\sigma_{max}$ & $-10\sigma(i)$ & \multicolumn{2}{l|}{-3} \\ \cline{2-6} 
                  & $\theta_{\textrm{V-D},i,-}$ & $(10^{14}-10^{\zeta(\textrm{DM})}-10^{\zeta(V_{\cdot,-,1})})\sigma_{max}$ & $-10\sigma(i)$ & \multicolumn{2}{l|}{-3} \\ \hline
\multirow{4}{*}{{Variable-Truth}} & $\theta_{\textrm{V-T},i,a,c}$ & $(10^{14}-10^{\zeta(V_{\cdot,+,1})}-10^{\zeta(\textrm{TR}_{\cdot,a})}-10^{\zeta(\textrm{TR}_{\cdot,c})})\sigma_{max}$ & $-30\sigma(i)$ & \multicolumn{2}{l|}{-4} \\ \cline{2-6} 
                  & $\theta_{\textrm{V-T},i,b,d}$ & $(10^{14}-10^{\zeta(V_{\cdot,+,2})}-10^{\zeta(\textrm{TR}_{\cdot,b})}-10^{\zeta(\textrm{TR}_{\cdot,d})})\sigma_{max}$ & $-30\sigma(i)$ & \multicolumn{2}{l|}{-4} \\ \cline{2-6} 
                  & $\theta_{\textrm{V-T},i,a,d}$ & $(10^{14}-10^{\zeta(V_{\cdot,-,1})}-10^{\zeta(\textrm{TR}_{\cdot,a})}-10^{\zeta(\textrm{TR}_{\cdot,d})})\sigma_{max}$ & $-30\sigma(i)$ & \multicolumn{2}{l|}{-4} \\ \cline{2-6} 
                  & $\theta_{\textrm{V-T},i,b,c}$ & $(10^{14}-10^{\zeta(V_{\cdot,-,2})}-10^{\zeta(\textrm{TR}_{\cdot,b})}-10^{\zeta(\textrm{TR}_{\cdot,c})})\sigma_{max}$ & $-30\sigma(i)$ & \multicolumn{2}{l|}{-4} \\ \hline
\end{tabular}
}
\caption{Job processing times}
\label{table:job-time-whole}
\end{table}

\section{Proof of Theorem~\ref{thm:lower-bound}}\label{appsec:remaining-proofs}
The proof is carried out in 4 steps. We first show in Section~\ref{subsec:unique-time} that every job in the constructed instance has a unique processing time. This allows us to refer to a job by its symbol (e.g., $V_{i,+,1}^T$) as well as by its processing time. Next, we show in Section~\ref{subsec:sat-to-scheduling-simple} that if a significant fraction of clauses in the 3SAT$'$ instance are satisfiable, then the constructed scheduling instance admits a solution with a small objective value. Next, we show in Section~\ref{subsec:sche-to-sat} that if any truth-assignment for the 3SAT$'$ instance will leave a significant fraction of clauses unsatisfied, then the constructed scheduling instance does not admit a solution with a small objective value. Finally, we are able to prove the correctness of our reduction in Section~\ref{subsec:formal-proof} by leveraging the above two facts.

\subsection{Uniqueness of job processing times}\label{subsec:unique-time}
We claim that the processing time of each job we create is unique, whereas there is a one-to-one correspondence between the symbol of a job and its processing time. To see the claim, consider Table~\ref{table:job-time-whole}. It suffices to compare the processing time of jobs within each subtype. Given that $\sigma$ is an injection, it is easy to see that the processing time of each variable job, truth-assignment job, clause job, variable-dummy job and variable-truth job is unique. For variable-clause-dummy jobs, by property 4 of Lemma~\ref{lemma:uniquesum-2} we know the sum $\sigma(\ell)+\sigma(i)$ for $i=\ell-1,\ell,\ell+1$ is unique. The uniqueness of link jobs follows from the uniqueness of $\veb^h_{i}$'s from Lemma~\ref{lemma:veb-uniquesum-1}. The uniqueness of link-link jobs follows from the uniqueness of the summation $\veb^h_{i}+\veb^{h+1}_i$ from Lemma~\ref{lemma:veb-uniquesum-2}.

\subsection{3SAT\texorpdfstring{$'$}{Lg} to Scheduling}
\label{subsec:sat-to-scheduling-simple}
The goal of this subsection is to prove the following lemma.
\begin{lemma}\label{lemma:sat-sche}
If there are at most $\vartheta n$ clauses which are not satisfied, then the constructed scheduling instance admits a feasible schedule with objective value at most $(10^{14}\sigma_{max})^{q} (2\gamma n+ 8n) +\vartheta n\cdot \frac{q(q-1)}{2}(10^{14}\sigma_{max})^{q-2}+o(n\sigma_{max}^{q-2})$.
\end{lemma}

Recall that every main job, except the clause job, admits a true copy and false copy, while the clause job admits a true copy and two false copies. We first ignore the true/false version of jobs and schedule them according to Table~\ref{table:job-schedule}, where each row represents jobs that are scheduled on one machine. 

\begin{table}[!ht]
\renewcommand\arraystretch{1.3}
\setlength\tabcolsep{1.5pt}
\centering
\resizebox{9cm}{!}{
\begin{tabular}{|c|l|l|l|l|ll}
\hline
\multirow{2}{*}{Variable-Link} & $\theta_{\textrm{V-L},i,+}$ & $V_{i,+,2}$  & LN$_{i,1,+}$ & $\backslash$ \\ \cline{2-5} 
                  & $\theta_{\textrm{V-L},i,-}$ & $V_{i,-,2}$ & LN$_{i,1,-}$ & $\backslash$  \\ \hline
\multirow{3}{*}{\makecell{Link-Link\\$h\in\{1,2,\cdots,2\gamma+1\}$}} & {$\theta_{\textrm{L-L},i,h,+}$} & LN$_{i,h,+}$ & LN$_{i,h+1,+}$ & $\backslash$  \\ \cline{2-5} 
                  & {$\theta_{\textrm{L-L},i,h,-}$} & LN$_{i,h,-}$ & LN$_{i,h+1,-}$ & $\backslash$    \\ \cline{2-5} 
                  & $\theta_{\textrm{L-L},i,+,-}$ & LN$_{i,2\gamma+2,+}$ & LN$_{\tau(i),2\gamma+2,-}$ & $\backslash$  \\ \cline{2-5} \hline
\multirow{2}{*}{\makecell{Variable-Clause-Dummy\\$|i-\ell|\le 1$}} & $\theta_{\textrm{V-C-D},\ell,i,+}$ & $V_{i,+,1}$ & CL$_{\ell}$ & DM  \\ \cline{2-5} 
                  & $\theta_{\textrm{V-C-D},\ell,i,-}$ & $V_{i,-,1}$ & CL$_{\ell}$ & DM  \\ \hline
\multirow{2}{*}{{Variable-Dummy}} & $\theta_{\textrm{V-D},i,+}$ & $V_{i,+,1}$ & DM & $\backslash$  \\ \cline{2-5} 
                  & $\theta_{\textrm{V-D},i,-}$ & $V_{i,-,1}$ & DM & $\backslash$  \\ \hline
\multirow{4}{*}{{Variable-Truth}} & $\theta_{\textrm{V-T},i,a,c}$ & $V_{i,+,1}$ & TR$_{i,a}$ & TR$_{i,c}$  \\ \cline{2-5} 
                  & $\theta_{\textrm{V-T},i,b,d}$ & $V_{i,+,2}$ & TR$_{i,b}$ & TR$_{i,d}$  \\ \cline{2-5} 
                  & $\theta_{\textrm{V-T},i,a,d}$ & $V_{i,-,1}$ & TR$_{i,a}$ & TR$_{i,d}$  \\ \cline{2-5} 
                  & $\theta_{\textrm{V-T},i,b,c}$ & $V_{i,-,2}$ & TR$_{i,b}$ & TR$_{i,c}$  \\ \hline
\end{tabular}
}

\caption{SAT to Scheduling -- Jobs scheduled on each machine}
\label{table:job-schedule}
\end{table}

We show that if we schedule according to Table~\ref{table:job-schedule}, then every job has been scheduled (ignoring the superscripts $T$ or $F$, which will be determined later). It is obvious that every gap job is scheduled. For simplicity, we abuse the notation a bit by using the symbol of a gap job to denote the machine on which it is scheduled. 


$\bullet$ Consider clause jobs. Recall that for each clause $cl_\ell$ and $i=\ell-1,\ell,\ell+1$, we either construct
$\theta_{\textrm{V-D},i,-}$ and $\theta_{\textrm{V-C-D},\ell,i,+}$ if the positive literal $z_i$ occurs in $C_1$, or
construct $\theta_{\textrm{V-D},i,+}$ and $\theta_{\textrm{V-C-D},\ell,i,-}$ if the negative literal $\neg z_i$ occurs in $C_1$. Hence the three copies of job $\textrm{CL}_\ell$ appear on machine $\theta_{\textrm{V-C-D},\ell,\ell-1,+}$ or $\theta_{\textrm{V-C-D},\ell,\ell-1,-}$, machine $\theta_{\textrm{V-C-D},\ell,\ell,+}$ or $\theta_{\textrm{V-C-D},\ell,\ell,-}$, and machine $\theta_{\textrm{V-C-D},\ell,\ell+1,+}$ or $\theta_{\textrm{V-C-D},\ell,\ell+1,-}$. Thus, all three copies of a clause job are scheduled.

$\bullet$ Consider truth-assignment jobs. There are two copies of TR$_{i,a}$, TR$_{i,b}$, TR$_{i,c}$ and TR$_{i,d}$. It is easy to see that all of them are scheduled on machines $\theta_{\textrm{V-T},i,a,c}$, $\theta_{\textrm{V-T},i,b,d}$, $\theta_{\textrm{V-T},i,a,d}$ and $\theta_{\textrm{V-T},i,b,c}$.

$\bullet$ Consider variable jobs. There are two copies of $V_{i,+,1}$, $V_{i,+,2}$, $V_{i,-,1}$ and $V_{i,-,2}$. It is easy to see that one copy of them are scheduled on machines $\theta_{\textrm{V-T},i,a,c}$, $\theta_{\textrm{V-T},i,b,d}$, $\theta_{\textrm{V-T},i,a,d}$ and $\theta_{\textrm{V-T},i,b,c}$. One copy of $V_{i,+,2}$ and $V_{i,-,2}$ are scheduled on machines $\theta_{\textrm{V-L},i,+}$ and $\theta_{\textrm{V-L},i,-}$. If machines
$\theta_{\textrm{V-D},i,-}$ and $\theta_{\textrm{V-C-D},\ell,i,+}$ exist (when the positive literal $z_i$ occurs in $C_1$), then $V_{i,-,1}$ and $V_{i,+,1}$ are scheduled on them respectively; otherwise  machines $\theta_{\textrm{V-D},i,+}$ and $\theta_{\textrm{V-C-D},\ell,i,-}$ exist (the negative literal $\neg z_i$ occurs in $C_1$), then $V_{i,+,1}$ and $V_{i,-,1}$ are scheduled on them respectively.

$\bullet$ Consider link jobs. There are two copies of LN$_{i,h,+}$ (or LN$_{i,h,-}$) for $1\le h\le 2\gamma+2$. Let $\iota\in\{+,-\}$. The two copies of LN$_{i,1,\iota}$ are scheduled on machines $\theta_{\textrm{V-L},i,\iota}$ and $\theta_{\textrm{L-L},i,1,\iota}$. The two copies of LN$_{i,h,\iota}$ are scheduled on $\theta_{\textrm{L-L},i,h,\iota}$ and $\theta_{\textrm{L-L},i,h+1,\iota}$ for $2\le h\le 2\gamma+1$. The two copies of LN$_{i,2\gamma+2,+}$ are scheduled on machines $\theta_{\textrm{L-L},2\gamma+1,+}$ and $\theta_{\textrm{L-L},i,+,-}$, and the two copies of LN$_{i,2\gamma+2,-}$ are scheduled on machines $\theta_{\textrm{L-L},2\gamma+1,-}$ and $\theta_{\textrm{L-L},\tau^{-1}(i),+,-}$, where $\tau^{-1}$ is the inverse of the mapping $\tau$ (note that $\tau^{-1}$ exists since $\tau$ is one-to-one).

$\bullet$ Consider dummy jobs. There are in total $2n$ dummy jobs. It is obvious that for every $i$, 2 dummy jobs are scheduled on machines $\theta_{\textrm{V-C-D},\ell,i,+}$, $\theta_{\textrm{V-D},i,-}$ or machines $\theta_{\textrm{V-C-D},\ell,i,-}$, $\theta_{\textrm{V-D},i,+}$. 


Next, we consider the load of every machine. According to Table~\ref{table:job-time-whole}, it is easy to verify that if we sum up the type-component of jobs on each machine, it becomes $10^{14}\sigma_{max}$; if we sum up the index-component of jobs on each machine, it becomes $0$. Now we consider the T/F-component of jobs. It is easy to verify that the T/F-components of all jobs add up to $0$, hence we have the following direct observation.

\begin{observation}\label{obs:total-processing}
The total processing time of all jobs add up to $10^{14}\sigma_{max} \cdot (2\gamma n +8n)$.
\end{observation}

\begin{table}[!ht]
\renewcommand\arraystretch{1.3}
\setlength\tabcolsep{1.5pt}
\centering
\resizebox{12cm}{!}{
\begin{tabular}{|c|l|l|l|l|l|}
\hline
\multirow{2}{*}{Variable-Link} & $\theta_{\textrm{V-L},i,+}$ & $V_{i,+,2}^T$  & LN$_{i,1,+}^F$ & $\backslash$  \\ \cline{2-5} 
                  & $\theta_{\textrm{V-L},i,-}$ & $V_{i,-,2}^F$ & LN$_{i,1,-}^{T}$ & $\backslash$  \\ \hline
\multirow{3}{*}{\makecell{Link-Link\\$h\in\{1,2,\cdots,2\gamma+1\}$}} & {$\theta_{\textrm{L-L},i,h,+}$} & LN$_{i,h,+}^F$ & LN$_{i,h+1,+}^T$ & $\backslash$  \\ \cline{2-5} 
                  & {$\theta_{\textrm{L-L},i,h,-}$} & LN$_{i,h,-}^F$ & LN$_{i,h+1,-}^T$ & $\backslash$    \\ \cline{2-5} 
                  & $\theta_{\textrm{L-L},i,+,-}$ & LN$_{i,2\gamma+2,+}^T$ & LN$_{\tau(i),2\gamma+2,-}^*$ & $\backslash$  \\ \cline{2-5} \hline
\multirow{2}{*}{\makecell{Variable-Clause-Dummy \& Variable-Dummy\\Case 1: positive literal $z_i\in C_1$}} & $\theta_{\textrm{V-C-D},\ell,i,+}$ & $V_{i,+,1}^T$ & CL$_{\ell}^{*}$ & DM$^{*}$  \\ \cline{2-5} 
                  & $\theta_{\textrm{V-D},i,-}$ & $V_{i,-,1}^F$ & DM$^*$ & $\backslash$  \\ \hline
\multirow{2}{*}{\makecell{Variable-Clause-Dummy \& Variable-Dummy\\Case 2: negative literal $\neg z_i\in C_1$}} & $\theta_{\textrm{V-C-D},i,-}$ & $V_{i,-,1}^F$ & CL$_{\ell}^{*}$ & DM$^*$  \\ \cline{2-5} 
                  & $\theta_{\textrm{V-D},i,+}$ & $V_{i,+,1}^T$ & DM$^*$ & $\backslash$  \\ \hline
\multirow{4}{*}{{Variable-Truth}} & $\theta_{\textrm{V-T},i,a,c}$ & $V_{i,+,1}^F$ & TR$_{i,a}^F$ & TR$_{i,c}^F$  \\ \cline{2-5} 
                  & $\theta_{\textrm{V-T},i,b,d}$ & $V_{i,+,2}^F$ & TR$_{i,b}^F$ & TR$_{i,d}^F$  \\ \cline{2-5} 
                  & $\theta_{\textrm{V-T},i,a,d}$ & $V_{i,-,1}^T$ & TR$_{i,a}^T$ & TR$_{i,d}^T$  \\ \cline{2-5} 
                  & $\theta_{\textrm{V-T},i,b,c}$ & $V_{i,-,2}^T$ & TR$_{i,b}^T$ & TR$_{i,c}^T$  \\ \hline
\end{tabular}
}
\caption{Scheduling of Truth/False types of jobs if the variable $z_i$ is true}
\label{table:job-variable-true}
\end{table}

\begin{table}[!ht]
\renewcommand\arraystretch{1.3}
\setlength\tabcolsep{1.5pt}
\centering
\resizebox{12cm}{!}{
\begin{tabular}{|c|l|l|l|l|l|}
\hline
\multirow{2}{*}{Variable-Link} & $\theta_{\textrm{V-L},i,+}$ & $V_{i,+,2}^F$  & LN$_{i,1,+}^T$ & $\backslash$ \\ \cline{2-5} 
                  & $\theta_{\textrm{V-L},i,-}$ & $V_{i,-,2}^T$ & LN$_{i,1,-}^{F}$ & $\backslash$  \\ \hline
\multirow{3}{*}{\makecell{Link-Link\\$h\in\{1,2,\cdots,2\gamma+1\}$}} & {$\theta_{\textrm{L-L},i,h,+}$} & LN$_{i,h,+}^T$ & LN$_{i,h+1,+}^F$ & $\backslash$\\ \cline{2-5} 
                  & {$\theta_{\textrm{L-L},i,h,-}$} & LN$_{i,h,-}^T$ & LN$_{i,h+1,-}^F$ & $\backslash$   \\ \cline{2-5} 
                  & $\theta_{\textrm{L-L},i,+,-}$ & LN$_{i,2\gamma+2,+}^F$ & LN$_{\tau(i),2\gamma+2,-}^*$ & $\backslash$  \\ \cline{2-5} \hline
\multirow{2}{*}{\makecell{Variable-Clause-Dummy \& Variable-Dummy\\Case 1: positive literal $z_i\in C_1$}} & $\theta_{\textrm{V-C-D},\ell,i,+}$ & $V_{i,+,1}^F$ & CL$_{\ell}^{*}$ & DM$^{*}$  \\ \cline{2-5} 
                  & $\theta_{\textrm{V-D},i,-}$ & $V_{i,-,1}^T$ & DM$^*$ & $\backslash$  \\ \hline
\multirow{2}{*}{\makecell{Variable-Clause-Dummy \& Variable-Dummy\\Case 2: negative literal $\neg z_i\in C_1$}} & $\theta_{\textrm{V-C-D},i,-}$ & $V_{i,-,1}^T$ & CL$_{\ell}^{*}$ & DM$^*$  \\ \cline{2-5} 
                  & $\theta_{\textrm{V-D},i,+}$ & $V_{i,+,1}^F$ & DM$^*$ & $\backslash$  \\ \hline
\multirow{4}{*}{{Variable-Truth}} & $\theta_{\textrm{V-T},i,a,c}$ & $V_{i,+,1}^T$ & TR$_{i,a}^T$ & TR$_{i,c}^T$ \\ \cline{2-5} 
                  & $\theta_{\textrm{V-T},i,b,d}$ & $V_{i,+,2}^T$ & TR$_{i,b}^T$ & TR$_{i,d}^T$  \\ \cline{2-5} 
                  & $\theta_{\textrm{V-T},i,a,d}$ & $V_{i,-,1}^F$ & TR$_{i,a}^F$ & TR$_{i,d}^F$  \\ \cline{2-5} 
                  & $\theta_{\textrm{V-T},i,b,c}$ & $V_{i,-,2}^F$ & TR$_{i,b}^F$ & TR$_{i,c}^F$  \\ \hline
\end{tabular}
}
\caption{Scheduling of Truth/False types of jobs if the variable $z_i$ is false}
\label{table:job-variable-false}
\end{table}

Consider the truth-assignment of $I_{sat}$. If the variable $z_i$ is true, then we determine the true/false version of main jobs according to Table~\ref{table:job-variable-true}. Otherwise the variable $z_i$ is false in the assignment, then we flip the True/False version of all jobs in Table~\ref{table:job-variable-true}, i.e., we schedule according to Table~\ref{table:job-variable-false}. It is easy to see that in each row of Table~\ref{table:job-variable-true}, if there is no job with a superscript of $*$, then their T/F-components sum up to $0$, i.e., the load of this machine is exactly $10^{14}\sigma_{max}$. We call the current schedule a semi-schedule. It remains to determine the true/false version of jobs with the superscript $*$. 

$\bullet$ Consider link-link machines. We only need to consider machines $\theta_{\textrm{L-L},i,+,-}$. The T/F-type of the job LN$_{i,2\gamma+2,+}$ has already been decided based on the true/false of variable $z_i$. Consider the other job LN$_{\tau(i),2\gamma+2,-}$ scheduled on this machine. Notice that based on the true/false of the variable $z_{\tau(i)}$, one copy of LN$_{\tau(i),2\gamma+2,-}$ is scheduled on $\theta_{\textrm{L-L},\tau(i),2\gamma+1,-}$, and the remaining copy is scheduled on $\theta_{\textrm{L-L},i,+,-}$. If $z_{\tau(i)}$ is true, the remaining copy is LN$_{\tau(i),2\gamma+2,-}^F$; otherwise, the remaining copy is LN$_{\tau(i),2\gamma+2,-}^T$. Hence, we have the following observation:
\begin{itemize}[--]
\item if variables $z_i$ is true and $z_{\tau(i)}$ is false, then LN$_{i,2\gamma+2,+}^T$ and LN$_{\tau(i),2\gamma+2,-}^T$ are on this machine, whereas the load is $10^{14}\sigma_{max}-1$;
\item if variables $z_i$ is false and $z_{\tau(i)}$ is true, then LN$_{i,2\gamma+2,+}^F$ and LN$_{\tau(i),2\gamma+2,-}^F$ are on this machine, whereas the load is $10^{14}\sigma_{max}+1$;
\item if variables $z_i$ and $z_{\tau(i)}$ are both true or both false, then one of LN$_{i,2\gamma+2,+}$ and LN$_{\tau(i),2\gamma+2,-}$ is true and the other is false, whereas the load is $10^{14}\sigma_{max}$.
\end{itemize}

The above observation leads to the following claim.
\begin{claim}\label{claim:sat-to-schedule-1}
The load of machine $\theta_{\textrm{L-L},i,+,-}$ is $10^{14}\sigma_{max}$ if the clause $(z_i\oplus\neg z_{\tau(i)})$ is satisfied, and is $10^{14}\sigma_{max}\pm 1$ otherwise.
\end{claim}

$\bullet$ Consider variable-clause-dummy and variable-dummy machines. Notice that there is one true copy and two false copies of CL$_{\ell}$, scheduled on machines $\theta_{\textrm{V-C-D},\ell,i,\kappa_{i-\ell}}$ where $i\in\{\ell-1,\ell,\ell+1\}$ and $\kappa_{i-\ell}\in\{+,-\}$. If there exists at least one $i=i^*$ such that $V_{i,\kappa_{i^*-\ell},1}^T$ is on machine $\theta_{\textrm{V-C-D},\ell,i^*,\kappa_{i^*-\ell}}$, then we schedule CL$_{\ell}^T$ on machine $\theta_{\textrm{V-C-D},\ell,i^*,\kappa_{i^*-\ell}}$, and schedule the two copies of CL$_{\ell}^F$ on the remaining two machines, respectively. Otherwise, we schedule CL$_{\ell}^T$ on machine $\theta_{\textrm{V-C-D},\ell,\ell-1,\kappa_{-1}}$ and the two false copies CL$_{\ell}^F$ on machines $\theta_{\textrm{V-C-D},\ell,i,\kappa_{i-\ell}}$ where $i=\ell,\ell+1$.

Finally, we determine the true/false version of dummy jobs on variable-clause-dummy and variable-dummy machines. Recall that there are $n+n/3$ true dummy and $n-n/3$ false dummy jobs. 

On variable-clause-dummy machines, if the clause job is true, schedule a true dummy job. Otherwise, the clause job is false, then if the variable job is true (or false), schedule a false (or true) dummy job.

On variable-dummy machines, we schedule dummy jobs in the following way. A false variable job is always scheduled with a true dummy job. For true variable jobs, 
we first partition the indices of variables, $\{1,2,\cdots,n\}$, into two subsets $S_1,S_2$ such that 
$$S_1=\{i: {\textrm{On machine $\theta_{\textrm{V-C-D},\ell,i,\kappa_{i-\ell}}$ there is a true clause job and a false variable job}}\},$$
and $S_2$ consists of the remaining indices. On machine $\theta_{\textrm{V-D},i,+}$ or $\theta_{\textrm{V-D},i,-}$ where $i\in S_1$ and the variable job is true, we schedule a true dummy job; on machine $\theta_{\textrm{V-D},i,+}$ or $\theta_{\textrm{V-D},i,-}$ where $i\not\in S_1$ and the variable job is true, we schedule a false dummy job. 

Consider the true/false versions of all two jobs on a variable-dummy machine and use $(T/F,T/F)$ to denote the true/false version of the two jobs in the order of variable job, dummy job. Then the above scheduling can be restated as follows. A variable-dummy machine $\theta_{\textrm{V-D},i,+}$ or $\theta_{\textrm{V-D},i,+}$ is:
\begin{itemize}
    \item $(F,T)$, if a false variable job is on it;
    \item $(T,F)$, if a true variable job is on it and $i\not\in S_1$;
    \item $(T,T)$, if a true variable job is on it and $i\in S_1$.
\end{itemize}
Hence there are in total three kinds of variable-dummy machines $(F,T), (T,T), (T,F)$.

Now we check the total number of true and false dummy jobs scheduled in the above way. Similarly we consider the true/false versions of all three jobs on a variable-clause-dummy machine and use $(T/F,T/F,T/F)$ to denote the true/false versions of the three jobs in the order of variable job, clause job and dummy job, then there are in total four kinds of variable-clause-dummy machines: $(F,T,T), (T,T,T), (T,F,F), (F,F,T)$.  Let $\sharp(T/F,T/F,T/F)$ and $\sharp(T/F,T/F)$ be the number of machines of each kind. Then we have the following observations:
\begin{subequations}
\begin{eqnarray}
&&\sharp(F,T,T)=|S_1| \label{ILP:1}\\
&&\sharp(F,T,T)+\sharp(T,T,T)=n/3 \label{ILP:2}\\
&&\sharp(T,F,F)+\sharp(F,F,T)=2n/3 \label{ILP:3}\\
&&\sharp(F,T)+\sharp(T,T)+\sharp(T,F)=n \label{ILP:4}\\
&&\sharp(F,T,T)+\sharp(F,F,T)+\sharp(F,T)=n \label{ILP:5}\\
&&\sharp(T,T)=|S_1|\label{ILP:6}
\end{eqnarray}
\end{subequations}

Here Eq~\eqref{ILP:1} follows from the definition of $S_1$. Eq~\eqref{ILP:2} follows from the fact that there are in total $n/3$ true clause jobs. Eq~\eqref{ILP:3} follows from the fact that there are in total $n$ clause jobs, and hence $2n/3$ false clause jobs. Eq~\eqref{ILP:4} follows from the fact that there are in total $n$ variable-dummy machines. Eq~\eqref{ILP:5} follows from the fact that there are in total $n$ false variable jobs. We now explain Eq~\eqref{ILP:6}. Notice that for each $i$ there are in total $8$ variable jobs (i.e., $V_{i,\cdot,\cdot}$), 4 true copies and 4 false copies. Among them 2 true and 2 false copies are scheduled on variable-truth machines, 1 true and 1 false copies are scheduled on variable-link machines (see Table~\ref{table:job-variable-true}). Hence, 1 true and 1 false copies are scheduled on variable-clause-dummy and variable-dummy machines. For any $i$, if the false (or true) variable job $V_{i,\cdot,\cdot}$ is scheduled on a variable-clause-dummy machine, then the remaining true (or false) variable job is scheduled on a variable-dummy machine.  Now consider the set of all $i$'s where the true variable job $V_{i,\cdot,\cdot}$ is scheduled with a true dummy job on a variable-dummy machine and let it be $S_3$. According to the way we schedule, on machine $\theta_{\textrm{V-D},i,+}$ or $\theta_{\textrm{V-D},i,-}$, we schedule a true variable job and a true dummy job only if $i\in S_1$ (otherwise, either the variable job or the dummy job is false), hence $S_3\subseteq S_1$. Meanwhile, for any $i\in S_1$, we know the false variable job $V_{i,\cdot,\cdot}$ is scheduled on a variable-clause-dummy machine, whereas the true variable job must be scheduled on a variable-dummy machine, this implies that any $i\in S_1$ also satisfies that $i\in S_3$. Hence $S_1=S_3$ and Eq~\eqref{ILP:6} is true.

The total number of true dummy jobs scheduled equals $\sharp(F,T,T)+\sharp(T,T,T)+\sharp(F,F,T)+\sharp(F,T)+\sharp(T,T)=n+\sharp(T,T,T)+\sharp(T,T)=n+n/3-|S_1|+|S_1|=4n/3$. Similarly, we can show the total number of false jobs scheduled equals $2n/3$. Hence, our way of scheduling dummy jobs is feasible.

Now we check the load of every variable-clause-dummy machines and variable-dummy machines. It is easy to verify that for a variable-clause-dummy machine, if its kind is $(T,T,T)$, or $(T,F,F)$, or $(F,F,T)$, then its load is $10^{14}\sigma_{max}$; if its kind is $(F,T,T)$, then its load is $10^{14}\sigma_{max}+1$. For a variable-dummy machine, if its kind is $(F,T)$ or $(T,F)$, then its load is $10^{14}\sigma_{max}$; if its kind is $(T,T)$, then its load is $10^{14}\sigma_{max}-1$. 

Notice that for every $1\le i\le n$, the variable-dummy machine $\theta_{\textrm{V-D},i,\cdot}$ is of $(T,T)$ if and only if the variable-clause-dummy machine $\theta_{\textrm{V-C-D},\ell,i,\cdot}$ is of $(F,T,T)$. Recall that we always try to schedule the true clause job CL$_{\ell}^T$ with a true variable job, if possible. Hence, CL$_{\ell}^T$ is scheduled with a false variable job if and only if all the three variable jobs scheduled on variable-clause-dummy machines, i.e., $V_{\ell-1,\kappa_{-1},1}$, $V_{\ell,\kappa_0,1}$ and $V_{\ell+1,\kappa_1,1}$, are all false where $\kappa_{-1},\kappa_{0},\kappa_1\in\{+,-\}$. Consider $V_{\ell-1,\kappa_{-1},1}$. If $\kappa_{-1}=+$, then $\theta_{\textrm{V-C-D},\ell,\ell-1,+}$ exists, indicating case 1 of Table~\ref{table:job-variable-true} or Table~\ref{table:job-variable-false} occurs, i.e., the positive literal $z_{\ell-1}$ is in clause $cl_\ell\in C_1$. Furthermore, as $V_{\ell-1,+,1}^F$ is scheduled on the variable-clause-dummy machine, the scheduling follows Table~\ref{table:job-variable-false}, the variable $z_{\ell-1}$ is false in the assignment of $I_{sat}$. That is, $cl_\ell$ is not satisfied by $z_{\ell-1}$. Similarly, we can show that if $\kappa_{-1}=-$, then the negative literal $\neg z_{\ell-1}$ is in $cl_\ell$ and variable $z_i$ is true, whereas $cl_\ell$ is not satisfied by $z_{\ell-1}$, either. Using the same argument, we can show that if all three jobs $V_{\ell-1,\kappa_{-1},1}$, $V_{\ell,\kappa_0,1}$ and $V_{\ell+1,\kappa_1,1}$ scheduled together with CL$_\ell$ are all false, then $cl_\ell$ is not satisfied by the assignment. Furthermore, according to our scheduling method, if we cannot schedule CL$_{\ell}^T$ with a true variable job, we schedule it with the false job $V_{\ell-1,\kappa_{-1},1}^F$. That means, among the three machines $\theta_{\textrm{V-C-D},\ell,i,\cdot}$, only $\theta_{\textrm{V-C-D},\ell-1,i,\cdot}$ is of kind $(F,T,T)$ and has a load of $10^{14}\sigma_{max}+1$. The other two machines have a load of $10^{14}\sigma_{max}$. Similarly, we check variable-dummy machines and see that among the three machines $\theta_{\textrm{V-D},i,\cdot}$ where $i\in\{\ell-1,\ell,\ell+1\}$, only machine $\theta_{\textrm{V-D},\ell-1,\cdot}$ is of kind $(T,T)$ and has a load of $10^{14}\sigma_{max}-1$. The other two machines have a load of $10^{14}\sigma_{max}$.

According to our observation in the above paragraph, we have the following claim.
\begin{claim}\label{claim:sat-to-schedule-2}
If $cl_\ell\in C_1$ is satisfied, then the three clause-variable-dummy machines $\theta_{\textrm{V-C-D},i,\ell,\cdot}$ and the three variable-dummy machines $\theta_{\textrm{V-D},i,\cdot}$, $i\in\{\ell-1,\ell,\ell+1\}$ all have a load of $10^{14}\sigma_{max}$; otherwise, machine $\theta_{\textrm{V-C-D},\ell-1,\ell,\cdot}$ has a load of $10^{14}\sigma_{max}+1$, $\theta_{\textrm{V-D},\ell-1,\cdot}$ has a load of $10^{14}\sigma_{max}-1$, and all the remaining 4 machines have a load of $10^{14}\sigma_{max}$.
\end{claim}

Combining Claim~\ref{claim:sat-to-schedule-1} amd Claim~\ref{claim:sat-to-schedule-2}, we know that each unsatisfied clause can lead to at most $2$ machines with load $10^{14}\sigma_{max}\pm 1$. Recall that the total processing time of all jobs is $10^{14}\sigma_{max}\cdot (2\gamma n+8n)$, hence the number of machines with load $10^{14}\sigma_{max}+ 1$ should equal the number of machines with load $10^{14}\sigma_{max}-1$. Consequently, if there are $\vartheta n$ unsatisfied clauses, the resulted schedule will contain at most $2\vartheta n$ machines with load $10^{14}\sigma_{max}\pm 1$. Using Taylor's expression, we have that

\begin{align*}
    (x+1)^q+(x-1)^q &= x^q[(1+\frac{1}{x})^q+(1-\frac{1}{x})^q] \\
    &= x^q(1+\frac{q(q-1)}{2x^2}+o(\frac{1}{x^2}))\\
    &= x^q+\frac{q(q-1)}{2}\cdot x^{q-2}+o(x^{q-2}),
\end{align*}
Hence, by simple calculations Lemma~\ref{lemma:sat-sche} is proved.

\subsection{Scheduling to 3SAT\texorpdfstring{$'$}{Lg}}
\label{subsec:sche-to-sat}
The goal of this subsection is to show that if the constructed scheduling instance admits a feasible schedule of a small objective value, then the given 3SAT$'$ instance admits a truth-assignment that satisfies most clauses. More precisely, we prove the following lemma.
\begin{lemma}\label{lemma:sche-sat}
If there are at least $\vartheta n$ clauses not satisfied, then any feasible schedule has an objective value at least $(2\gamma n+8n)(10^{14}\sigma_{max})^q+\frac{q(q-1)\vartheta n}{48}\cdot (10^{14}\sigma_{max})^{q-2}+o(n\sigma_{max}^{q-2})$.
\end{lemma}

In the following we consider a solution $Sol$ for scheduling whose objective value is bounded by $(2\gamma n+8n)(10^{14}\sigma_{max})^q+(\frac{q(q-1)\vartheta n}{48}-\epsilon')\cdot (10^{14}\sigma_{max})^{q-2}$ for arbitrarily small $\epsilon'>0$.

Recall that we have constructed in total $2\gamma n+8n$ machines. According to Subsection~\ref{subsec:sat-to-scheduling-simple}, the total processing time of all jobs is $(2\gamma n+8n)\cdot 10^{14}\sigma_{max}$. Consider an arbitrary schedule. We say a machine is good if its load is exactly $10^{14}\sigma_{max}$; otherwise, the machine is bad. Since the processing times are half-integral (multiples of $1/2$), the load of a bad machine is either no larger than $10^{14}\sigma_{max}-0.5$, or no less than $10^{14}\sigma_{max}+0.5$. Furthermore, we say a machine is very bad if its load deviates from $10^{14}\sigma_{max}$ by at least $\sigma_{max}$, i.e., the load of a very bad machine is either no larger than $(10^{14}-1)\sigma_{max}$, or no smaller than $(10^{14}+1)\sigma_{max}$.

\begin{lemma}\label{lemma:almost-good-schedule}
If there exists a very bad machine, then the objective value of the schedule is at least $m(10^{14}\sigma_{max})^q+c_1\sigma_{max}^{q}$ for some constant $c_1>0$. 
\end{lemma}

Towards the proof, we need the following lemma.
\begin{lemma}
For $x,q,m> 1$ and $k\ge 1$, it holds that
$$(x-k)^q+(m-1)(x+\frac{k}{m-1})^q\ge mx^q+\frac{q(q-1)}{4}\min\{(x-1)^{q-2},(x+1)^{q-2}\}.$$
\end{lemma}
\begin{proof}
Taking the derivative of $(x-k)^q+(m-1)(x+\frac{k}{m-1})^q$ with respect to $k$, we get $-q(x-k)^{k-1}+q(x+\frac{k}{m-1})^{q-1}> 0$ when $x,q,m>1$ and $k\ge 1$, hence the function $(x-k)^q+(m-1)(x+\frac{k}{m-1})^q$ is an increasing function of $k$, thus it suffices to prove the lemma for $k=1$. According to the mean value theorem, we have
\begin{eqnarray*}
&\Gamma:=&\frac{1}{m}(x-1)^q+\frac{m-1}{m}(x+\frac{1}{m-1})^q-x^q\\
&=&-\frac{1}{m}[x^q-(x-1)^q]+\frac{m-1}{m}[(x+\frac{1}{m-1})^q-x^q]\\
&=&-\frac{1}{{m}}[x^q-(x-\frac{1}{2})^q]-\frac{1}{{m}}[(x-\frac{1}{2})^q-(x-1)^q]+\frac{m-1}{m}[(x+\frac{1}{m-1})^q-x^q]\\
&=& -\frac{1}{2m} q(x-\theta_1)^{q-1}-\frac{1}{2m} q(x-\frac{1}{2}-\theta_2)^{q-1}+\frac{m-1}{m} \cdot \frac{1}{m-1}\cdot q(x+\theta_3)^{q-1}
\end{eqnarray*}
for some $\theta_1,\theta_2\in (0,1/2)$ and $\theta_3\in (0,\frac{1}{m-1})$. Further apply the mean value theorem, we have
\begin{eqnarray*}
\Gamma&\ge&\frac{q}{2m}[(x+\theta_3)^{q-1}-(x-\frac{1}{2}-\theta_2)^{q-1}]\\
&=&\frac{q}{2m}(\theta_2+\theta_3+\frac{1}{2})(q-1)(x+\theta_4)^{q-2}
\end{eqnarray*}
for some $\theta_4\in (-1/2-\theta_2,\theta_3)$. If $q\ge 2$, then $(x+\theta_4)^{q-2}\ge (x-1)^{q-2}$. Otherwise $1<q<2$ and it holds that $(x+\theta_4)^{q-2}\ge (x+1)^{q-2}$. Thus
$$\Gamma\ge \frac{q(q-1)}{4m} \min\{ (x-1)^{q-2},(x+1)^{q-2}\}.$$

Hence, the lemma is proved.
\end{proof}

Similarly, we can prove that
\begin{lemma}
For $x,q,m> 1$ and $k\ge 1$, it holds that
$$(x+k)^q+(m-1)(x-\frac{k}{m-1})^q\ge mx^q+\frac{q(q-1)}{4}\min\{(x-1)^{q-2},(x+1)^{q-2}\}.$$
\end{lemma}

Now we are ready to prove Lemma~\ref{lemma:almost-good-schedule}.
\begin{proof}[Proof of Lemma~\ref{lemma:almost-good-schedule}]
Suppose the load of one very bad machine is $(10^{14}-k)\sigma_{max}$ for some $|k|\ge 1$, then total load of all other machines is $10^{14}(m-1)\sigma_{max}+k\sigma_{max}$. By the convexity of the function $x^q$, the objective value of such a solution is at least:
\begin{eqnarray*}
&&(10^{14}-k)^q\sigma_{max}^q+(m-1)\cdot [\frac{10^{14}(m-1)\sigma_{max}+k\sigma_{max}}{m-1}]^q\\
&=&\sigma_{max}^q[(10^{14}+k)^q+(m-1)(10^{14}+\frac{k}{m-1})^q]\\
&\ge&m(10^{14}\sigma_{max})^q+\sigma_{max}^q\cdot \frac{q(q-1)}{4}\min\{(10^{14}-1)^{q-2},(10^{14}+1)^{q-2}\}
\end{eqnarray*}

Hence, the lemma is proved.
\end{proof}

We have shown that if a schedule admits a very bad machine, then its objective is significantly large and cannot be $Sol$. To prove Lemma~\ref{lemma:sche-sat}, it suffices to restrict our attention to schedules without any very bad machine. 

Notice that the processing time of a gap job is at least $(10^{14}-2\times 10^{13})\sigma_{max}$, we know that there can be at most one gap job on a machine that is not very bad. Given the fact that the total number of gap jobs equals the number of machines, and there is no very bad machine in $Sol$, we have the following observation.
\begin{lemma}\label{lemma:structure1}
There is exactly one gap job on each machine in $Sol$.
\end{lemma}

Given Lemma~\ref{lemma:structure1}, we will use the symbol of a gap job, e.g., $\theta_{\textrm{V-L},i,+}$, to denote the machine on which this job is scheduled.

The following lemma is straightforward by observing that $\sigma_{max}>x\cdot \sigma(i)$ for all $i$, and hence the type coordinates (i.e., the term $10^j\sigma_{max}$) of jobs on a machine that is not very bad cannot add up to smaller than $(10^{14}-2)\sigma_{max}$ or larger than $(10^{14}+2)\sigma_{max}$. 

\begin{lemma}\label{lemma:structure2}
If in a solution there is no very bad machine, then 
\begin{itemize}
\item On a variable-link machine {$\theta_{\textrm{V-L},i,\iota}$ where $\iota\in\{+,-\}$}, there are exactly three jobs -- a gap job, a variable job and a link job. 
\item On a link-link machine $\theta_{\textrm{L-L},i,h,\iota}$ where $\iota\in\{+,-\}$, there are exactly three jobs -- a gap job and two link jobs.  
\item On a variable-dummy machine $\theta_{\textrm{V-D},i,\iota}$ where $\iota\in\{+,-\}$, there are exactly three jobs -- a gap job, a variable job and a dummy job.
\item On a variable-clause-dummy machine $\theta_{\textrm{V-C-D},\ell,i,\iota}$ where $\iota\in\{+,-\}$, there are exactly four jobs -- a gap job, a variable job, a clause job and a dummy job.
\item On a variable-truth machine $\theta_{\textrm{V-T},i,\rho}$ where $\rho\in\{(a,c),(b,d),(a,d),(b,c)\}$, there are exactly four jobs -- a gap job, a variable job and two truth-assignment jobs; Furthermore, the two truth-assignment jobs are: 
\begin{itemize}
\item $TR_{\cdot,a}$ and $TR_{\cdot,c}$ $\quad$ if $\rho=(a,c)$;
\item $TR_{\cdot,b}$ and $TR_{\cdot,d}$ $\quad$ if $\rho=(b,d)$;
\item $TR_{\cdot,a}$ and $TR_{\cdot,d}$ $\quad$ if $\rho=(a,d)$;
\item $TR_{\cdot,b}$ and $TR_{\cdot,c}$ $\quad$ if $\rho=(b,c)$.
\end{itemize}
\end{itemize}
\end{lemma}
\begin{proof}
The proof can be carried out through a counting argument in the order of dummy jobs, clause jobs, truth-assignment jobs, link jobs and variable jobs according to Table~\ref{table:job-time}. In the following, we prove dummy jobs and the other types of jobs can be proved in a similar way. The reader may refer to Table~\ref{table:job-time-whole} for a quick overview on job processing times. Note that a dummy job has a processing time at least $(10^{13}-1/2)\sigma_{max}$. It is easy to see that if a variable-link machine, or link-link machine, or variable-truth machine accepts one dummy job, then the load of this machine is larger than $(10^{14}+1)\sigma_{max}$, contradicting the fact that there is no very bad machine. Hence, dummy jobs can only be scheduled on variable-clause-dummy machine or a variable-dummy machine. Similarly, if a variable-clause-dummy machine or a variable-dummy machine accepts two or more dummy jobs, its load becomes larger than $(10^{14}+1)\sigma_{max}$, hence each of these machines can accept at most 1 dummy job. On the other hand, there are $2n$ dummy jobs, which is equal to the sum of the number of variable-clause-dummy machines (which is $n$) and the number of variable-dummy machines (which is also $n$). Hence, each variable-clause-dummy machine or variable-dummy machine accepts exactly one dummy job. Subtracting one dummy job together with the gap job on each variable-clause-dummy machine or variable-dummy machine, we know that if the machine is not very bad, then the remaining jobs on a variable-clause-dummy machine should add up to some value within $[(10^{12}+10^5-2)\sigma_{max},(10^{12}+10^5+2)\sigma_{max}]$ (if this machine is $\theta_{\textrm{V-C-D},\ell,i,+}$) or $[(10^{12}+10^3-2)\sigma_{max},(10^{12}+10^3+2)\sigma_{max}]$ (if this machine is $\theta_{\textrm{V-C-D},\ell,i,-}$), and the remaining jobs on a variable-dummy machine should add up to some value within $[(10^5-2)\sigma_{max},(10^5+2)\sigma_{max}]$ (if this machine is $\theta_{\textrm{V-D},i,+}$) or $[(10^3-2)\sigma_{max},(10^3+2)\sigma_{max}]$ (if this machine is $\theta_{\textrm{V-D},i,-}$). Consequently, we can apply the same argument to clause jobs, and then truth-assignment jobs, then link jobs and then variable jobs. 
\end{proof}

Using Lemma~\ref{lemma:structure2}, we further have the following observation.

\begin{lemma}\label{lemma:term-cancel}
On a good machine, the type-component of jobs add up to $10^{14}\sigma_{max}$, the index-components and the true/false-components of jobs add up to $0$, respectively.
\end{lemma}

Now we further identify the index-component of jobs on each machine.

\begin{lemma}\label{lemma:vd-index}
Consider an arbitrary variable-dummy machine $\theta_{\textrm{V-D},i,\iota}$ where $\iota\in\{+,-\}$. If the machine is good, then the variable job on this machine is $V_{i,\iota,2}$.
\end{lemma}

Applying Lemma~\ref{lemma:term-cancel}, the proof is straightforward by checking the sum of type-components and index-components of jobs, respectively.

\begin{lemma}\label{lemma:vcd-index}
Consider an arbitrary variable-clause-dummy machine $\theta_{\textrm{V-C-D},\ell,i,\iota}$ where $\iota\in\{+,-\}$. If the machine is good, then the clause job on this machine is $\textrm{CL}_{\ell}$, and the variable job on this machine is $V_{i,\iota,1}$.
\end{lemma}
\begin{proof}
By Lemma~\ref{lemma:term-cancel}, the type-components of the three jobs add up to $10^{14}\sigma_{max}$, hence it is easy to see that the variable job should be $V_{i',\iota,1}$ for some $i'$. Let the clause job be $\textrm{CL}_{\ell'}$ for some $\ell'$. As the index-components of the three jobs add up to $0$, we have 
$$\sigma(\ell')+\sigma(i')=\sigma(\ell)+\sigma(i).$$

Notice that for any machine $\theta_{\textrm{V-C-D},\ell,i,\iota}$ it holds that $i\in\{\ell-1,\ell,\ell+1\}$ and $\ell\in\{2,5,\cdots,n-1\}$. We claim that $\ell'=\ell$ and $i'=i$. To see why, consider two cases. If $\ell=i$, then $\sigma(\ell')+\sigma(i')=2\sigma(\ell)$. According to Lemma~\ref{lemma:uniquesum-2}, we have $\ell'=i'=\ell=i$ and the claim follows. Otherwise, $i=\ell\pm 1$. 
According to Lemma~\ref{lemma:uniquesum-2}, the only solution for $\sigma(j)+\sigma(j+1)=\sum_{h=1}^k\sigma(j_h)$, $k\le 5$, is $k=2$ and $\{j_1,j_2\}=\{j,j+1\}$. Hence, we have $\{\ell',i'\}=\{\ell,i\}$. Note that $\ell',\ell\in\{2,5,\cdots,n-1\}$, hence $\ell',\ell\equiv 2 (\mod 3)$. But $i\not\equiv 2 (\mod 3)$. Thus, $\ell=\ell'$ and $i=i'$. In both cases, Lemma~\ref{lemma:vcd-index} holds.
\end{proof}

\begin{lemma}\label{lemma:vtd-index}
Consider an arbitrary variable-truth machine $\theta_{\textrm{V-T},i,\rho}$ where $\rho\in\{(a,c),(b,d),(a,d),\\(b,c)\}$. If the machine is good, then the variable and truth-assignment jobs are:
\begin{itemize}
\item $V_{i,+,1}$ and $TR_{i,a}$, $TR_{i,c}$ $\quad$ if $\rho=(a,c)$;
\item $V_{i,+,2}$ and $TR_{i,b}$, $TR_{i,d}$ $\quad$ if $\rho=(b,d)$;
\item $V_{i,-,1}$ and $TR_{i,a}$, $TR_{i,d}$ $\quad$ if $\rho=(a,d)$;
\item $V_{i,-,2}$ and $TR_{i,b}$, $TR_{i,c}$ $\quad$ if $\rho=(b,c)$.
\end{itemize}
\end{lemma}
\begin{proof}
According to Lemma~\ref{lemma:term-cancel}, the type-components of jobs add up to $10^{14}\sigma_{max}$. Hence, it is easy to verify that for some $i_1,i_2,i_3$ the variable and truth-assignment jobs are $V_{i_1,+,1}$ and $TR_{i_2,a}$, $TR_{i_3,c}$, if $\rho=(a,c)$; $V_{i_1,+,2}$ and $TR_{i_2,b}$, $TR_{i_3,d}$, if $\rho=(b,d)$; $V_{i_1,-,1}$ and $TR_{i_2,a}$, $TR_{i_3,d}$, if $\rho=(a,d)$; $V_{i_1,-,2}$ and $TR_{i_2,b}$, $TR_{i_3,c}$, if $\rho=(b,c)$.

We prove $i_1=i_2=i_3=i$, and Lemma~\ref{lemma:vtd-index} follows. 
According to Lemma~\ref{lemma:term-cancel}, the index-components add up to $0$, hence 
$$10\sigma(i_1)+10\sigma(i_2)+10\sigma(i_3)=30\sigma(i),$$
i.e., $\sigma(i_1)+\sigma(i_2)+\sigma(i_3)=3\sigma(i)$. 
According to Lemma~\ref{lemma:uniquesum-2}, the above equation has a unique solution, which is $i_1=i_2=i_3=i$.
\end{proof}

\begin{lemma}\label{lemma:vl-index-+}
Consider an arbitrary variable-link machine $\theta_{\textrm{V-L},i,\iota}$ where $\iota\in\{+,-\}$. If the machine is good, then the variable job on this machine is $V_{i,\iota,2}$, and the link job on this machine is $\textrm{LN}_{i,1,\iota}$.
\end{lemma}
\begin{proof}
Using the fact that the type-components of all jobs add up to $10^{14}\sigma_{max}$, it is easy to see that the variable job should be $V_{i_1,\iota,2}$ and the link job should be $\textrm{LN}_{i_2,h,\iota}$ for some $1\le i_1,i_2\le n$ and $1\le h\le 2\gamma+2$. We prove the lemma for $\iota=+$. The case that $\iota=-$ can be proved in the same way.  

Given that the index-components should add up to $0$, we have the following: 
\begin{eqnarray}
\vea_i+\veb^1_{i}=\vea_{i_1}+\veb^h_{i_2}
\end{eqnarray}

Recall that $\vea_i=\veb^0_{i}$. According to Lemma~\ref{lemma:veb-uniquesum-2}, we have $h=1$ and $i_1=i_2=i$.
\end{proof}

\begin{lemma}\label{lemma:ll-index-+}
Consider an arbitrary link-link machine $\theta_{\textrm{L-L},i,h,\iota}$ where $\iota\in\{+,-\}$, $1\le h\le 2\gamma+1$. If the machine is good, then the two link jobs on this machine are $\textrm{LN}_{i,h,\iota}$ and  $\textrm{LN}_{i,h+1,\iota}$.
\end{lemma}
The proof is similar to that of Lemma~\ref{lemma:vl-index-+} by utilizing Lemma~\ref{lemma:veb-uniquesum-2}.

\begin{lemma}\label{lemma:ll-index--}
Consider an arbitrary link-link machine $\theta_{\textrm{L-L},i,+,-}$. If the machine is good, then the two link jobs on this machine are $\textrm{LN}_{i,2\gamma+2,+}$ and  $\textrm{LN}_{\tau(i),2\gamma+2,-}$.
\end{lemma}
\begin{proof}
Using the fact that the type-components of all jobs add up to $10^{14}\sigma_{max}$, it is easy to see that the two link jobs should LN$_{i_1,h_1,+}$ and LN$_{i_2,h_2,-}$. Given that the index-components should add up to $0$, we have the following: 
\begin{eqnarray}
2\veb^{2\gamma+2}_i=\veb^{h_1}_{i_1}+\hat{\veb}^{h_2}_{i_2}
\end{eqnarray}

According to Lemma~\ref{lemma:veb-uniquesum-3}, we have $i_1=i$ and $i_2=\tau(i)$, and $h_1=h_2=2\gamma+2$.
\end{proof}

We have proved, so far, that if a machine is good, then the jobs scheduled on it must follow Table~\ref{table:job-schedule}.
Finally we consider the true/false-components of jobs on good machines. Based on the T/F-component of jobs, the following lemma is easy to verify.
\begin{lemma}\label{lemma:true-type}
The followings are true:
\begin{itemize}
\item If a variable-link machine is good, then the T/F-type of the variable job and link job on this machine is $(T,F)$ or $(F,T)$;
\item If a link-link machine is good, then the T/F-type of the two link jobs on this machine is $(T,F)$ or $(F,T)$;
\item If a variable-clause-dummy machine is good, then the T/F-type of the variable job, clause job and dummy job on this machine is $(T,T,T)$ or $(T,F,F)$ or $(F,F,T)$;
\item If a variable-dummy machine is good, then the T/F-type of the variable job and dummy job on this machine is $(T,F)$ or $(F,T)$;
\item If a variable-truth machine is good, then the T/F-type of the variable job and two truth-assignment jobs on this machine is $(T,T,T)$ or $(F,F,F)$;
\end{itemize}
\end{lemma}

\subsubsection{Truth-assignment based on scheduling} 
Given a feasible schedule $Sol$, we give a truth-assignment of $I_{sat}$ as follows: if the job $V_{i,+,1}^T$ is scheduled on machine $\theta_{\textrm{V-T},i,a,c}$, then we let variable $z_i$ be false; if the job $V_{i,+,1}^F$ is scheduled on machine $\theta_{\textrm{V-T},i,a,c}$, then we let variable $z_i$ be true. If $V_{i,+,1}$ is not scheduled on machine $\theta_{\textrm{V-T},i,a,c}$, we let $z_i$ be true.

We call the machines in the following Table~\ref{table:variable-job-true} as machines of group $i$. Notice that groups are not disjoint, particularly machines $\theta_{\textrm{L-L},i,+,-}$, $\theta_{\textrm{L-L},\tau^{-1}(i),+,-}$ will appear in two groups. Besides the two machines, all other machines in a group do not appear in other groups. We have the following lemma.

\begin{table}[!ht]
\renewcommand\arraystretch{1.3}
\setlength\tabcolsep{1.5pt}
\centering
\resizebox{12cm}{!}{
\begin{tabular}{|c|l|l|l|l|l|}
\hline
\multirow{2}{*}{Variable-Link} & $\theta_{\textrm{V-L},i,+}$ & $V_{i,+,2}^T$  & LN$_{i,1,+}^F$ & $\backslash$  \\ \cline{2-5} 
                  & $\theta_{\textrm{V-L},i,-}$ & $V_{i,-,2}^F$ & LN$_{i,1,-}^{T}$ & $\backslash$  \\ \hline
\multirow{4}{*}{\makecell{Link-Link\\$h\in\{1,2,\cdots,2\gamma+1\}$}} & {$\theta_{\textrm{L-L},i,h,+}$} & LN$_{i,h,+}^F$ & LN$_{i,h+1,+}^T$ & $\backslash$  \\ \cline{2-5} 
                  & {$\theta_{\textrm{L-L},i,h,-}$} & LN$_{i,h,-}^F$ & LN$_{i,h+1,-}^T$ & $\backslash$    \\ \cline{2-5} 
                  & $\theta_{\textrm{L-L},i,+,-}$ & LN$_{i,2\gamma+2,+}^T$ & LN$_{\tau(i),2\gamma+2,-}^F$ & $\backslash$  \\ \cline{2-5} 
                  & $\theta_{\textrm{L-L},\tau^{-1}(i),+,-}$ & LN$_{\tau^{-1}(i),2\gamma+2,+}^T$ & LN$_{i,2\gamma+2,-}^F$ & $\backslash$  \\ \cline{2-5}\hline
\multirow{2}{*}{\makecell{Variable-Clause-Dummy \& Variable-Dummy\\Case 1: positive literal $z_i\in C_1$}} & $\theta_{\textrm{V-C-D},\ell,i,+}$ & $V_{i,+,1}^T$ & CL$_{\ell}^{*}$ & DM$^{*}$  \\ \cline{2-5} 
                  & $\theta_{\textrm{V-D},i,-}$ & $V_{i,-,1}^F$ & DM$^T$ & $\backslash$  \\ \hline
\multirow{2}{*}{\makecell{Variable-Clause-Dummy \& Variable-Dummy\\Case 2: negative literal $\neg z_i\in C_1$}} & $\theta_{\textrm{V-C-D},i,-}$ & $V_{i,-,1}^F$ & CL$_{\ell}^{F}$ & DM$^T$  \\ \cline{2-5} 
                  & $\theta_{\textrm{V-D},i,+}$ & $V_{i,+,1}^T$ & DM$^F$ & $\backslash$  \\ \hline
\multirow{4}{*}{{Variable-Truth}} & $\theta_{\textrm{V-T},i,a,c}$ & $V_{i,+,1}^F$ & TR$_{i,a}^F$ & TR$_{i,c}^F$  \\ \cline{2-5} 
                  & $\theta_{\textrm{V-T},i,b,d}$ & $V_{i,+,2}^F$ & TR$_{i,b}^F$ & TR$_{i,d}^F$  \\ \cline{2-5} 
                  & $\theta_{\textrm{V-T},i,a,d}$ & $V_{i,-,1}^T$ & TR$_{i,a}^T$ & TR$_{i,d}^T$  \\ \cline{2-5} 
                  & $\theta_{\textrm{V-T},i,b,c}$ & $V_{i,-,2}^T$ & TR$_{i,b}^T$ & TR$_{i,c}^T$  \\ \hline
\end{tabular}
}
\caption{Scheduling of group $i$ machines when $V_{i,+,1}^F$ is scheduled on machine $\theta_{\textrm{V-T},i,a,c}$ (and we set variable $z_i$ to be true)}
\label{table:variable-job-true}
\end{table}

\begin{table}[!ht]
\renewcommand\arraystretch{1.3}
\setlength\tabcolsep{1.5pt}
\centering
\resizebox{12cm}{!}{
\begin{tabular}{|c|l|l|l|l|l|}
\hline
\multirow{2}{*}{Variable-Link} & $\theta_{\textrm{V-L},i,+}$ & $V_{i,+,2}^F$  & LN$_{i,1,+}^T$ & $\backslash$  \\ \cline{2-5} 
                  & $\theta_{\textrm{V-L},i,-}$ & $V_{i,-,2}^T$ & LN$_{i,1,-}^{F}$ & $\backslash$  \\ \hline
\multirow{4}{*}{\makecell{Link-Link\\$h\in\{1,2,\cdots,2\gamma+1\}$}} & {$\theta_{\textrm{L-L},i,h,+}$} & LN$_{i,h,+}^T$ & LN$_{i,h+1,+}^F$ & $\backslash$  \\ \cline{2-5} 
                  & {$\theta_{\textrm{L-L},i,h,-}$} & LN$_{i,h,-}^T$ & LN$_{i,h+1,-}^F$ & $\backslash$    \\ \cline{2-5} 
                  & $\theta_{\textrm{L-L},i,+,-}$ & LN$_{i,2\gamma+2,+}^F$ & LN$_{\tau(i),2\gamma+2,-}^T$ & $\backslash$  \\ \cline{2-5} 
                  & $\theta_{\textrm{L-L},\tau^{-1}(i),+,-}$ & LN$_{\tau^{-1}(i),2\gamma+2,+}^F$ & LN$_{i,2\gamma+2,-}^T$ & $\backslash$  \\ \cline{2-5}\hline
\multirow{2}{*}{\makecell{Variable-Clause-Dummy \& Variable-Dummy\\Case 1: positive literal $z_i\in C_1$}} & $\theta_{\textrm{V-C-D},\ell,i,+}$ & $V_{i,+,1}^F$ & CL$_{\ell}^{F}$ & DM$^{T}$  \\ \cline{2-5} 
                  & $\theta_{\textrm{V-D},i,-}$ & $V_{i,-,1}^T$ & DM$^F$ & $\backslash$  \\ \hline
\multirow{2}{*}{\makecell{Variable-Clause-Dummy \& Variable-Dummy\\Case 2: negative literal $\neg z_i\in C_1$}} & $\theta_{\textrm{V-C-D},i,-}$ & $V_{i,-,1}^T$ & CL$_{\ell}^{*}$ & DM$^*$  \\ \cline{2-5} 
                  & $\theta_{\textrm{V-D},i,+}$ & $V_{i,+,1}^F$ & DM$^T$ & $\backslash$  \\ \hline
\multirow{4}{*}{{Variable-Truth}} & $\theta_{\textrm{V-T},i,a,c}$ & $V_{i,+,1}^T$ & TR$_{i,a}^T$ & TR$_{i,c}^T$  \\ \cline{2-5} 
                  & $\theta_{\textrm{V-T},i,b,d}$ & $V_{i,+,2}^T$ & TR$_{i,b}^T$ & TR$_{i,d}^T$  \\ \cline{2-5} 
                  & $\theta_{\textrm{V-T},i,a,d}$ & $V_{i,-,1}^F$ & TR$_{i,a}^F$ & TR$_{i,d}^F$  \\ \cline{2-5} 
                  & $\theta_{\textrm{V-T},i,b,c}$ & $V_{i,-,2}^F$ & TR$_{i,b}^F$ & TR$_{i,c}^F$  \\ \hline
\end{tabular}
}
\caption{Scheduling of group $i$ machines when $V_{i,+,1}^T$ is scheduled on machine $\theta_{\textrm{V-T},i,a,c}$ (and we set variable $z_i$ to be false)}
\label{table:variable-job-false}
\end{table}

\begin{lemma}\label{lemma:groupi}
Suppose all machines in group $i$ are good. If $V_{i,+,1}^F$ is scheduled on machine $\theta_{\textrm{V-T},i,a,c}$, then the jobs scheduled on these machines are according to Table~\ref{table:variable-job-true} ; if $V_{i,+,1}^T$ is scheduled on machine $\theta_{\textrm{V-T},i,a,c}$, then the jobs scheduled on these machines are according to Table~\ref{table:variable-job-false}.
\end{lemma}
\begin{proof}
We prove the first half of Lemma~\ref{lemma:groupi}, the second half can be proved in the same way. If $V_{i,+,1}^T$ is scheduled on machine $\theta_{\textrm{V-T},i,a,c}$, then by Lemma~\ref{lemma:true-type} we know the other two jobs are TR$_{i,a}^F$ and TR$_{i,c}^F$, consequently, TR$_{i,a}^T$ is scheduled on machine $\theta_{\textrm{V-T},i,a,d}$. Using similar argument it is easy to see the jobs scheduled on the 4 variable-truth machines follow Table~\ref{table:variable-job-true}.

We consider variable-clause-dummy and variable-dummy machines. It follows that the remaining $V_{i,+,1}^T$ and $V_{i,-,1}^F$ are scheduled on these machines. The T/F type of the other jobs on these machines follow from Lemma~\ref{lemma:true-type}.

Next, we consider variable-link machines. Again by the scheduling on variable-truth machines, the remaining $V_{i,+,2}^T$ and $V_{i,-,2}^F$ are scheduled on these machines. The T/F-type of the link jobs are determined by Lemma~\ref{lemma:true-type}.

Finally we consider link-link machines. Based on the link jobs scheduled on variable-link machines and Lemma~\ref{lemma:true-type}, LN$_{i,1,+}^F$ and  LN$_{i,2,+}^T$ must be scheduled on $\theta_{\textrm{L-L},i,1,+}$, consequently the remaining LN$_{i,2,+}^F$ must be scheduled on machine $\theta_{\textrm{L-L},i,2,+}$. Iteratively carrying on the above argument we can show that jobs scheduled on machines $\theta_{\textrm{L-L},i,h,+}$ must follow Table~\ref{table:variable-job-true}. Similar arguments can be applied to machines $\theta_{\textrm{L-L},i,h,-}$, $\theta_{\textrm{L-L},i,+,-}$ and $\theta_{\textrm{L-L},\tau^{-1}(i),+,-}$.
 \end{proof}

The T/F-type of the clause job CL$_{\ell}$ is not determined in Table~\ref{table:variable-job-true}. Recall that among three copies of CL$_{\ell}$ there is one true copy CL$_{\ell}^T$. Suppose CL$_{\ell}^T$ is scheduled on group $i$ machines. If $V_{i,+,1}^T$ is scheduled on machine $\theta_{\textrm{V-T},i,a,c}$ and we set variable $z_i$ to be true, then from Table~\ref{table:variable-job-true} we know case 1 must happen, which implies that clause $cl_\ell$ is satisfied by $z_i$. If $V_{i,+,1}^F$ is scheduled on machine $\theta_{\textrm{V-T},i,a,c}$ and we set variable $z_i$ to be false, then from Table~\ref{table:variable-job-false} we know case 2 must happen, which implies that clause $cl_\ell$ is satisfied by $\neg z_i$. Hence the following lemma is true.
 
 \begin{lemma}\label{lemma:schedule-satisfy-1}
 If all machines in group $i$ are good and CL$_{\ell}^T$ is scheduled on these machines, then the clause $cl_\ell \in C_1$ that contains variable $z_i$ is satisfied by this variable. 
 \end{lemma}

Now consider clauses in $C_2$ and we have the following lemma.
\begin{lemma}\label{lemma:schedule-satisfy-2}
If all machines in group $i$ are good, and all machines in group $\tau(i)$ are also good, then the clause $(z_i\oplus \neg z_{\tau(i)})$ is satisfied.
 \end{lemma}
 \begin{proof}
 There are two possibilities. If $V_{i,+,1}^F$ is scheduled on machine $\theta_{\textrm{V-T},i,a,c}$ and we set variable $z_i$ to be true, then LN$_{\tau(i),2\gamma+2,-}^F$, implying that LN$_{\tau(i),2\gamma+2,-}^T$ is scheduled in group $\tau(i)$. By checking Table~\ref{table:variable-job-true} and Table~\ref{table:variable-job-false} for variable $z_{\tau(i)}$, it follows that Table~\ref{table:variable-job-true} is the case when LN$_{\tau(i),2\gamma+2,-}^T$ is scheduled, and consequently variable $z_{\tau(i)}$ is set to be true, whereas $(z_i\oplus \neg z_{\tau(i)})$ is satisfied. The other case when $V_{i,+,1}^T$ is scheduled on machine $\theta_{\textrm{V-T},i,a,c}$ and we set variable $z_i$ to be false can be proved in a similar way. 
 \end{proof}

\begin{lemma}\label{lemma:bound-not-good-group}
In a feasible schedule $Sol$, if there are at most $m'$ machines which are not good, then in the corresponding truth-assignment, there are at most $6m'$ clauses that are not satisfied.
\end{lemma}
\begin{proof}
We say a group is good if all machines in this group are good. According to Lemma~\ref{lemma:schedule-satisfy-1}, if a clause $cl_\ell$ in $C_1$ is not satisfied, then the group that contains the job CL$_\ell^T$ is not good, that is, there is at least one machine that is not good in this group. Hence, if there are $m_1$ clauses in $C_1$ not satisfied, then there are at least $m_1$ groups that are not good. According to Lemma~\ref{lemma:schedule-satisfy-2}, if a clause $(z_i\oplus \neg z_{\tau(i)})$ in $C_2$ is not good, then among group $i$ and group $\tau(i)$ there is at least one group which is not good. Given that each group is only involved in two clauses of $C_2$, if there are $m_2$ clauses in $C_2$ not satisfied, then there are at least $m_2/2$ groups which are not good. Hence, there are at least $\max\{m_1,m_2/2\}$ groups which are not good, given $m_1+m_2$ clauses which are not satisfied. Using the fact that $\frac{m_1+m_2}{\max\{m_1,m_2/2\}}\le 3$, we know if there are at most $m'$ machines which are not good, then there are at most $2m'$ groups which are not good, and hence there are at most $6m'$ clauses which are not satisfied.
\end{proof}

\begin{lemma}\label{lemma:calculate}
In a feasible schedule, if there are at least $m'$ machines which are not good, then its objective value is at least $(2\gamma n+8n) \times (10^{14}\sigma_{max})^2+m'/4$.
\end{lemma}
\begin{proof}
Recall that if a machine is not good, then its load is either $\ge 10^{14}\sigma_{max}+0.5$ or $\le 10^{14}\sigma_{max}-0.5$. Suppose there are $m_1'$ machines, with load $10^{14}\sigma_{max}+\mu_1$, $10^{14}\sigma_{max}+\mu_2$, $\cdots$, $10^{14}\sigma_{max}+\mu_{m_1'}$ where $\mu_j\ge 1/2$; there are $m_2'$ machines, with load $10^{14}\sigma_{max}-\nu_1$, $10^{14}\sigma_{max}-\nu_2$, $\cdots$, $10^{14}\sigma_{max}+\nu_{m_2'}$ where $\nu_j\ge 1/2$. It follows that $\sum_j\mu_j=\sum_j\nu_j$ and $m_1'+m_2'=m'$. The objective value of the schedule is
\begin{eqnarray*}
&& \!\!(2\gamma n+8n-m')(10^{14}\sigma_{max})^q+\sum_{j=1}^{m_1'}(10^{14}\sigma_{max}+\mu_j)^q+\sum_{j=1}^{m_2'}(10^{14}\sigma_{max}-\nu_j)^q\\
&=& \!\!(2\gamma n+8n) (10^{14}\sigma_{max})^q+\frac{q(q-1)}{2}[\sum_{j=1}^{m_1'}\mu_j^2\!+\!\sum_{j=1}^{m_2'}\nu_j^2] (10^{14}\sigma_{max})^{q-2}+o([\sum_{j=1}^{m_1'}\mu_j^2\!+\!\sum_{j=1}^{m_2'}\nu_j^2]\sigma^{q-2}_{max})\\
&\ge& \!\!(2\gamma n+8n)(10^{14}\sigma_{max})^q+\frac{q(q-1)}{2}[\frac{m_1'}{4}+\frac{m_2'}{4}] (10^{14}\sigma_{max})^{q-2}+o(m'\sigma_{max}^{q-2})\\
&=& \!\!(2\gamma n+8n)(10^{14}\sigma_{max})^q+\frac{q(q-1)m'}{8} (10^{14}\sigma_{max})^{q-2}+o(m'\sigma_{max}^{q-2}).
\end{eqnarray*}
\end{proof}
Combining the above lemmas, Lemma~\ref{lemma:sche-sat} is proved.

\subsection{Finalizing the Proof of Theorem~\ref{thm:lower-bound}}~\label{subsec:formal-proof}
Suppose on the contrary there exists a PTAS for $P||\sum_i C_i^q$ that runs in time $2^{O((1/\varepsilon)^{1/2-\delta})}+n^{O(1)}$, we show that this algorithm can be used to distinguish between instances of 3SAT$'$ with $4n/3$ clauses where at least $4(1-\epsilon')/3$ clauses are satisfiable from instances where at most $4(\beta+\epsilon')n/3$ clauses are satisfiable in time $2^{O(n^{1-\delta})}$, contradicting Lemma~\ref{lemma:maxsat-eth2}.

Consider the constructed scheduling instance with $2\gamma n+8n=O(\frac{n\log n}{\log\log n})$ machines. Recall $\sigma_{max}=n^{1+O(\frac{1}{\log\log n})}$. If the 3SAT$'$ instance has at most $4\epsilon'n/3$ unsatisfied clauses, then by Lemma~\ref{lemma:sat-sche} (taking $\vartheta=4\epsilon'/3$) the objective value $Obj_1$ of the constructed scheduling instance is at most 
\begin{eqnarray*}
Obj_1&\leq&(2\gamma n+8n)(10^{14}\sigma_{max})^q+4\epsilon'n/3 \cdot \frac{q(q-1)}{2}(10^{14}\sigma_{max})^{q-2}+o(n\sigma_{max}^{q-2})\\
&=& (2\gamma n+8n)(10^{14}\sigma_{max})^q+\frac{2\epsilon' q(q-1)n}{3}\cdot (10^{14}\sigma_{max})^{q-2}+o(n\sigma_{max}^{q-2})
\end{eqnarray*}

 If the 3SAT$'$ instance has at least $4(1-\beta-\epsilon')n/3$ unsatisfied clauses, then by Lemma~\ref{lemma:sche-sat} (taking $\vartheta=4(1-\beta-\epsilon')/3$) the objective value $Obj_2$ of any feasible solution for the constructed scheduling instance is at least 
\begin{eqnarray*}
Obj_2&\ge&(2\gamma n+8n) (10^{14}\sigma_{max})^q+\frac{q(q-1)\cdot 4(1-\beta-\epsilon')n/3}{48}\cdot (10^{14}\sigma_{max})^{q-2}+o(n\sigma_{max}^{q-2})\\
&=&(2\gamma n+8n) (10^{14}\sigma_{max})^q+\frac{q(q-1)(1-\beta-\epsilon')n}{36}\cdot (10^{14}\sigma_{max})^{q-2}+o(n\sigma_{max}^{q-2})
\end{eqnarray*}
 for some constant $\beta<1$.

We apply the PTAS for $P||\sum_i C_i^q$ by setting $\varepsilon=\frac{1}{(2\gamma +8) \times (10^{14}\sigma_{max})^2}\cdot \frac{q(q-1)\epsilon'}{36}=\Theta(\gamma^{-1} \sigma_{max}^{-2})=n^{-2-O(\frac{\log\log n}{\log n})}$, then it follows that the PTAS runs in time $2^{O(n^{1-o(1)})}$. 
If there exists a feasible schedule with objective value at most $Obj_1$, then the PTAS returns a solution with objective value at most
$$Obj_1\cdot(1+\varepsilon)\le (2\gamma n+8n) (10^{14}\sigma_{max})^q+\frac{q(q-1)\epsilon'n}{36}\cdot (10^{14}\sigma_{max})^{q-2}+o(n\sigma_{max}^{q-2})<Obj_2.$$

Otherwise, any feasible solution has an objective value of at least $Obj_2$. That is, the PTAS can be used to distinguish between scheduling instances that admit a feasible schedule at most $Obj_1$ and scheduling instances that do not admit any feasible schedule of objective value no more than $Obj_2$, and thus can also be used to distinguish 3SAT$'$ where at least $4(1-\epsilon')n/3$ clauses are satisfiable from instances where at most $4(\beta+\epsilon')n/3$ clauses are satisfiable, contradicting Lemma~\ref{lemma:maxsat-eth2}.

\noindent\textbf{Remark.} It is important to observe that our reduction is only valid when the number of machines $m=\tilde{O}({n})=\tilde{O}(\sqrt{1/\epsilon})$. If $m=O((1/\epsilon)^{\kappa})$ for $\kappa<1/2$, then applying the same reduction we have $\epsilon=\tilde{O}(n^{-1/\kappa})$ by using that $m=\tilde{O}(n)$, whereas we have Corollary~\ref{coro:lower-bound} below. On the other hand, if $m=\Omega((1/\epsilon)^{\kappa})$ for $\kappa>1/2$, then we also have $n=\Omega((1/\epsilon)^{\kappa})$. The objective value is $\Theta(mT^2)=\Omega((1/\epsilon)^{3\kappa})$. Therefore, a PTAS brings an error of $O(mT^2\epsilon)=\Omega(n\cdot(1/\epsilon)^{2\kappa-1})=\Omega(n)$, which is large enough to accommodate the gap of $O(n)$ in the reduction, i.e., the reduction does not work any more.

\begin{corollary}\label{coro:lower-bound}
Let $q>1$ be an arbitrary constant. Assuming ETH, for any $\epsilon$ such that $m=O((1/\epsilon)^{\kappa})$ for some $\kappa\le 1/2$, there is no $(1+\epsilon)$-approximation algorithm for $P||\sum_i C_i^q$ that runs in time $2^{O((1/\varepsilon)^{\kappa-\delta})}+n^{O(1)}$ time for any constant $\delta>0$.
\end{corollary}

\section{Conclusion}
We consider $P||\sum_i C_i^q$, which is identical machine scheduling with the objective of minimizing the $\ell_q$-norm of machine loads for arbitrary constant $q>1$. We establish a PTAS of running time $2^{\tilde{O}(\sqrt{1/\epsilon})}+n^{O(1)}$ and prove that it is essentially the best possible under the exponential time hypothesis. This is the first PTAS that runs in sub-exponential time in $1/\epsilon$ for strongly NP-hard scheduling and other related problems. It is interesting and also important to explore the sub-exponential phenomenon in PTASs for other problems. In particular, it will be interesting to investigate the scheduling problem with the objective of minimizing weighted job completion times, i.e.,  $P||\sum_j w_jC_j$. Another interesting open problem is to show subexponential lower bound or develop FPTAS for $P||\sum_i C_i^q$ where $m=\Theta((1/\epsilon)^\theta)$ for $\theta\in (1/2,1]$.

\bibliography{scheduling}
\bibliographystyle{plainurl}

\end{document}